\documentclass{article}

\usepackage{microtype}
\usepackage{graphicx}
\usepackage{subfigure}
\usepackage{booktabs} 

\usepackage{hyperref}

\usepackage{setspace}

\usepackage[accepted]{icml2024}

\usepackage{amsmath}
\usepackage{amssymb}
\usepackage{mathtools}
\usepackage{amsthm}

\usepackage[capitalize,noabbrev]{cleveref}

\theoremstyle{plain}
\newtheorem{theorem}{Theorem}[section]

\newtheorem{lemma}[theorem]{Lemma}
\newtheorem{corollary}[theorem]{Corollary}
\theoremstyle{definition}
\newtheorem{definition}[theorem]{Definition}
\newtheorem{assumption}[theorem]{Assumption}
\theoremstyle{remark}
\newtheorem{remark}[theorem]{Remark}

\usepackage[textsize=tiny]{todonotes}

\icmltitlerunning{Learning Low-dimensional Latent Dynamics from High-dimensional Observations}

\usepackage{edgeStd}

\newcommand{\kk}{_{k\in[K]}}
\newcommand{\tk}{_{t=0}^{T_k}}

\DeclareMathAlphabet{\mathscr}{U}{dutchcal}{m}{n}
\SetMathAlphabet{\mathscr}{bold}{U}{dutchcal}{b}{n}

\usepackage{bm}
\usepackage{multicol}

\usepackage{acronym}
\newacro{hdsysid}[\texttt{HDSYSID}]{High-dimensional System Identification Problem}
\newacro{metasysid}[\texttt{MetaSYSID}]{High-dimensional Meta System Identification Problem}
\newacro{oracle}[Sys-Oracle]{System Identification Oracle}
\newacro{subspaceID}[Col-Approx]{Column Space Approximation}
\newacro{col-algo}[Col-SYSID]{Column Space Projection SYSID}
\newacro{meta-algo}[Meta-Col-SYSID]{Meta Column Space Projection SYSID}

\begin{document}

\twocolumn[
\icmltitle{Learning Low-dimensional Latent Dynamics from High-dimensional Observations: Non-asymptotics and Lower Bounds}



\icmlsetsymbol{equal}{*}

\begin{icmlauthorlist}
\icmlauthor{Yuyang Zhang}{harvard}
\icmlauthor{Shahriar Talebi}{harvard}
\icmlauthor{Na Li}{harvard}
\end{icmlauthorlist}

\icmlaffiliation{harvard}{SEAS, Harvard University, Cambridge, USA}

\icmlcorrespondingauthor{Yuyang Zhang}{yuyangzhang@g.harvard.edu}

\icmlkeywords{Machine Learning, ICML}

\vskip 0.3in
]



\printAffiliationsAndNotice{} 

\begin{abstract}
    In this paper, we focus on learning a linear time-invariant (LTI) model with low-dimensional latent variables but high-dimensional observations. We provide an algorithm that recovers the high-dimensional features, i.e. column space of the observer, embeds the data into low dimensions and learns the low-dimensional model parameters. Our algorithm enjoys a sample complexity guarantee of order $\tilde{\calO}(n/\epsilon^2)$, where $n$ is the observation dimension. We further establish a fundamental lower bound indicating this complexity bound is optimal up to logarithmic factors and dimension-independent constants.
    We show that this inevitable linear factor of $n$ is due to the learning error of the observer's column space in the presence of high-dimensional noises.
    Extending our results, we consider a meta-learning problem inspired by various real-world applications, where the observer column space can be collectively learned from datasets of multiple LTI systems. An end-to-end algorithm is then proposed, facilitating learning LTI systems from a meta-dataset which breaks the sample complexity lower bound in certain scenarios.
\end{abstract}

\section{Introduction}
Analyzing high-dimensional time series data is essential for numerous real-world applications in finance \cite{mudassir_time-series_2020}, economics \cite{maliar_merging_2015, masini_machine_2023} and biology \cite{churchland_neural_2012, hajnal_continuous_2023, xia_stable_2021, gallego_long-term_2020, stringer_high-dimensional_2019}. 
High-dimensional time series observations often find succinct representation through a set of low-dimensional latent variables. 
In this paper, the focus is to learn the low-dimensional dynamics capturing the very essence of the time series, which becomes useful in various down-stream tasks like prediction and inference~\cite{churchland_neural_2012,mudassir_time-series_2020, pandarinath_inferring_2018}.

Popular techniques for such analysis range from linear models, such as linear time-invariant (LTI) systems~\cite{sikander_linear_2015, bui-thanh_model_2008, hespanha_linear_2018} and auto-regressive models~\cite{dong_extracting_2022, poloni_closed-form_2019, qin_latent_2022}, to more complex nonlinear models, exemplified by recurrent neural networks~\cite{yu_analysis_2021,sussillo_opening_2013,medsker_recurrent_1999}. Herein, we focus on LTI models as they are interpretable and require much less computational power in many real-world applications~\cite{gallego_long-term_2020, churchland_neural_2012, dong_extracting_2022}.
%

Specifically, we learn partially observed LTI systems in the following form
\begin{equation}\label{eq:sing_sys}
    x_{t+1} = Ax_t + Bu_t + w_t, \quad y_t = Cx_t + \eta_t,
\end{equation}
where $y_t\in\bbR^n$ are  \emph{high-dimensional} observations, $x_t\in\bbR^r$ are low-dimensional latent variables with $n\gg r$, $u_t\in\bbR^m$ are inputs and $w_t, \eta_t$ are process and observation noises, respectively. Matrices $A$ and $B$ are the system parameters for the dynamics and inputs, respectively, and $C$ can be seen as the observer, mapping latent variables to \emph{noisy observations}. 
Recently, a contemporary perspective has revitalized well-established system identification (SYSID) algorithms for similar setups \cite{shirani_faradonbeh_finite_2018,sarkar_near_2019,oymak_non-asymptotic_2019,sarkar_finite_2021,zheng_non-asymptotic_2021}. Without any structural assumptions, they have established finite time convergence results with a rate of $\tilde{\calO}(\sqrt{n\cdot\poly\b{r,m}/T})$, leading to a sample complexity of $\tilde{\calO}(n\cdot\poly\b{r,m}/\epsilon^2)$ for an $\epsilon$-well approximation of model parameters.
Such results are not satisfying for LTI systems with a large observation dimension $n$.
One alternative way to directly applying the SYSID algorithm is to initially project the high-dimensional data to lower dimensions before conducting further analysis \cite{saul_introduction_2001,churchland_neural_2012}, either for efficiency or for interpretability. Moreover, these projections may often consist of ``meta information'' that is shared across different tasks. This is consistent with many scenarios in meta-learning settings where we deal with more complex datasets collected from similar observers but possibly with different latent dynamics. 
The above facts motivate us to investigate the provable efficiency of this type of projection method and its generalization ability. 

\textbf{Contributions.} 
In this paper, we study learning LTI systems with \emph{high-dimensional noisy observations}. We adopt a two-stage procedure to first extract high-dimensional features, i.e. column space of the observer. We subsequently perform standard SYSID on the resulting low-dim data and recover the original model parameters. We establish sample complexity for this bifold procedure that scales as 
$\tilde{\calO}([n + \poly\b{r,m}]/\epsilon^2)$ which further reduces to only $\tilde{\calO}(\poly\b{r,m}/\epsilon^2)$ in the absence of observation noises. This then naturally leads to the following question:

\noindent\emph{Can this linear dependence on observer dimension $n$ be possibly improved by some carefully designed algorithm?}
 
We show that the linear dependence on $n$ is unavoidable in the presence of observation noises --- as shown in our lower bound result \Cref{thm:1lowerbound}, requiring at least $\calO(n/\epsilon^2)$ samples from our learning problem.
Unfortunately, observation noises always exist for real-world applications, as the sensors always involve uncertainties. The lower bound also indicates that the sample complexity of the proposed algorithm is optimal up to some logarithmic factors and dimension-independent constants.

Additionally, our pragmatic solution--namely, the separation of learning the high-dimensional features (i.e. the observer column space) and learning the rest of the model parameters-- 
can be extended to meta-learning setups by collecting metadata from a set of dynamical systems that share the same observer model. 
By adopting a ``leave-one-out'' strategy for statistical consistency, we show that as long as the metadata is \emph{collectively} large in order of $\tilde{\calO}(n/\epsilon^2)$, we can successfully obtain an $\epsilon$-well approximation of all the systems parameters involved in the meta-data (\Cref{thm:2meta}).  We finally note that such metadata with the same observer is common in real-world applications \cite{hajnal_continuous_2023, xia_stable_2021, marks_stimulus-dependent_2021, gallego_long-term_2020}. One example is in neuroscience where neuron activities are measured over a long period by the same set of electrodes or imaging devices. Although the subject may perform different tasks or demonstrate behavior corresponding to different latent dynamics, the observer model (i.e. the electrodes or the imaging device) remains the same, resulting in the metadata considered here.

\textit{To summarize our contributions:} Firstly, we provide \acf{col-algo} (\Cref{alg:1single} in \Cref{sec:3single}) that learns LTI systems through high-dimensional observations with complexity $\tilde{\calO}(n/\epsilon^2)$ (\Cref{thm:1single} in \Cref{sec:3single}). As compared with existing algorithms, ours reduces the multiplicative polynomial dependency on the latent dimension $r$ and the input dimension $m$. We then establish a sample complexity lower bound for this problem, indicating the optimality of the above result (\Cref{thm:1lowerbound} in \Cref{sec:3lower}). 
Our algorithmic idea is further extended to a meta-learning setting, where we have access to datasets from multiple similar systems.
We provide an end-to-end framework for handling this meta-dataset and learning all included systems (\Cref{alg:2meta} in \Cref{sec:4meta}). With the help of the meta-dataset, we break the sample complexity lower bound in certain scenarios (\Cref{thm:2meta} in \Cref{sec:4meta}).

Before proceeding, we set the following notations throughout the paper.




\textit{Notations:} Without further explanation, let $\delta$ be any probability in $(0,1/e)$. 
We use $\poly\b{\cdot}$ to denote polynomial and logarithmic dependences. We use $\lesssim$ and $\gtrsim$ to hide absolute constants. We use $\tilde{\calO}(\cdot)$ to hide all problem-related constants, absolute constants and logarithmic factors. Let $[N]$ be the set of integers $\{1, \dots, N\}$. Let $\calM_{[N]}$ denote the set $\{\calM_n\}_{n\in[N]}$ and let $\calM_{-n}$ denote the set $\{\calM_{n'}\}_{n'\in[N]\setminus \{n\}}$ when the full set $[N]$ is clear from the context. For any matrix $M\in\bbR^{m\times n}$, let $\sigma_1(M) \geq \sigma_2(M) \geq \cdots \geq \sigma_{\max\{m,n\}}(M)$ be its singular values, and let $\sigma_{\min}(M)$ be the \textit{minimum non-zero} singular value. Let $\norm{M}$ denote its operator norm, and $M\t$ denote its transpose. Let $\col(M)$ be the column space of $M$ and let $\Phi_M$ denote any orthonormal matrix whose columns form a basis of $\col(M)$. For any orthonormal matrix $\Phi$, we use $\Phi^\perp$ to denote the matrix such that $\begin{bmatrix}
    \Phi & \Phi^\perp
\end{bmatrix}$ is unitary, and we refer to $\Phi^\perp$ as the orthogonal complement of $\Phi$. 
For any \textit{positive semi-definite matrix} $\Sigma\in\bbR^{n\times n}$, we slightly abuse the notation and let $\calN(0,\Sigma)$ denote the distribution of $\Sigma^{1/2}x$, where $x\sim\calN(0,I_n)$ follows the standard Gaussian distribution.


\subsection{Other Related Works}
\paragraph{Other Linear Models with Low-Dimensional Dynamics. } 
Recently, there are lines of research on autoregressive models with low-dimensional representations \cite{qin_latent_2022, dong_extracting_2022, poloni_closed-form_2019}. Indeed, we may also be able to reconstruct low-dimensional dynamics from the model with the learned representations. However, sample complexity results are not currently available for such models. Another related line of research is on dynamic factor models \cite{hallin_factor_2023, anderson_linear_2022, breitung_dynamic_2006}. \citet{anderson_linear_2022} only provides asymptotic convergence guarantees. Finite time convergence is indeed developed in \citet{hallin_factor_2023}. However, their results can not imply dependence on the system dimensions and are only about the convergence rates. Moreover, certain assumptions might not be easily verified in our setting. 

\paragraph{General SYSID Algorithms. } The most relevant sample complexity results are provided in recent SYSID literature for learning partially observed LTI systems  \cite{zheng_non-asymptotic_2021,lee_improved_2022,djehiche_efficient_2022,sarkar_finite_2021,oymak_non-asymptotic_2019}. The last four papers focus on systems with stable latent dynamics, while the first paper extends to unstable latent dynamics. Unfortunately, all algorithms have suboptimal sample complexities. 
There also exists a line of work providing sample complexity for learning fully observed systems \cite{sarkar_near_2019, shirani_faradonbeh_finite_2018}, whose analysis techniques is applicable to unstable dynamics. 
This paper focuses on stable latent dynamics, and we leave for future directions to extend the results to unstable latent dynamics. 

\paragraph{Existing SYSID Lower Bounds. } There exists literature providing lower bounds for learning partially observed LTI systems. However, their results are not directly comparable to ours in our setting. The most relevant papers are \citet{sun_finite_2023, fattahi_learning_2021}. The former provides a lower bound depending on $r/n$, which is tailored for systems with $r \gg n$. The latter provides a lower bound proportional to $1/\epsilon^4$. The lower bound follows a slower rate than ours, though being logarithmic in the system dimensions. \citet{sun_finite_2022} provides lower bounds defined for noise-free settings. \citet{mao_finite-time_2021} provides lower bounds for systems with $C = I$. \citet{bakshi_new_2023} provides lower bounds for learning almost uncontrollable and unobservable systems. Other related works include \citet{djehiche_non_2021, jedra_finite-time_2023, simchowitz_learning_2018}, which develop lower bounds for fully observed systems. 

\paragraph{General Subspace Learning Algorithms.} To learn the observer column space, several existing literature may be related. \citet{vaswani_robust_2018,balzano_streaming_2018, candes_robust_2011,  blanchard_statistical_2007} focus on analysis for PCA subspace learning. However, they assume that the dataset is i.i.d. sampled, which is not the case for dynamic systems. Dynamic factor model techniques are also related \cite{hallin_factor_2023}. As discussed previously, the results can not imply the dependence on the system dimensions. Other ideas include \citet{deng_invariant_2020}, whose sample complexity remains an open question.

\vspace{-0.2cm}
\section{Preliminaries and The Problem Setup}\label{sec:2prelim}

Consider linear dynamical systems in the form of \Cref{eq:sing_sys} with $x_0 = 0$ (for simplicity), process noises $w_t\overset{i.i.d}{\sim}\calN(0, \Sigma_w)$ and isotropic observation noises $\eta_t\overset{i.i.d}{\sim} \calN(0, \sigma_\eta^2 I)$. Here $\Sigma_w$ and $\sigma_\eta^2 I$ are \textit{positive semi-definite matrices.}
We denote such systems by $\calM = (r,n,m,A,B,C,\Sigma_w,\sigma_\eta^2 I)$. 
In standard system identification setup, 
given the input-output trajectory data from the system,
the objective is to learn system parameters up to the well-known similarity transformation class, i.e., to learn tuple $(\htA, \htB, \htC)$ such that for some invertible matrix $S$,\footnote{Here, $S$ represents a change of basis for the latent variables resulting in a modification of the system representations from $(A,B,C)$ to $(S^{-1}A S, S^{-1} B, C S)$--yet describing exactly the same input-output behavior. See Section 4.4 in \cite{hespanha_linear_2018}.}
\begin{equation*}\begin{split}
    \htA = S^{-1}A S, \quad \htB = S^{-1}B, \quad \htC = C S.
\end{split}\end{equation*}
 Furthermore, in order to ensure this learning problem is well-posed, we assume $(A,B)$ is controllable and $(A,C)$ is observable. This is often referred to as a \textit{minimal realization} (a state-space description of minimal size that explains the given input-output data \cite{schutter_minimal_2000}). The corresponding system $\calM$ is then called a minimal system. 

\paragraph{\acf{hdsysid}:} Consider learning minimal system $\calM = (r,n,m,A,B,C,\Sigma_w,\sigma_\eta^2I)$ with \textit{high-dimensional observations} and \textit{low-dimensional latent states and inputs}; namely, $r,m \ll n$. Here, covariances $\Sigma_w$, $\sigma_\eta^2 I$ are \textit{positive semi-definite covariances}. For $k=1, 2$, we choose input trajectories $\calU_k = \{u_{k,t}\}_{t=0}^{T_k-1}$ and get corresponding observations $\calY_k = \{y_{k,t}\}_{t=0}^{T_k}$. 
Here, every input $u_{k,t}\overset{\iid}{\sim}\calN(0, \Sigma_u)$ is sampled independently with \textit{positive definite covariance $\Sigma_u\succ 0$}.\footnote{Our approach can be easily extended to a weaker assumption on $\Sigma_u$: $(A,(\Sigma_w+B\Sigma_uB\t)^{1/2})$ being controllable. Here we assume $\Sigma_u\succ 0$ to keep the results clean and interpretable.}
With the two datasets $\calD_1 = \calY_1\cup\calU_1$ and $\calD_2 = \calY_2\cup\calU_2$, our objective is to output approximate system matrices $\b{\htA, \htB, \htC}$ such that with high probability,
\begin{equation*}\begin{split}\label{eq:objective}
    {}& \max\left\{\norm{S^{-1}AS - \htA}, \norm{S^{-1}B - \htB}, \norm{CS - \htC}\right\}
    \leq \epsilon
\end{split}\end{equation*}
for some invertible matrix $S$.

In \ac{hdsysid} problem, we specifically consider systems with $r,m\ll n$ and isotropic observation noises. 
We select inputs with positive definite covariance $\Sigma_u$ so that the system is fully excited and therefore learnable. 
The problem setup can be extended to the case where we have access to $K \geq 2$ independent data trajectories. Our proposed approach and theoretical analysis in the following sections can also be readily adapted. 





\section{SYSID with Column Space Projection}
\label{sec:3single}
Before diving into the details, we first state the necessary assumption and definition for \ac{hdsysid}.
\begin{assumption}\label{assmp:sys1}
    There exist constants $\psi_A\geq 1$ and $\rho_A \in (0,1)$, which are independent of system dimensions, such that
    \begin{equation*}\begin{split}
     \norm{A^i} \leq \psi_A\b{\rho_A}^{i-1}, \quad\forall i\in\bbN.
    \end{split}\end{equation*}
\end{assumption}


The existence of $\psi_A$ and $\rho_A$ is a standard assumption \cite{oymak_non-asymptotic_2019} and is guaranteed by Lemma 6 in \citet{talebi_data-driven_2023} as long as the spectral radius of $A$ is smaller than $1$. However, the constants may be dependent on system dimensions in the literature.
Here we assume $\psi_A$ and $\rho_A$ to be independent of system dimensions.

%


\begin{definition}[\ac{oracle}]\label{def:idoracle}
   Consider system $\calM=(r,n,m,A,B,C,\Sigma_w,\Sigma_\eta)$ with a minimal realization $(A,B,C)$ and arbitrary dimensions $(r,m,n)$. Assume $\calM$ satisfies Assumption \ref{assmp:sys1}. Given an input-output trajectory $\calD=\calY\cup\calU$, where $\calU = \{u_t\}_{t=0}^{T-1}$ are the inputs with $u_{t}\overset{\iid}{\sim}\calN(0, \Sigma_u)$ and $\calY = \{y_t\}_{t=0}^{T-1}$ are the corresponding outputs, the \ac{oracle} outputs approximation $\b{\htA, \htB, \htC}$ such that with probability at least $1-\delta$,
    \begin{equation*}\begin{split}
        {}& \max\left\{\norm{S^{-1}AS - \htA}, \norm{S^{-1}B - \htB}, \norm{CS - \htC}\right\}\\
        {}&\quad \leq \frac{\poly\b{r,n,m}}{\sqrt{T}}\cdot \kappa_2,
    \end{split}\end{equation*}
    for any $T \geq \kappa_1\cdot \poly\b{r,n,m}$ and some invertible matrix $S$. Here $\kappa_1=\kappa_1\b{\calM,\calU,\delta}, \kappa_2=\kappa_2\b{\calM,\calU,\delta}$ are problem-related constants independent of system dimensions modulo logarithmic factors.
\end{definition}

The $\ac{oracle}$ defined above is applicable to all minimal systems with arbitrary dimensions $(r,n,m)$, which does not necessarily lie in the regime where $n \gg r, m$. 
It represents standard system identification algorithms, including the celebrated Ho-Kalman algorithm \cite{sarkar_finite_2021, oymak_non-asymptotic_2019}. The constraints of such existing algorithms are captured by Assumption \ref{assmp:sys1}. These algorithms also require the inputs and noises to be sampled independently and fully exciting, which is captured by the definition of system $\calM$ and inputs $\calU$ in \ac{oracle}. In \Cref{sec:idoracle}, we will provide an example of such an oracle for the completeness of this paper. 

\subsection{The Proposed Algorithm}
\begin{algorithm}[ht]
    \caption{\acf{col-algo}}
    \label{alg:1single}
 \begin{algorithmic}[1]
 \small
 \setstretch{0.5}
    \STATE {\bfseries Inputs:} Data $\calY_1, \calD_2$;
    Subroutines Col-Approx, \ac{oracle}; \\
    
    \STATE Approximate the observer column space:
    \[\wh{\Phi}_C \gets \text{Col-Approx}\b{\calY_1}\]
    \STATE Project dataset $\calD_2 = \{y_{2,t}\}_{t=0}^{T_2}\cup\{u_{2,t}\}_{t=0}^{T_2-1}$ onto the column space:
    \[\tilde{\calD}_2 \gets \big\{\wh{\Phi}_C\t y_{2,t}\big\}_{t=0}^{T_2} \cup \left\{u_{2,t}\right\}_{t=0}^{T_2-1}\]
    \STATE Identify low-dimensional parameters:
    \[\htA, \htB, \tilC \gets \ac{oracle}\b{\tilde{\calD}_2}\]
    \STATE Recover the high-dimensional observer:
    \[\htC \gets \wh{\Phi}_C\tilC\]

    \STATE {\bfseries Outputs:} $\b{\htA, \htB, \htC}$
 \end{algorithmic}
\end{algorithm}

Our proposed \Cref{alg:1single} consists of two components. Firstly in Line 2, using the first data trajectory, we approximate the column space of the high-dimensional observer $C$ and get matrix $\wh{\Phi}_C\in\bbR^{n\times \rank(C)}$. The columns of $\wh{\Phi}_C$ form an approximate orthonormal basis of this column space. This approximation is accomplished by the \ac{subspaceID} subroutine, which essentially calculates the covariance of the observations $\{y_{1,t}\}_{t=0}^{T_1}$ and extracts the eigenspace associated with the large eigenvalues. The details of the subroutine is postponed to \Cref{sec:proof-of-1single}.

Secondly, in Lines 3-5, we learn the system parameters with the second data trajectory. With the help of $\wh{\Phi}_C$, we can project the high-dimensional observations onto lower dimensions, i.e. we let $\tily_{2,t} = \wh{\Phi}_C\t y_{2,t}$. This projected sequence of observations satisfies the following dynamics
\begin{equation*}\begin{split}
    x_{t+1} = {}& Ax_t + Bu_t + w_t,\\
    \tily_{2,t} = {}& \wh{\Phi}_C\t Cx_t + \wh{\Phi}_C\t \eta_t.
\end{split}\end{equation*}
that is generated by an equivalent system $\wh{\calM}$ denoted by  
\begin{equation}\label{eq:mhat}
    \wh{\calM} = (r,\rank(\wh{\Phi}_C),m,A,B,\wh{\Phi}_C\t C, \Sigma_w, \sigma_\eta^2I).
\end{equation}

\begin{remark}[Why two trajectories?] The reason for us to switch to the second trajectory for parameter learning is to \textit{ensure the independence between $\wh{\Phi}_C$ and the second trajectory,} which will ensure that $\wh{\Phi}_C\t \eta_t$ is still independent of other system variables. This is critical for the application of \ac{oracle}. Note that if the algorithm is equipped with any \ac{oracle} that handles dependent noises, one can easily simplify \Cref{alg:1single} to a single trajectory. Developing such \ac{oracle} is an interesting future direction and ideas from \cite{tian_toward_2023, simchowitz_learning_2018} may be related.\qed
\end{remark}

Lastly, we then feed this low-dimensional dataset to the \ac{oracle} for learning the corresponding low-dimensional parameters $\b{\htA, \htB, \tilC} \approx \b{A, B, \wh{\Phi}_C\t C}$. Finally, we recover $C$ from $\tilC$ through $\wh{\Phi}_C$. If $\wh{\Phi}_C$ approximates the column space of $C$ well, then $\wh{\Phi}_C\wh{\Phi}_C\t$ should also be close to $\Phi_C\Phi_C\t$, which is the projection onto the column space of $C$. Therefore, it is intuitive to expect that $\wh{\Phi}_C\wh{\Phi}_C\t C \approx \Phi_C\Phi_C\t C = C$.

\subsection{Sample Complexity of the Algorithm}
We now provide the following sample complexity result for \Cref{alg:1single}.
\begin{theorem}\label{thm:1single}
    Consider system $\calM$ and datasets $\calD_1=\calU_1\cup\calY_1, \calD_2=\calU_2\cup\calY_2$ (with lengths $T_1$, $T_2$ respectively) in \ac{hdsysid}. Suppose $\calM$ satisfies Assumption \ref{assmp:sys1}.
    If 
    \begin{equation*}\begin{split}
        T_1 \gtrsim \kappa_3\cdot n^2r^3, \quad T_2 \geq \kappa_1\cdot \poly\b{r,m}, \\ 
    \end{split}\end{equation*}
    then ($\htA,\htB,\htC$) from \Cref{alg:1single} satisfy the following for some invertible matrix $S$ with probability at least $1-\delta$
    \begin{equation*}\begin{split}
        {}& \max\left\{\norm{S^{-1}AS - \htA}, \norm{S^{-1}B - \htB}, \norm{CS - \htC}\right\}\\
        {}& \lesssim \kappa_4 \cdot \sqrt{\frac{n}{T_1}}\norm{\htC} + \kappa_2\cdot \sqrt{\frac{\poly\b{r,m}}{T_2}}.
    \end{split}\end{equation*}
    Here $\kappa_1=\kappa_1(\wh{\calM}, \calU_2,\delta), \kappa_2=\kappa_2(\wh{\calM}, \calU_2,\delta),
    \kappa_3 = \kappa_3\b{\calM,\calU_{[2]},\delta}$ and $\kappa_4 = \kappa_4\b{\calM,\calU_{1},\delta}$ are all problem-related constants independent of system dimensions modulo logarithmic factors with $\wh{\calM}$ defined in \Cref{eq:mhat}. Also, $\kappa_1$ and $\kappa_2$ are defined in \Cref{def:idoracle}, while the definitions of $\kappa_3, \kappa_4$ are summarized in \Cref{thm:1single_full} in the appendix.
\end{theorem}

The above error bound consists of two terms. The first term, also the dominating term, is the error of observer column space approximation and the second term is due to learning of the rest of the system. Together, the error bound directly translates to a sample complexity of $\tilde{\calO}([n+\poly\b{r,m}]/\epsilon^2)$. This complexity is equivalent to $\tilde{\calO}(n/\epsilon^2)$ in our setting with $r,m\ll n$. 

\begin{remark}
    Here the first term indicates an interesting insight: \textit{the norm of our approximated observation matrix may affect the performance of the algorithm}. 
    Therefore, in practice, one might want to choose \ac{oracle} that outputs $\tilC$ with a reasonable norm (line 4 of \Cref{alg:1single}). This ensures $\norm{\htC}$ is not too large because 
    $\norm{\htC} \leq \norm{\wh{\Phi}_C}\norm{\tilC} = \norm{\tilC}.$ \qed
\end{remark}


For more concrete results, we instantiate \Cref{alg:1single} with Ho-Kalman algorithm \cite{sarkar_finite_2021, oymak_non-asymptotic_2019},  \Cref{alg:hokalman} in \Cref{sec:idoracle} as \ac{oracle}. Then the above theorem leads to the following corollary.
\begin{corollary}\label{cor:e2e}
    Consider the special minimal system $\calM = (r, n, m, A, B, C, \sigma_w^2I, \sigma_\eta^2I)$ and datasets $\calD_1=\calU_1\cup\calY_1,\calD_2=\calU_2\cup\calY_2$ (with length $T_1, T_2$ respectively), where the inputs are sampled independently from $\calN(0,\sigma_u^2I)$. Suppose $\calM$ satisfies Assumption \ref{assmp:sys1}.
    If
    \begin{equation*}\begin{split}
        &T_1 \gtrsim \tilde{\kappa}_3\cdot n^2r^3, \quad T_2 \geq \tilde{\kappa}_1\cdot r^3(r+m),
    \end{split}\end{equation*}
    then $(\htA,\htB,\htC)$ from Algorithm \ref{alg:1single} satisfy the following for some invertible matrix $S$ with probability at least $1-\delta$
    \begin{equation*}\begin{split}
        {}& \max\left\{\norm{S^{-1}AS - \htA}, \norm{S^{-1}B - \htB}, \norm{CS - \htC}\right\}\\
        {}& \lesssim \tilde{\kappa}_4 \cdot \sqrt{\frac{n}{T_1}}\norm{\htC} + \tilde{\kappa}_2\cdot\sqrt{\frac{r^5(r+m)}{T_2}}.
    \end{split}\end{equation*}
    Here $\tilde{\kappa}_1, \tilde{\kappa}_2, \tilde{\kappa}_3, \tilde{\kappa}_4$ are all problem-related constants only dependent of system $\calM$ and inputs $\calU_1\cup\calU_2$ (as compared to $\kappa_1$ and $\kappa_2$ in \Cref{thm:1single} that depend on $\wh{\calM}$). They are also independent of system dimensions modulo log factors. Detailed definitions of the constants are listed in \Cref{cor:e2e_full}.
\end{corollary}

Based on the above corollary, we now compare the error of our algorithm instantiated with the example oracle, denoted by $\Delta_1$, with the error of the example oracle when it is directly applied to the dataset, denoted by $\Delta_2$. 
For a fair comparison, we consider $\Delta_2$ as the error with datasets generated by inputs $\calU=\{u_t\}_{t=0}^{T_1+T_2-1}$ of length $T_1+T_2$. Here the inputs are sampled independently from $\calN(0, \sigma_u^2I)$. 

The error of our algorithm is upper bounded by
\begin{equation*}
    \Delta_1 \leq \tilde{\calO}\b{\sqrt{\frac{n}{T_1}} + \sqrt{\frac{r^5(r+m)}{T_2}}}.
\end{equation*}
The above error directly translates to the sample complexity of $\tilde{\calO}([n+r^5(r+m)]/\epsilon^2) = \tilde{\calO}(n/\epsilon^2)$ when $r,m\ll n$. 
On the other hand, from \Cref{thm:9hokalman}
\begin{equation*}
    \Delta_2 \leq \tilde{\calO}\b{\sqrt{\frac{r^5(r+n+m)}{T_1+T_2}}}.
\end{equation*}
And this gives a sample complexity of order $\tilde{\calO}([nr^5+r^5(r+m)]/\epsilon^2) = \tilde{\calO}(nr^5/\epsilon^2)$ when $r,m\ll n$. 
Therefore, in the regime where $r,m\ll n$, it is clear that our algorithm enjoys a better sample complexity as compared to directly applying the example oracle. 


\subsection{Proof Ideas of \Cref{thm:1single}}\label{sec:proof-of-1single}
Although conceptually simple, justifying this tight sample complexity result is hard. The difficulty lies in ensuring the accuracy of the observer column space approximation $\wh{\Phi}_C$, which is learned mainly by PCA on the given data (detailed in the following algorithm). 
\begin{algorithm}[H]
    \caption{Column Space Approximation (Col-Approx)}
    \label{alg:column_space}
 \begin{algorithmic}[1]
    \small
    \setstretch{0.5}
    \STATE {\bfseries Input:} Observations $\calY=\{y_{t}\}_{t=0}^{T}$;\\
    \STATE Calculating the data covariance
    \begin{equation*}
        \Sigma_y \gets \sum_{t=0}^T y_{t}y_t\t
    \end{equation*}
    \STATE Estimating the observer rank
    \begin{equation*}
        \htr_c \gets \mathop{\arg\max}_i \b{\sigma_i(\Sigma_y) - \sigma_{i+1}(\Sigma_y) > T^{3/4}}
    \end{equation*}

    \STATE Estimating the column space
    \begin{equation*}
        \wh{\Phi}_C \gets \text{ first $~\htr_c~$ eigenvectors of } ~\Sigma_y
    \end{equation*}
    \STATE {\bfseries Output:} $\wh{\Phi}_C$
 \end{algorithmic}
 \end{algorithm}
Existing analysis of PCA \cite{candes_robust_2011, vaswani_robust_2018, chen_consistent_2015}, matrix factorization \cite{gribonval_sample_2015} and other subspace learning techniques \cite{tripuraneni_provable_2021,zhang_meta-learning_2023} do not apply to our setting. This is because these methods assume i.i.d. data samples, while our dataset $\calD_1$ consists of correlated data points generated by a dynamical system. We need to apply martingale tools \cite{abbasi-yadkori_improved_2011,sarkar_near_2019} for our analysis. 

Moreover, the noises in the dynamical system accumulate across time steps. Naively applying the martingale tools leads to loose sample complexity bounds. Since we are learning a subspace embedded in $\bbR^n$, it is reasonable to expect a sample complexity of $\calO(n/\epsilon^2)$ for an $\epsilon$-good approximation of the subspace
\footnote{Other intuitions for this complexity comes from standard covariance concentration results for i.i.d. data (Corollary 2.1 of \cite{rudelson_random_1999}). In these results, $\calO(n/\epsilon^2)$ samples are required to learn an $\epsilon$-good covariance, which directly translates to a good column space estimation of the covariance. }. We achieve this result by developing our own subspace perturbation lemma specifically tailored to this setting. 
The above ideas on analyzing \ac{subspaceID} translate to the following lemma.
\begin{lemma}\label{lem:1col}
    Consider system $\calM$ and dataset $\calD_1=\calU_1\cup\calY_1$ in \ac{hdsysid}. Suppose $\calM$ satisfies \Cref{assmp:sys1}.
    If
    \begin{equation*}\begin{split}
        T_1 \gtrsim \kappa_5\cdot n^2r^3,
    \end{split}\end{equation*}
    then $\wh{\Phi}_C=\text{Col-Approx}(\calY_1,\Sigma_\eta)$ satisfies the following with probability at least $1-\delta$
    \begin{equation*}\begin{split}
        \rank (\wh{\Phi}_C) = \rank(C), \quad \norm{\wh{\Phi}_C^\perp\t\Phi_C} \lesssim \kappa_4\cdot \sqrt{\frac{n}{T_1}}.
    \end{split}\end{equation*}
    Here $\kappa_4 = \kappa_4\b{\calM,\calU_1,\delta}$ and $\kappa_5 = \kappa_5\b{\calM,\calU_1,\delta}$ are both problem-related constants independent of system dimensions modulo logarithmic factors (details in \Cref{lem:1col_full}).
\end{lemma}
Recall $\Phi_C$ denotes the orthonormal matrix whose columns form a basis of $\col(C)$ and $\wh{\Phi}_C^\perp$ denote the matrix such that $\begin{bmatrix}
    \wh{\Phi}_C^\perp & \wh{\Phi}_C
\end{bmatrix}$ is unitary. 


\begin{proof}[Proof Sketch of \Cref{lem:1col}]
    $\wh{\Phi}_C$ is constructed from the eigenvectors of matrix $\Sigma_y=\sum_{t=0}^{T_1}y_{1,t}y_{1,t}\t$. We decompose it as follows
    \begin{equation*}\begin{split}
        {}& \sum_{t=0}^{T_1}y_{1,t}y_{1,t}\t = \underbrace{C \b{\sum_{t=0}^{T_1}x_{1,t}x_{1,t}\t} C\t}_{\Sigma_C} + \sigma_\eta^2(T_1+1)I\\
        + {}& \underbrace{\sum_{t=0}^{T_1} \b{\eta_{1,t}\eta_{1,t}\t-\sigma_\eta^2I}}_{\Delta_2} + \underbrace{\sum_{t=0}^{T_1} \b{Cx_{1,t}\eta_{1,t}\t + \eta_{1,t}x_{1,t}\t C\t}}_{\Delta_1+\Delta_1\t}.
    \end{split}\end{equation*}
    Here $\Sigma_C$ contains the information on the observer column space $\Phi_C$, while $\Delta_1$ and $\Delta_2$ are noise terms.
    The rest of the proof is decomposed into three steps. In the first step, we upper bound the noise terms $\Delta_1$, $\Delta_2$ and lower bound the latent state covariance $\sum_{t=0}^{T_1}x_{1,t}x_{1,t}$. In the second step, we show that the eigenvectors of $\Sigma_y$ in the column space of $C$ has eigenvalues much larger then the other ones. This gap eigenvalue gap enables us to identify the dimension of the column space of $C$. Finally, in the third step, we apply our tailored subspace perturbation result to bound the accuracy of the approximated column space.

    \textbf{Step 1.} With system $(A,B)$ being controllable, the Gaussian inputs with $\Sigma_u \succ 0$ can fully excite the latent dynamics and thus lead to the following lower bound on $\sum_{t=0}^{T_1}x_{1,t}x_{1,t}$:
    \begin{equation*}\begin{split}
        \sum_{t=0}^{T_1}x_{1,t}x_{1,t} \succ \tilde{\calO}(T_1)I.
    \end{split}\end{equation*}
    The upper bound on $\Delta_2$, i.e. $\norm{\Delta_2} \leq \tilde{\calO}(\sqrt{T_1})$, follows from standard Guassian concentration arguments. To upper bound $\Delta_1$, we apply the martingale tools from previous works \cite{abbasi-yadkori_improved_2011,sarkar_near_2019} and get the following relative error bound
    \begin{equation}\begin{split}\label{eq:1col_sketch1}
        \norm{(\Sigma_C+T_1I)^{-\frac{1}{2}}\Delta_1} \leq \tilde{\calO}(\sqrt{T_1}).
    \end{split}\end{equation}

    \textbf{Step 2.} Since the perturbations are small, standard eigenvalue perturbation bounds give
    \begin{equation*}\begin{split}
        \sigma_{\rank(C)}(\Sigma_y) - \sigma_{\rank(C)+1}(\Sigma_y) \geq \tilde{\calO}(T_1).
    \end{split}\end{equation*}
    This eigenvalue gap as large as $\tilde{\calO}(T_1)$ makes it easy to determine the rank of $C$. Therefore, with high probability, $\rank(\wh{\Phi}_C) = \rank(C)$. 

    \textbf{Step 3.} Finally, we bound the angle between the eigenspace of the first $\rank(C)$ eigenvalues of $\Sigma_y = \Sigma_C + \Delta_2 + (\Delta_1+\Delta_1\t)$ and that of $\Sigma_C$ --- which coincides with the column space of $C$. Here $\Delta_1$ is the so-called relative perturbation, because we bound its norm by comparing it with the information $(\Sigma_C+T_1I)^{\frac{1}{2}}$ (in \Cref{eq:1col_sketch1}). Since $\norm{(\Sigma_C+T_1I)^{\frac{1}{2}}}$ can be large, directly applying standard subspace perturbation results leads to loose error bound. For this, we develop a specific relative subspace perturbation bound in \Cref{lem:perturbation}, adapted from classic eigenspace perturbation results. It ensures both $\Delta_2$ and the relative noise $\Delta_1$ will not perturb the eigenspace of $\Sigma_y$ too much. This finishes the proof.
\end{proof}

With the above lemma, the proof of \Cref{thm:1single} can then be decomposed into two steps. \textit{In the first step}, we guarantee the accuracy of learning the low-dimensional system parameters $(A,B,\wh{\Phi}_C\t C)$ from the projected dataset (Line 3 in \Cref{alg:1single}). Intuitively, as $\wh{\Phi}_C$ is accurate enough, the projected dataset preserves almost all the information in the original dataset. This also ensures that the equivalent system generating the projected dataset ($\wh{\calM}$ in \Cref{eq:mhat}) is observable and controllable. Therefore, \ac{oracle} outputs an accurate estimation of the low-dimensional system parameters. 
\textit{In the second step}, we analyze the errors of the recovered high-dimensional parameters $(\htA, \htB, \wh{\Phi}_C \tilC)$, where $\tilC$ is our estimation of $\wh{\Phi}_C\t C$. The errors on $\htA$ and $\htB$ are automatically bounded by the definition of \ac{oracle}. The error on the recovered high-dimensional observer $\htC = \wh{\Phi}_C\tilC$ can be decomposed as follows 

\vspace{-15pt}{\small
\begin{equation*}\begin{split}
    {} \norm{CS - \htC}
    \leq {}& \norm{CS - \wh{\Phi}_C\wh{\Phi}_C\t CS} + \norm{\wh{\Phi}_C\wh{\Phi}_C\t CS - \wh{\Phi}_C\tilC}\\
    = {}& \norm{\wh{\Phi}_C^\perp\b{\wh{\Phi}_C^\perp}\t CS} + \norm{\wh{\Phi}_C\t CS - \tilC}.
\end{split}\end{equation*}
}

Here the first term is the column space approximation error and the second term is the error learning the low-dimensional observer $\tilC\approx \wh{\Phi}_C\t C$. Combining the above inequality with the subspace perturbation bound in \Cref{lem:1col} and the \ac{oracle} in \Cref{def:idoracle}, along with some algebraic arguments, completes the proof.

\section{Lower Bounds}\label{sec:3lower}
Now that \Cref{alg:1single} accomplishes \ac{hdsysid} with sample complexity $\tilde{\calO}(n/\epsilon^2)$, one may ask: 
\noindent\emph{Is this linear dependence on the observer dimension $n$ unavoidable?}
In this section, we provide the following theorem to show that this linear dependence is in fact necessary.

While we present this lower bound for learning LTI systems with a single input-output data trajectory, it can be extended to multiple trajectories. The more general version is deferred to the appendix. (\Cref{thm:1lowerbound_full}).
\begin{theorem}\label{thm:1lowerbound}
    Suppose $n \geq m \geq r$ with $n \geq 2$, and choose any positive scalars $\delta \leq \frac{1}{2}$. 
    Consider the class of minimal systems $\calM=(r,n,m,A,B,C,\Sigma_w,\Sigma_\eta)$ with different $A$,$B$,$C$ matrices. All parameters except $A, B, C$ are fixed and known. Moreover, $r<n$ and $\Sigma_\eta$ is positive definite. Let $ \calD = \{y_t\}_{t=0}^T\cup\{u_t\}_{t=0}^{T-1}\cup\text{\{all known parameters\}}$ denote the associated single trajectory dataset. 
    Here the input $u_t$ satisfies: 1). $u_t$ is independently sampled; 2). $\bbE(u_t) = 0$. 
    Consider any estimator $\htf$ mapping $\calD$ to $(\htA(\calD), \htB(\calD),\htC(\calD))\in\bbR^{r\times r}\times \bbR^{r\times m}\times \bbR^{n\times r}$. If 
    \begin{equation*}
        T < \frac{\phi_\eta(1-2\delta)\log1.4}{50\b{\psi_w + \psi_u}}\cdot\frac{n+\log\frac{1}{2\delta}}{\epsilon^2},    
    \end{equation*}
    there exists a system $\calM_0=(r,n,m,A_0,B_0,C_0,\Sigma_w,\Sigma_\eta)$ with dataset $\calD$ such that
    \begin{equation*}\begin{split}
        \bbP {}& 
        \left\{\norm{C_0B_0-\htC\b{\calD}\htB\b{\calD}} \geq \epsilon\right\} \geq \delta.
    \end{split}\end{equation*}
    Here $\bbP$ denotes the distribution of $\calD$ generated by system $\calM_0$. Related constants are defined as follows
    \begin{equation*}\begin{split}
        \phi_\eta &\coloneqq \sigma_{\min}(\Sigma_\eta), \quad \psi_w \coloneqq \sigma_{1}(\Sigma_w), \\\quad \psi_u &\coloneqq \max_{t\in[0,T-1]} \sigma_{1}\b{\bbE(u_tu_t\t)}.    \end{split}\end{equation*}
        
\end{theorem}

\vspace{-1em}
The above theorem indicates that with less than $\calO(n/\epsilon^2)$ samples, any estimator has a constant probability to fail learning the product $CB$. This error lower bound on $\htC(\calD)\htB(\calD)$ can be translated to corresponding errors on
\begin{equation*}\begin{split}
    \max\left\{\norm{S^{-1}B-\htB\b{\calD}}, \norm{CS-\htC\b{\calD}}\right\} \geq \epsilon,
\end{split}\end{equation*}
under mild conditions.
This indicates that the estimator fails to learn either $B$ or $C$ well. (The proof details are provided in Corollary \ref{cor:thm1_1}). 
Another perspective to understand this error on $\htC(\calD)\htB(\calD)$ is as follows. In the celebrated Ho-Kalman algorithm for \ac{hdsysid}, the estimation error of the Hankel operator is lower bounded by the estimation error of $\htC(\calD)\htB(\calD)$. Therefore, our lower bound indicates the failure of estimating the Hankel operator, which prevents any algorithm from giving satisfying estimates.

This results in an $\calO(n/\epsilon^2)$ sample complexity lower bound on \ac{hdsysid}. Therefore, the sample complexity of \Cref{alg:1single}, i.e. $\tilde{\calO}(n/\epsilon^2)$, is \textit{optimal} up to logarithmic and problem-related constants.

We defer the proof of \Cref{thm:1lowerbound} to \Cref{sec:lowerbound}, but we briefly discuss the reason behind this inevitable linear dependence on $n$. Therein, we construct a set of system instances with different observation matrices such that their column spaces are ``relatively close'' but ``quantifiably distinct''. In fact, we take advantage of the fact that in the high dimensional space $\bbR^n$, there are too many distinct subspaces that are close to each other. And therefore, any estimator needs correspondingly many samples, i.e. $\calO(n/\epsilon^2)$, to distinguish the true column space from many other candidates.
This difficulty is indeed due to the high dimensional observation noises. 

In the following section, we introduce a more general problem pertinent to high-dimensional ``meta-datasets.'' While this lower bound still has implications even for this general problem, we show how to confine its implied difficulty on the total length of the meta-dataset and not on each personalized portions.

\section{Meta SYSID}\label{sec:4meta}
Thus far, we have introduced \Cref{alg:1single} for \ac{hdsysid} that provides near-optimal dependence on system dimensions utilizing the idea of low-dimensional embedding. Now we will investigate how this algorithmic idea facilitates learning from the so-called meta-datasets (or metadata)---multiple datasets from different systems sharing the same observer. As shown in \cite{hajnal_continuous_2023, xia_stable_2021, marks_stimulus-dependent_2021, gallego_long-term_2020}, this is a real-world setting of significant importance. 

Following our analysis in \Cref{sec:proof-of-1single}, we observe that the $\tilde{\calO}(n/\epsilon^2)$ sample complexity arises from the error of column space $\wh{\Phi}_C$. 
Modulo this error, the sample complexity reduces vastly to $\tilde{\calO}(\poly\b{r,m}/\epsilon^2)$. This motivates us to again separate the learning of the observer column space $\Phi_C$ from the learning of personalized system parameters. By utilizing multiple datasets from different systems, we can collectively learn the observer column space much more accurately. Subsequently, equipped with this accurate column space approximation, there is hope to learn every single system with much fewer samples. 
To formalize the above idea, we introduce the following ``\acf{metasysid}.'' 

\textbf{\acf{metasysid}:}
Consider the identification of $K$ minimal systems $\calM_k = (r,n,m,A_k,B_k,C,\Sigma_{w,k},\sigma_{\eta,k}^2I), k\in[K]$ with the same dimensions ($r,m\ll n$) and observation matrix $C$. Here covariances $\{\Sigma_{w,k}, \sigma_{\eta,k}^2I\}_{k\in[K]}$ are \textit{positive semi-definite matrices}.
For every system $\calM_k$, we choose an input sequence $\calU_k=\{u_{k,t}\}_{t=0}^{T_k-1}$ and get observations $\calY_k = \{y_{k,t}\}_{t=0}^{T_k}$. Here every input $u_{k,t}\overset{\iid}{\sim}\calN(0, \Sigma_{u,k})$ is sampled independently from both the system variables and other inputs with \textit{positive definite covariance $\Sigma_{u,k}$}. This single trajectory dataset is denoted by $\calD_k=\calY_k\cup\calU_k$.  
With the $K$ datasets $\calD_{[K]} \coloneqq \bigcup_{k\in[K]} \calD_k$, our objective is to learn all system parameters \textit{uniformly well}. Namely, we aim to identify $\{\htA_k,\htB_k, \htC_k\}_{k=1}^K$ such that, for every $k\in[K]$, the following holds for some invertible matrix $S_k$ with high probability 
\begin{equation*}\begin{split}
    \max\left\{\norm{S_k^{-1}A_kS_k - \htA_k}, \norm{S_k^{-1}B_k - \htB_k},\right.\\
    \left.\norm{CS_k - \htC_k} \right\} &\leq \epsilon.
\end{split}\end{equation*}

\subsection{The Meta-Learning Algorithm and Its Sample Complexity}
\begin{algorithm}[h]
    \caption{\acf{meta-algo}}
    \label{alg:2meta}
 \begin{algorithmic}[1]
 \small
 \setstretch{0.5}
    \STATE {\bfseries Inputs:} Meta Datasets $\calD_{[K]}=\calY_{[K]}\cup\calU_{[K]}$;\\ 
    Subroutines \text{Col-Approx}, \ac{oracle}; \\
    
    \FOR{$k\in[K]$}

    
        \STATE Leave one out and approximate observer column space
        \[\wh{\Phi}_{C,k} \gets \text{Col-Approx}\b{\calY_{-k}}\]
        \STATE Project dataset $\calD_k = \{y_{k,t}\}_{t=0}^{T_k}\cup\{u_{k,t}\}_{t=0}^{T_k-1}$ onto the column space:
        \[\tilde{\calD}_k \gets \big\{\wh{\Phi}_{C,k}\t y_{k,t}\big\}_{t=0}^{T_k} \cup \left\{u_{k,t}\right\}_{t=0}^{T_k-1}\]
        \STATE Identify low-dimensional parameters:
        \[\htA_k, \htB_k, \tilC_k \gets \ac{oracle}\b{\tilde{\calD}_k}\]
        \STATE Recover the high-dimensional observer:
        \[\htC_k \gets \wh{\Phi}_{C,k}\tilC_k\]
    \ENDFOR
    
    \STATE {\bfseries Outputs:} $\{\htA_k, \htB_k, \htC_k\}_{k \in [K]}$
 \end{algorithmic}
\end{algorithm}

\Cref{alg:2meta} follows similar steps as \Cref{alg:1single}. For every system $k$, we approximate the column space, learn the low-dimensional system parameters, and finally project them back into the high-dimension. 
And similarly, to ensure the independence between $\wh{\Phi}_{C,k}$ and the projected dataset $\tilde{\calD}_k$, the dataset $\calD_k$ itself is left out in the first step. This independence is critical for applying the \ac{oracle} subroutine. 

Now we provide the theoretical guarantee the proposed algorithm.
\begin{theorem}\label{thm:2meta}
    Consider systems $\calM_{[K]}$ and datasets $\calD_{[K]}=\calU_{[K]}\cup\calY_{[K]}$ in \ac{metasysid}. Suppose $\calM_{[K]}$ satisfy Assumption \ref{assmp:sys1}.
    Then for any fixed system $k_0\in[K]$, 
    if $T_{k_0}$  and $T_{-k_0} \coloneqq \sum_{k\neq k_0} T_k$ satisfy
    \begin{equation*}\begin{split}
        & T_{k_0} \geq \kappa_1\cdot \poly\b{r,m} \;\;\text{ and }\;\; T_{-k_0} \gtrsim \kappa_3\cdot n^2r^3\log^8K,
    \end{split}\end{equation*}
    then $(\htA_{k_0},\htB_{k_0},\htC_{k_0})$ from \Cref{alg:2meta} satisfy the following for some invertible matrix $S$ with probability at least $1-\delta$
    
    \vspace{-15pt}{\small\begin{equation*}\begin{split}
        {}& \max\left\{\norm{S^{-1}A_{k_0}S - \htA_{k_0}}, \norm{S^{-1}B_{k_0} - \htB_{k_0}}, \norm{C_{k_0}S - \htC_{k_0}}\right\}\\
        {}& \lesssim \kappa_2\cdot \sqrt{\frac{\poly\b{r,m}}{T_{k_0}}} + \kappa_4 \cdot \sqrt{\frac{n}{T_{-k_0}}}\norm{\htC_{k_0}}.
    \end{split}\end{equation*}}
    Here $\kappa_1=\kappa_1(\wh{\calM}_{k_0}, \calU_{k_0},\delta), \kappa_2=\kappa_2(\wh{\calM}_{k_0}, \calU_{k_0},\delta), \kappa_3 = \kappa_3\b{\calM_{[K]},\calU_{[K]},\delta}, \kappa_4 = \kappa_4\b{\calM_{-k_0},\calU_{-k_0},\delta}$ are all problem-related constants independent of system dimensions modulo logarithmic factors. $\kappa_1, \kappa_2$ are defined in \Cref{def:idoracle}, while the definitions of $\kappa_3, \kappa_4$ are summarized in \Cref{thm:2meta_full}. $\wh{\calM}_{k_0}$ is defined as follows
    \begin{equation*}\begin{split}
        \wh{\calM}_{k_0} = ({}&r,\rank(\wh{\Phi}_{C,k_0}),m,\\
        {}& A_{k_0},B_{k_0},\wh{\Phi}_{C,k_0}\t C, \Sigma_w, \sigma_{\eta,k}^2I). 
    \end{split}\end{equation*}
\end{theorem}

In the above result, we only require $T_{k_0} = \tilde{\calO}(\poly\b{r,m}/\epsilon^2)$ data points for an $\epsilon$ accurate approximation whenever $T_{-k_0} = \tilde{\calO}(n/\epsilon^2)$. Namely, as long as we have enough metadata, the number of samples needed from each single system is vastly reduced and \textit{is independent of the observer dimension $n$}. In real-world applications when we have a large, i.e. $\tilde{\calO}(n/\epsilon^2)$, dataset from a single system, the above result significantly helps for few-shot learning of other similar systems. Also, whenever we have access to numerous, i.e. $\tilde{\calO}(n)$ similar systems, the number of samples required from every system is also independent of $n$.


\begin{figure*}[ht]
    \begin{minipage}[h]{0.33\linewidth}
       \centering
        \includegraphics[width=1\linewidth]{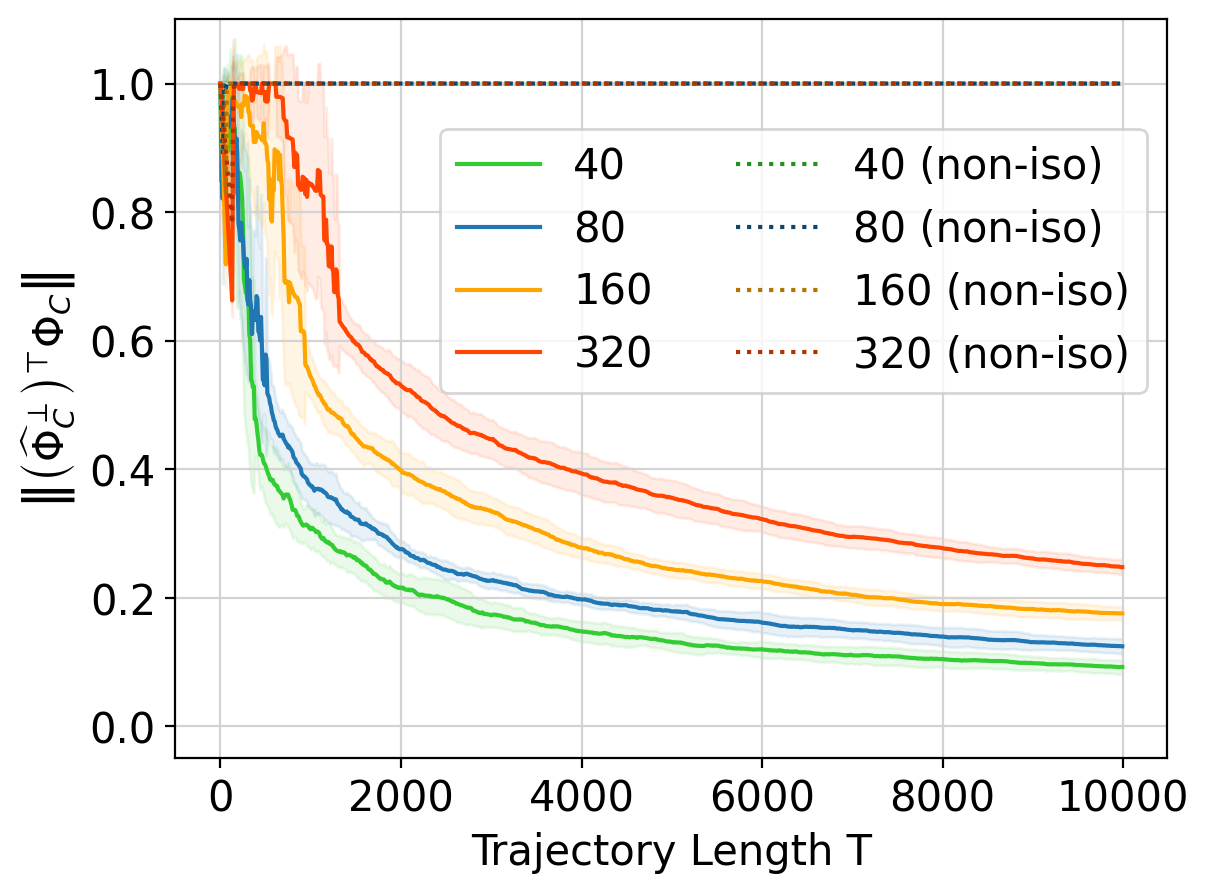}
    \end{minipage}
    \begin{minipage}[h]{0.33\linewidth}
       \centering
        \includegraphics[width=1\linewidth]{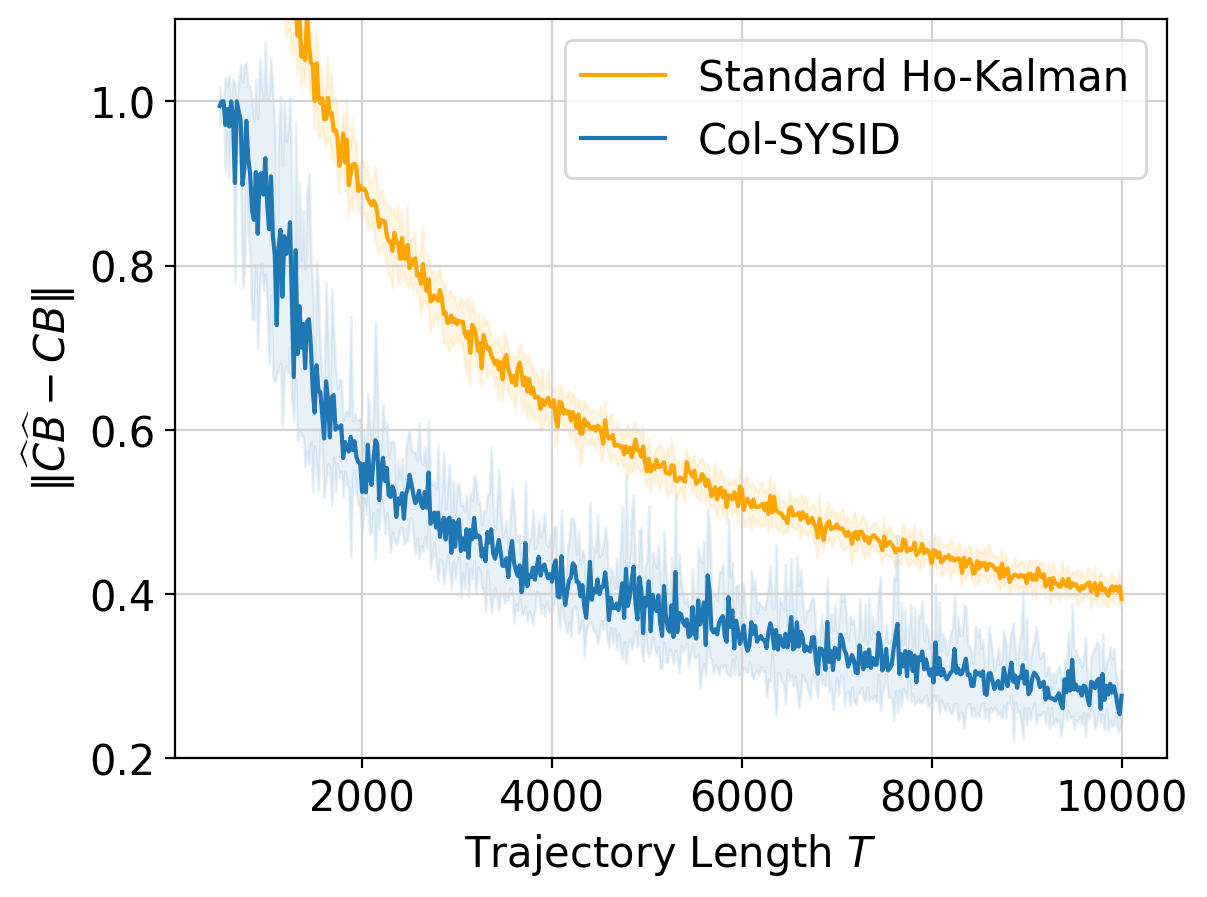}
    \end{minipage}
    \begin{minipage}[h]{0.33\linewidth}
       \centering
        \includegraphics[width=1\linewidth]{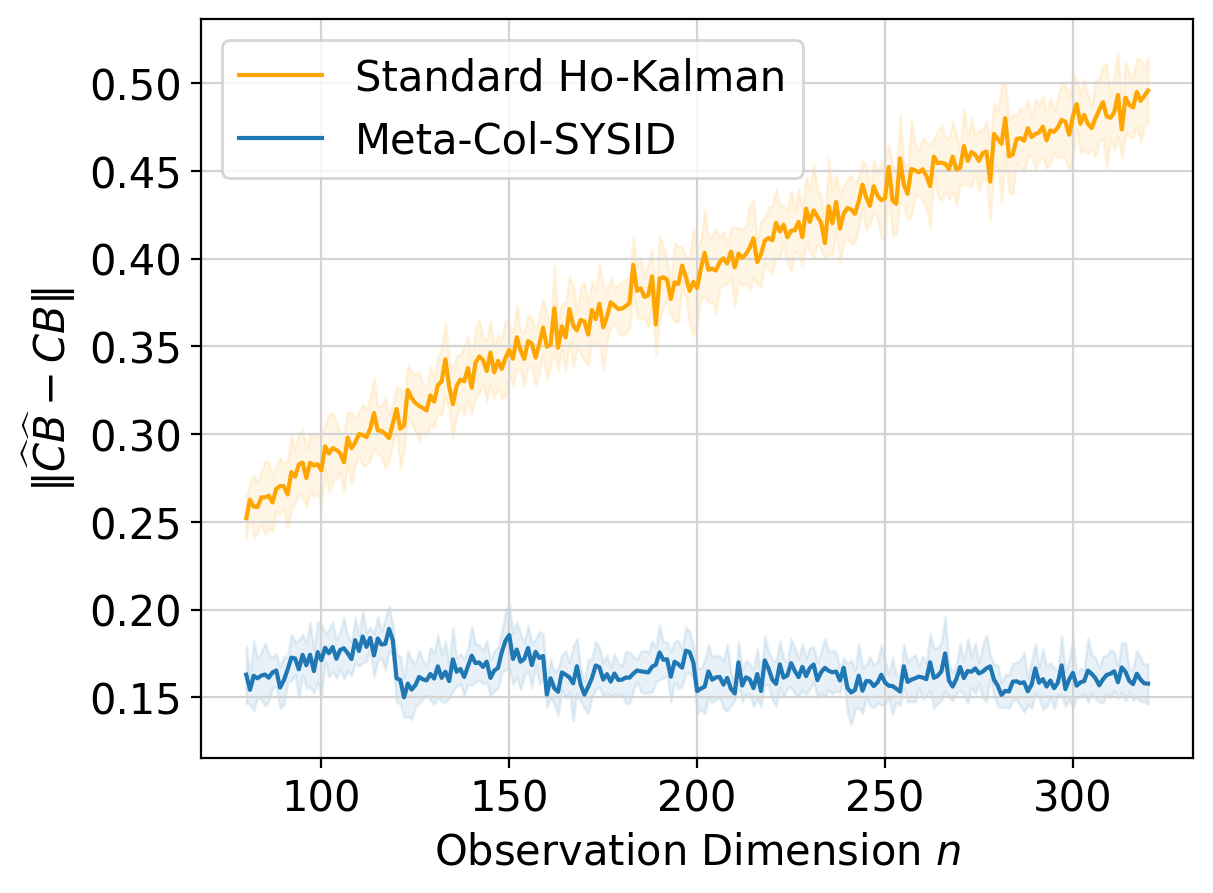}
    \end{minipage}
    \caption{\small {Left}: Error for \text{Col-Approx} (Algorithm \ref{alg:column_space}). {Center}: Error for \ac{col-algo} (Algorithm \ref{alg:1single}) and Standard Ho-Kalman with $n=320$. {Right}: Error for \ac{meta-algo} (Algorithm \ref{alg:2meta}) and Standard Ho-Kalman for $n\in[80,320]$.}
    \label{fig:simulation}
\end{figure*}
\section{Simulations}\label{sec:5sim}

We simulate our algorithm for a set of simple systems where $r=m=1$, $\Sigma_w = 0$ and $\sigma_\eta=1$. We use $A=0.9, B=1$ and randomly sample $C$ with orthonormal columns. We choose inputs with covariance $\Sigma_u=0.1$. The choice of $n$ will be clear in the context. The results are reported in \Cref{fig:simulation}.

We first simulate Col-Approx (Algorithm \ref{alg:column_space}) with $n=40, 80, 160, 320$ and a single trajectory data of length $T=10000$ separately (\Cref{fig:simulation}.Left). Recall from \Cref{lem:1col} that the quality of an estimation $\widehat{\Phi}_C$ of the column space $\Phi_C$ can be quantified by the error term $\|{(\widehat{\Phi}_C^\perp)}^\intercal \Phi_C\|$. For isotropic observation noises as in \ac{hdsysid}, i.e. $\Sigma_\eta = \sigma_\eta^2 I$, the approximated column space converges to the true column space as is shown by reduction in $\|{(\widehat{\Phi}_C^\perp)}^\intercal \Phi_C\|$ for the solid lines. Moreover, the algorithm is simulated for non-isotropic observation noises where the covariance is set to be $\Phi_C\Phi_C\t$. As shown in by the dotted curves, the algorithm fails to converge. This is because non-isotropic noises perturb the eigenspace of the observation covariance $\Sigma_y$ too much so that the information on $\Phi_C$ is drowned. 

Next, we simulate \ac{col-algo} (\Cref{alg:1single}) with $T=10000$ and $n=320$ (\Cref{fig:simulation}.Center). We plot error $\|\htC\htB - CB\|$ during the learning process as a proxy for the estimation error in \ac{hdsysid} problem. This error, i.e. the error of learning the first Markov parameter $CB$, is independent of any similarity transformation, yet relates to estimation errors for $B$ and $C$.  For comparison, the same error for the standard Ho-Kalman \cite{oymak_non-asymptotic_2019} is also included. It is clear that our Algorithm \ref{alg:2meta} outperforms the standard Ho-Kalman. 

Finally, we consider the \ac{metasysid} setting for $n \in [80, 320]$ with a meta-dataset from $\lfloor n/40\rfloor+1$ randomly sampled systems, each with trajectory length $4000$. As the algorithm is treating each system identically, without loss of generality, we only plot the error $\|\htC\htB - CB\|$ for learning the first system. For \ac{meta-algo} (Algorithm \ref{alg:2meta}), it is clear that the average learning error always stays below $0.2$ as the observation dimension increases. This is because our algorithm utilizes information from the meta-dataset to learn the observer column space, which constitutes the major challenge of system learning. However, the standard Ho-Kalman algorithm learns every system independently, and therefore the error grows significantly as the observation dimension increases.

\section{Conclusions and Future Directions}
In conclusion, our focus has been on learning a linear time-invariant (LTI) model characterized by low-dimensional latent variables and high-dimensional observations. The introduced \ac{col-algo} Algorithm serves as a solution with a commendable complexity of $\tilde{\calO}(n/\epsilon^2)$. Our analysis also delves into the fundamental limitations of this problem, establishing a sample complexity lower bound that essentially underscores the optimality of our proposed algorithm. 
Extending the scope of our results, we address a meta-learning setting where datasets from multiple analogous systems are available. This leads to the \ac{meta-algo} algorithm, an end-to-end framework adept at managing the meta-dataset and effectively learning all included systems.

While our current work lays a solid foundation, future directions could explore extensions to non-linear settings, or investigate adaptive approaches to handle varying or non-isotropic observation noises. Additionally, incorporating real-world applications and practical considerations could further enrich the utility of our results.


\section*{Impact Statement}
The goal of this work is to advance theory in high-dimension machine learning, specifically in latent space space learning. There might be minor potential societal consequences, but none of those are considered as necessary to be specifically highlighted here.

\section*{Acknowledgments}
This work is supported by  NSF AI institute 2112085, NSF ECCS 2328241, and NIH R01LM014465.

\bibliographystyle{icml2024}
\bibliography{references}

\appendix 
\label{appendicies}

\onecolumn


\newpage
\section{Upper Bounds for \ac{col-algo} --- Proof of \Cref{thm:1single}}\label{sec:hdsysid}
Recall the setting in \ac{hdsysid}. We consider system $\calM = (r,n,m,A,B,C,\Sigma_w,\sigma_\eta^2I)$ and inputs $\calU_1 = \{u_{1,t}\}_{t=0}^{T_1-1}$, $\calU_2=\{u_{2,t}\}_{t=0}^{T_2-1}$ sampled independently from $\calN(0, \Sigma_u)$. To simplify future analysis, we define the following notations
\begin{equation}\begin{split}
    {}& \psi_C = \sigma_1\b{C},\quad \psi_\eta = \sigma_1(\sigma_\eta^2I), \quad \psi_w = \sigma_1(\Sigma_w + B\Sigma_uB\t),\\
    {}& \phi_C = \sigma_{\min}(C), \quad \phi_O = \sigma_{\min}\b{\begin{bmatrix}
        C\\
        CA\\
        \vdots\\
        CA^{r-1}
    \end{bmatrix}}, \quad \phi_R = \sigma_{\min}\b{\begin{bmatrix}
        B & AB & \dots & A^{r-1}B
    \end{bmatrix}},\quad \phi_u = \sigma_{\min}(\Sigma_u).\\
\end{split}\end{equation}
Here we assume all $\psi$'s satisfy $\psi\geq 1$, otherwise we define $\psi$ to be $\max\{1, \sigma_1(\cdot)\}$. Similarly, we assume all $\phi$'s satisfy $\phi\leq 1$, otherwise we define $\phi$ to be $\min\{1, \sigma_{\min}(\cdot)\}$.

Recall that auxiliary system $\wh{\calM}$ is defined as follows with $\wh{\Phi}_C$ being the approximated observer column space:
\begin{equation}\begin{split}\label{eq:auxiliary_1}
    x_{t+1} = {}& Ax_t + Bu_t + w_t,\\
    y_{t} = {}& \wh{\Phi}_C\t Cx_t + \wh{\Phi}_C\t \eta_t.
\end{split}\end{equation}
Now we are ready to restate \Cref{thm:1single} in full details.
\begin{theorem}[\Cref{thm:1single} Restated]\label{thm:1single_full}
    Consider $\calM$, datasets $\calD_1=\calU_1\cup\calY_1,\calD_2=\calU_2\cup\calY_2$ in \ac{hdsysid} and constants defined above. Suppose system $\calM$ satisfies Assumption \ref{assmp:sys1} with constants $\psi_A$ and $\rho_A$.
    If $T_1$  and $T_2$ satisfy
    \begin{equation}\begin{split}\label{eq:3col_1}
        &T_1 \gtrsim \kappa_3\cdot n^2r^3, \quad T_2 \geq \kappa_1\cdot \poly\b{r,m},
    \end{split}\end{equation}
    then $(\htA,\htB,\htC)$ from Algorithm \ref{alg:1single} satisfy the following for some invertible matrix $S$ with probability at least $1-\delta$
    \begin{equation}\begin{split}
        {}& \max\left\{\norm{S^{-1}AS - \htA}, \norm{S^{-1}B - \htB}, \norm{CS - \htC}\right\}\lesssim \kappa_4 \cdot \sqrt{\frac{n}{T_1}}\norm{\htC} + \kappa_2\cdot\sqrt{\frac{\poly\b{r,m}}{T_2}}.
    \end{split}\end{equation}
    Here $\kappa_1=\kappa_1(\wh{\calM}, \calU_2,\delta)$ and $\kappa_2=\kappa_2(\wh{\calM}, \calU_2,\delta)$ are defined in \Cref{def:idoracle}.
    $\kappa_3 = \kappa_3\b{\calM,\calU_{[2]},\delta}, \quad \kappa_4 = \kappa_4\b{\calM, \calU_1,\delta}$ are detailed below. All of them are problem-related constants independent of system dimensions  modulo logarithmic factors.
    \begin{equation}\begin{split}
        \kappa_3\b{\calM, \calU_{[2]},\delta} = {}& \max\left\{\kappa_4^2\frac{\psi_A^2\psi_C^2}{(1-\rho_A^2)\phi_O^2}, ~\b{\frac{\psi_\eta^2\psi_C^2\psi_w\psi_A^2}{(1-\rho_A^2)\phi_u^2\phi_C^4\phi_R^4}}^2\log^2\b{\frac{\psi_C^2\psi_w\psi_A^4}{1-\rho_A^2} r\log\frac{r}{\delta}}\log^4(\frac{r}{\delta})\right\},\\
        \kappa_4\b{\calM, \calU_1,\delta} = {}& \frac{\psi_\eta}{\phi_u\phi_C^2\phi_R^2}\sqrt{\log\frac{1}{\delta}}.
    \end{split}\end{equation}
\end{theorem}

\begin{proof}
Based on \Cref{eq:3col_1}, $T_1$ satisfies the condition of \Cref{lem:1col_full}. We apply \Cref{lem:1col_full} on $(\calD_1, \Sigma_\eta)$ and get the following with probability at least $1-\frac{\delta}{2}$
\begin{equation}
    \norm{\wh{\Phi}_C^\perp\t \Phi_C} \lesssim \kappa_4 \sqrt{\frac{n}{T_1}} \coloneqq \Delta_{\Phi}.
\end{equation}
The system generating dataset $\tilde{\calD}_2$ (line 3 in \Cref{alg:1single}), denoted by $\wh{\calM}$, is rewritten as in \Cref{eq:auxiliary_1}.

\paragraph{Step 1. We first show that $\wh{\calM}$ is still a minimal system.} The controllability directly comes from the fact that $R = \begin{bmatrix}
        B & AB & \dots & A^{r-1}B
    \end{bmatrix}$ is full row rank, because system $\calM$ is minimal. On the other hand, for the observability matrix $O$, we know that 
    \begin{equation}\begin{split}\label{eq:1_col_full_1}
        {}& \rank(O) = \rank \b{\begin{bmatrix}
            \wh{\Phi}_C\t C\\
            \wh{\Phi}_C\t CA\\
            \vdots\\
            \wh{\Phi}_C\t CA^{r-1}
        \end{bmatrix}}\\
        \geq {}& \rank\b{\diag\b{\wh{\Phi}_C, \dots, \wh{\Phi}_C}
        \begin{bmatrix}
            \wh{\Phi}_C\t C\\
            \wh{\Phi}_C\t CA\\
            \vdots\\
            \wh{\Phi}_C\t CA^{r-1}
        \end{bmatrix}}
        = \rank \b{\begin{bmatrix}
            \wh{\Phi}_C\wh{\Phi}_C\t C\\
            \wh{\Phi}_C\wh{\Phi}_C\t CA\\
            \vdots\\
            \wh{\Phi}_C\wh{\Phi}_C\t CA^{r-1}
        \end{bmatrix}}\\
        = {}& \rank \b{\begin{bmatrix}
            C\\
            CA\\
            \vdots\\
            CA^{r-1}
        \end{bmatrix} - \begin{bmatrix}
            \wh{\Phi}_C^\perp\wh{\Phi}_C^\perp\t C\\
            \wh{\Phi}_C^\perp\wh{\Phi}_C^\perp\t CA\\
            \vdots\\
            \wh{\Phi}_C^\perp\wh{\Phi}_C^\perp\t CA^{r-1}
        \end{bmatrix}}.
    \end{split}\end{equation}
    Consider the second term. We first observe that the following holds for all $i\in[r-1]$
    \begin{equation}\begin{split}
        \norm{\wh{\Phi}_C^\perp\wh{\Phi}_C^\perp\t CA^i} = {}& \norm{\wh{\Phi}_C^\perp\wh{\Phi}_C^\perp\t \Phi_C\Phi_C\t CA^i} \leq \norm{\wh{\Phi}_C^\perp\wh{\Phi}_C^\perp\t \Phi_C} \norm{\Phi_C\t CA^i}\\
        = {}& \norm{\wh{\Phi}_C^\perp\t \Phi_C} \norm{CA^i}\\
        \leq {}& \Delta_\Phi \psi_A\psi_C \rho_A^{i-1}.
    \end{split}\end{equation}Moreover, for $i=0$, we know 
    \begin{equation}\begin{split}
        \norm{\wh{\Phi}_C^\perp\wh{\Phi}_C^\perp\t C} \leq \norm{\wh{\Phi}_C^\perp\t \Phi_C} \norm{C}
        \leq \Delta_\Phi \psi_C.
    \end{split}\end{equation}
    Therefore, 
    \begin{equation}\begin{split}\label{eq:1_col_full_2}
        \norm{\begin{bmatrix}
            \wh{\Phi}_C^\perp\wh{\Phi}_C^\perp\t C\\
            \wh{\Phi}_C^\perp\wh{\Phi}_C^\perp\t CA\\
            \vdots\\
            \wh{\Phi}_C^\perp\wh{\Phi}_C^\perp\t CA^{r-1}
        \end{bmatrix}}
        = {}& \sqrt{\norm{\sum_{i=0}^{r-1} A^i\t C\t\wh{\Phi}_C^\perp\wh{\Phi}_C^\perp\t CA^i}}\\
        \leq {}& \sqrt{\sum_{i=0}^{r-1} \norm{\wh{\Phi}_C^\perp\t CA^i}^2}\\
        \leq {}& \sqrt{(\Delta_{\Phi}\psi_C)^2+\sum_{i=1}^{r-1} \b{\Delta_\Phi\psi_A\psi_C\rho_A^{i-1}}^2}\\
        \leq {}& 2\frac{\psi_A\psi_C}{\sqrt{1-\rho_A^2}}\Delta_\Phi
    \end{split}\end{equation}
    From \Cref{eq:3col_1}, we know that $T_{1} \gtrsim \kappa_4^2\frac{\psi_A^2\psi_C^2}{(1-\rho_A^2)\phi_O^2}n$. Combining Theorem 1 in \cite{stewart_matrix_1990} gives
    \begin{equation}\begin{split}
        {}& \sigma_r\b{\begin{bmatrix}
            C\\
            CA\\
            \vdots\\
            CA^{r-1}
        \end{bmatrix} - \begin{bmatrix}
            \wh{\Phi}_C^\perp\wh{\Phi}_C^\perp\t C\\
            \wh{\Phi}_C^\perp\wh{\Phi}_C^\perp\t CA\\
            \vdots\\
            \wh{\Phi}_C^\perp\wh{\Phi}_C^\perp\t CA^{r-1}
        \end{bmatrix}}\\
        \geq {}& \sigma_r \b{\begin{bmatrix}
            C\\
            CA\\
            \vdots\\
            CA^{r-1}
        \end{bmatrix}} - \norm{\begin{bmatrix}
            \wh{\Phi}_C^\perp\wh{\Phi}_C^\perp\t C\\
            \wh{\Phi}_C^\perp\wh{\Phi}_C^\perp\t CA\\
            \vdots\\
            \wh{\Phi}_C^\perp\wh{\Phi}_C^\perp\t CA^{r-1}
        \end{bmatrix}}\\
        \geq {}& \frac{\phi_O}{2} > 0.
    \end{split}\end{equation}
    This implies that 
    \begin{equation}\begin{split}
        r \geq \rank(O) \geq r.
    \end{split}\end{equation}
    Namely, the system is observable. Since $\b{A, B, \wh{\Phi}_C\t C}$ is controllable and observable, we conclude that $\wh{\calM}$ is minimal.

    \textbf{Step 2. We now apply \ac{oracle} on this minimal system. } 
    Recall the dynamics of $\wh{\calM}$:
    \begin{equation}\begin{split}\label{eq:auxiliary_2}
        x_{t+1} = {}& Ax_t + Bu_t + w_t,\\
        y_{t} = {}& \wh{\Phi}_C\t Cx_t + \wh{\Phi}_C\t \eta_t.
    \end{split}\end{equation}
    Since $\wh{\Phi}_C$ is independent of the second trajectory, $\{\wh{\Phi}_C\eta_{2,t}\}_{t=0}^{T_2}$ are i.i.d. noises independent of other variables of the second trajectory. Moreover, $\wh{\calM}$ satisfy Assumption \ref{assmp:sys1}. Therefore, we can apply \ac{oracle}. 
    Let $r_c = \rank(C)$. With $T_2 \geq \kappa_1\poly\b{r,r_c,m} = \kappa_1\poly\b{r,m}$, outputs $\htA, \htB, \tilC$ satisfy the following for some invertible matrix $S$ with probability at least $1-\frac{\delta}{2}$
    \begin{equation}\begin{split}\label{eq:proj_1}
        {}& \max\left\{\norm{S^{-1}AS - \htA}, \norm{S^{-1}B - \htB}, \norm{\wh{\Phi}_C\t CS - \tilC}\right\} \leq  \kappa_2 \sqrt{\frac{\poly\b{r,r_c,m}}{T_2}}
        \eqqcolon \Delta_{\wh{\calM}}.
    \end{split}\end{equation}
    Therefore, our final approximation for $C$ satisfies
    \begin{equation}\begin{split}\label{eq:1_col_full_3}
        \norm{CS - \htC} ={}&  \norm{CS - \wh{\Phi}_C\wh{\Phi}_C\t CS + \wh{\Phi}_C\wh{\Phi}_C\t CS - \wh{\Phi}_C\tilC}\\
        \leq {}& \norm{CS - \wh{\Phi}_C\wh{\Phi}_C\t CS} + \norm{\wh{\Phi}_C\wh{\Phi}_C\t CS - \wh{\Phi}_C\tilC}\\
        \leq {}& \norm{\wh{\Phi}_C^\perp\b{\wh{\Phi}_C^\perp}\t CS} + \norm{\wh{\Phi}_C\t CS - \tilC}\\
        \leq {}& \norm{\wh{\Phi}_C^\perp\b{\wh{\Phi}_C^\perp}\t CS} + \Delta_{\wh{\calM}}.
    \end{split}\end{equation}
    For the first term, we know that
    \begin{equation}\begin{split}
        \norm{\wh{\Phi}_C^\perp\b{\wh{\Phi}_C^\perp}\t CS} ={}&  \norm{\wh{\Phi}_C^\perp\b{\wh{\Phi}_C^\perp}\t \Phi_C\Phi_C\t CS} = \norm{\b{\wh{\Phi}_C^\perp}\t \Phi_C\Phi_C\t CS}\\
        \leq {}& \norm{\b{\wh{\Phi}_C^\perp}\t \Phi_C}\norm{\Phi_C\t}\norm{CS}\\
        \leq {}& \Delta_\Phi\norm{CS}.
    \end{split}\end{equation}
    Therefore, 
    \begin{equation}\begin{split}
        \norm{CS-\htC} \leq \Delta_\Phi\norm{CS} + \Delta_{\wh{\calM}}.
    \end{split}\end{equation}
    To get the upper bound w.r.t. $\norm{\htC}$, we notice that 
    \begin{equation}\begin{split}
        \norm{CS} \leq \norm{\htC} + \norm{CS-\htC} \leq \norm{\htC} + \Delta_{\wh{\calM}} + \Delta_{\Phi}\norm{CS}.
    \end{split}\end{equation}
    Rearraging the terms gives
    \begin{equation}\begin{split}
        \norm{CS} \leq \frac{\norm{\htC} + \Delta_{\wh{\calM}}}{1-\Delta_\Phi}.
    \end{split}\end{equation}
    Substituting back gives the following with probability at least $1-\delta$
    \begin{equation}\begin{split}
        \norm{CS-\htC} \leq {}& \frac{1}{1-\Delta_\Phi}\Delta_{\wh{\calM}} + \frac{\Delta_\Phi}{1-\Delta_\Phi}\norm{\htC}
        \leq 2\Delta_{\wh{\calM}} + \frac{\Delta_\Phi}{1-\Delta_\Phi}\norm{\htC}\\
        \overset{(i)}{\leq} {}& 2\Delta_{\wh{\calM}} + 2\Delta_{\Phi}\norm{\htC}\\
        \lesssim {}& \kappa_2\cdot\sqrt{\frac{\poly\b{r,r_c,m}}{T_2}} + \kappa_4 \cdot \sqrt{\frac{n}{T_{1}}}\norm{\htC}\\
        \leq {}& \kappa_2\cdot\sqrt{\frac{\poly\b{r,m}}{T_2}} + \kappa_4 \cdot \sqrt{\frac{n}{T_{1}}}\norm{\htC}\\
    \end{split}\end{equation} 
    Here $(i)$ is because $\Delta_{\Phi} \leq 1/2$ due to \Cref{eq:3col_1}. 
    Finally, we conclude that 
    \begin{equation}
        \max\left\{\norm{S^{-1}AS - \htA}, \norm{S^{-1}B - \htB}, \norm{CS - \htC}\right\} \lesssim \kappa_2\cdot\sqrt{\frac{\poly\b{r,m}}{T_2}} + \kappa_4 \cdot \sqrt{\frac{n}{T_1}}\norm{\htC}.
    \end{equation}
\end{proof}

\subsection[]{Upper Bounds for \ac{subspaceID}}



The theoretical guarantee for Col-Approx is presented in the following lemma.
\begin{lemma}[\Cref{lem:1col} Restated]\label{lem:1col_full}
    Consider system $\calM$, dataset $\calD_1=\calU_1\cup\calY_1$ in \ac{hdsysid} and constants defined at the beginning of \Cref{sec:hdsysid}. Suppose $\calM$ satisfies Assumption \ref{assmp:sys1} with constants $\psi_A$ and $\rho_A$.
    If
    \begin{equation}\begin{split}\label{eq:thm_col_1}
        T_1 \gtrsim \underbrace{\b{\frac{\psi_\eta^2\psi_C^2\psi_w\psi_A^2}{(1-\rho_A^2)\phi_u^2\phi_C^4\phi_R^4}}^2\log^2\b{\frac{\psi_C^2\psi_w\psi_A^4}{1-\rho_A^2} r\log\frac{r}{\delta}}\log^4(\frac{r}{\delta})}_{\kappa_5\b{\calM,\calU_1, \delta}}\cdot n^2r^3,
    \end{split}\end{equation}
    then $\wh{\Phi}_C = \text{col-approx}(\calY_1, \Sigma_\eta)$ satisfies the following with probability at least $1-\delta$
    \begin{equation}\begin{split}
        \htr_c = \rank(C), \quad \norm{\wh{\Phi}_C^\perp\t\Phi_C} \lesssim \underbrace{\frac{\psi_\eta}{\phi_u\phi_C^2\phi_R^2}\sqrt{\log\frac{1}{\delta}}}_{\kappa_4\b{\calM,\calU_1,\delta}}\cdot\sqrt{\frac{n}{T_1}}.
    \end{split}\end{equation}
\end{lemma}

\begin{proof}
For simplicity, we omit all subscript $1$ for the rest of this section. From the system dynamics, we know that
\begin{equation}\begin{split}
    \Sigma_y = {}& \sum_{t=0}^{T} y_ty_t\t = \sum_{t=0}^T C x_tx_t\t C\t + \sum_{t=0}^T \eta_t\eta_t\t + \sum_{t=0}^T \b{Cx_t\eta_t\t + \eta_t x_t\t C\t}\\
    = {}& \sum_{t=0}^T C x_tx_t\t C\t + \sum_{t=0}^T \b{\eta_t\eta_t\t-\sigma_\eta^2I} + \sum_{t=0}^T \b{Cx_t\eta_t\t + \eta_t x_t\t C\t} + (T+1)\sigma_\eta^2I\\
\end{split}\end{equation}
Here the first term is the information on $\col(C)$, while the second and third terms are noises. 
For the rest of the proof, we first upper bound norms of the noise terms (\textit{step 1}). With this, we show that $\htr_c = \rank(C)$ with high probability (\textit{step 2}). We then apply our subspace perturbation result to upper bound the influence of the noises on the eigenspace of the first term (\textit{step 3}).

\textbf{Step 1: Noise Norm Upper Bounds.} Define $r_c=\rank(C)$. Notice that we can write $C = \Phi_C\alpha$, where $\Phi_C\in\bbR^{n\times r_c}$ consists of orthonormal columns that form a basis of $\col(C)$ and $\alpha \in\bbR^{r_c\times r}$ is a full row rank matrix. 
It is then clear that
\begin{equation}\begin{split}
    \sigma_{\min}(\alpha ) = \sigma_{\min}(C) \geq \phi_C, \quad \sigma_1(\alpha ) = \sigma_1(C) \leq \psi_C.
\end{split}\end{equation}
Let $\Sigma_C = \sum_{t=0}^T C x_tx_t\t C\t$, $\bar{\Sigma}_C = \Sigma_C +  T  I$ and $\Sigma_\alpha = \sum_{t=0}^T \alpha  x_tx_t\t \alpha \t$. From the definitions, it is clear that $\Sigma_C = \Phi_C\Sigma_\alpha \Phi_C\t$ and $\bar{\Sigma}_C = \Phi_C\Sigma_\alpha \Phi_C\t +  T I$.
Then from Lemma \ref{lem:inv}, Lemma \ref{lem:inv_crossConc_stable}, Lemma \ref{lem:inv_cov_noise}, and Lemma \ref{lem:gram_upper}, we have the following for $T\gtrsim \frac{\psi_w^2\psi_C^4\psi_A^4}{\phi_u^2\phi_{C}^4\phi_R^4}r^3\log\frac{r}{\delta}\cdot\bar{\psi}^2$ (from Equation \ref{eq:thm_col_1}) with probability at least $1-\delta$ 
\begin{equation}\begin{split}
    {}& \frac{\phi_u\phi_C^2\phi_R^2}{8} T I \preceq \Sigma_\alpha \precsim \frac{\psi_C^2\psi_w\psi_A^2}{1-\rho_A^2}r T \log\frac{1}{\delta}I,\\
    {}& \norm{\b{\bar{\Sigma}_C}^{-\frac{1}{2}}\sum_{t=0}^T Cx_t\eta_t\t} \lesssim \bar{\psi}\sqrt{\psi_\eta n\log\frac{1}{\delta}},\\
    {}& \norm{\sum_{t=0}^T \b{\eta_t\eta_t\t - \sigma_{\eta}^2I}} \lesssim \psi_\eta\sqrt{n T \log\frac{1}{\delta}}.
\end{split}\end{equation}
Here $\bar{\psi} = \sqrt{\log\b{\frac{\psi_C^2\psi_w\psi_A^4}{1-\rho_A^2} r\log\frac{r}{\delta}}}$.

\textbf{Step 2: Order Estimation Guarantee.} Consider the following matrix
\begin{equation}\begin{split}
    \Sigma_y = {}& \Sigma_C + \underbrace{\sum_{t=0}^T \b{\eta_t\eta_t\t - \sigma_{\eta}^2I} + \sum_{t=0}^T \b{Cx_t\eta_t\t + \eta_t C\t x_t\t}}_{\Delta} + (T+1)\sigma_\eta^2I.
\end{split}\end{equation}
The inequalities of Step 1 imply
\begin{equation}\begin{split}
    \norm{\Delta} \leq {}& \norm{\sum_{t=0}^T \b{\eta_t\eta_t\t - \sigma_{\eta}^2I}} + 2\norm{\sum_{t=0}^T Cx_t\eta_t\t}\\
    \lesssim {}& \psi_\eta\sqrt{n T \log\frac{1}{\delta}} + \bar{\psi}\sqrt{\psi_\eta n\log\frac{1}{\delta}}\norm{(\bar{\Sigma}_C)^{\frac{1}{2}}}\\
    = {}& \psi_\eta\sqrt{n T \log\frac{1}{\delta}} + \bar{\psi}\sqrt{\psi_\eta n\log\frac{1}{\delta}}\b{\norm{( T  I+\Sigma_\alpha)^{\frac{1}{2}}}}\\
    \lesssim {}& \psi_\eta\sqrt{n T \log\frac{1}{\delta}} + \bar{\psi}\sqrt{\psi_\eta n\log\frac{1}{\delta}}\sqrt{\frac{\psi_C^2\psi_w\psi_A^2}{1-\rho_A^2}r T \log\frac{1}{\delta}+ T }\\
    \lesssim {}& \bar{\psi}\psi_\eta \sqrt{n\log\frac{1}{\delta}}\sqrt{\frac{\psi_C^2\psi_w\psi_A^2}{1-\rho_A^2}r T \log\frac{1}{\delta}}\\
    = {}& \frac{\sqrt{\psi_\eta^2\psi_C^2\psi_w\psi_A^2}}{\sqrt{1-\rho_A^2}}\bar{\psi} \log\b{\frac{1}{\delta}} \sqrt{nr T }\\
\end{split}\end{equation}
Therefore, for each $i\in[r_c]$, we have the following for some positive constant $c_{1}$:
\begin{equation}\begin{split}
    \sigma_i\b{\Sigma_y} - (T+1)\sigma_\eta^2 \geq {}& \sigma_i(\Sigma_C) - \norm{\Delta} \geq \sigma_{\min}(\Sigma_\alpha) - \norm{\Delta}\\
    \geq {}& \frac{\phi_u\phi_C^2\phi_R^2}{8} T  - c_1\frac{\sqrt{\psi_\eta^2\psi_C^2\psi_w\psi_A^2}}{\sqrt{1-\rho_A^2}}\bar{\psi} \log\b{\frac{1}{\delta}} \sqrt{nr T }\\
    \geq {}& \frac{\phi_u\phi_C^2\phi_R^2}{16} T.
\end{split}\label{eq:sgima-T-lowerbound}\end{equation}
Here the first inequality holds according to Theorem 1 in \cite{stewart_matrix_1990} and the last line is because $T\gtrsim \frac{\psi_\eta^2\psi_C^2\psi_w\psi_A^2}{(1-\rho_A^2)\phi_u^2\phi_C^4\phi_R^4}\log\b{\frac{\psi_C^2\psi_w\psi_A^4}{1-\rho_A^2} r\log\frac{r}{\delta}}\log^2(\frac{1}{\delta})\cdot nr$ (from Equation \ref{eq:thm_col_1}). 
For $i\in[r_c+1, n]$, 
\begin{equation}\begin{split}
    \sigma_i(\Sigma_y) - (T+1)\sigma_\eta^2 \leq \norm{\Delta} \leq c_1\frac{\sqrt{\psi_\eta^2\psi_C^2\psi_w\psi_A^2}}{\sqrt{1-\rho_A^2}}\bar{\psi} \log\b{\frac{1}{\delta}} \sqrt{nr T }.
\end{split}\end{equation}
Based on the above three inequalities and Equation \ref{eq:thm_col_1}, we conclude that with probability at least $1-\delta$, the following hold 
\begin{equation*}\begin{split}
    \sigma_{i}(\Sigma_y) - \sigma_{j}(\Sigma_y) \leq {}& 2c_1\frac{\sqrt{\psi_\eta^2\psi_C^2\psi_w\psi_A^2}}{\sqrt{1-\rho_A^2}}\bar{\psi} \log\b{\frac{1}{\delta}} \sqrt{nr T }\\
    < {}& T^{3/4},\quad \forall i<j\in[r_c+1,n]\\
    \sigma_{r_c}(\Sigma_y) - \sigma_{r_c+1}(\Sigma_y) \geq {}& \frac{\phi_u\phi_C^2\phi_R^2}{16} T - c_1\frac{\sqrt{\psi_\eta^2\psi_C^2\psi_w\psi_A^2}}{\sqrt{1-\rho_A^2}}\bar{\psi} \log\b{\frac{1}{\delta}} \sqrt{nr T } \geq \frac{\phi_u\phi_C^2\phi_R^2}{32} T\\
    > {}& T^{3/4}.
\end{split}\end{equation*}
The above inequality follows from Equation \ref{eq:thm_col_1} similar to (\ref{eq:sgima-T-lowerbound}).
Therefore, from the definition of $\htr_c$, we know that $\htr_c = r_c$ with probability at least $1-\delta$.

\textbf{Step 3: Column Space Estimation Guarantee.} With $\htr_c = r_c$, now we try to apply our subspace perturbation result, i.e. Lemma \ref{lem:perturbation}, on matrix $\Sigma_y-\sigma_\eta^2(T+1)I + TI$. Notice that this matrix has exactly the same eigenspace as $\Sigma_y$ and therefore the eigenspace of its first $r_c$ eigenvectors is $\wh{\Phi}_C$ (line 9 in Algorithm \ref{alg:column_space}). This matrix can be decomposed as
\begin{equation}\begin{split}
    \Sigma_y -\sigma_\eta^2(T+1)I + TI = {}& \Sigma_C + \sum_{t=0}^T \b{\eta_t\eta_t\t - \sigma_{\eta}^2I} + \sum_{t=0}^T \b{Cx_t\eta_t\t + \eta_t C\t x_t\t} + TI.\\
    = {}& \underbrace{\bar{\Sigma}_C}_{M \text{ in Lemma \ref{lem:perturbation}}} + \underbrace{\sum_{t=0}^T \b{\eta_t\eta_t\t - \sigma_{\eta}^2I}}_{\Delta_2 \text{ in Lemma \ref{lem:perturbation}}} + \underbrace{\sum_{t=0}^T \b{Cx_t\eta_t\t + \eta_t C\t x_t\t}}_{\Delta_1+\Delta_1\t \text{ in Lemma \ref{lem:perturbation}}}.
\end{split}\end{equation}
For matrix $\bar{\Sigma}_C = \Phi_C \Sigma_\alpha \Phi_C\t +  T  I = \Phi_C\b{\Sigma_\alpha+ T  I}\Phi_C\t +  T \Phi_C^\perp\Phi_C^\perp\t$, it is clear that its SVD can be written as
\begin{equation}\begin{split}
    \bar{\Sigma}_C = \begin{bmatrix}
        \Phi_C & \Phi_C^\perp
    \end{bmatrix} \begin{bmatrix}
        \Lambda_1 & \\
        &  T  I
    \end{bmatrix} \begin{bmatrix}
        \Phi_C\t\\
        \Phi_C^\perp\t
    \end{bmatrix}, \quad \Lambda_1 = \diag\b{\sigma_1(\Sigma_\alpha)+ T , \dots, \sigma_{\min}(\Sigma_\alpha)+ T }. 
\end{split}\end{equation}
where $\Phi_C$ is an orthonormal basis of $\col(C)$. 
Then we conclude that the following holds for some large enough positive constant $c_{3}$
\begin{equation}\begin{split}
    {}& \sigma_1\b{\bar{\Sigma}_C} \leq  T  + \sigma_1(\Sigma_\alpha) \leq c_{3} \b{1+\frac{\psi_C^2\psi_w\psi_A^2}{1-\rho_A^2}r\log\frac{1}{\delta}} T \leq 2c_{3} \b{\frac{\psi_C^2\psi_w\psi_A^2}{1-\rho_A^2}r\log\frac{1}{\delta}} T ,\\
    {}& \sigma_{r_c}\b{\bar{\Sigma}_C} - \sigma_{r_c+1}(\bar{\Sigma}_C) = \sigma_{\min}(\Sigma_\alpha) \geq \frac{\phi_u\phi_C^2\phi_R^2}{32} T , \quad \sigma_{\min}(\bar{\Sigma}_C) =  T .
\end{split}\end{equation}
Now we are ready to apply Lemma \ref{lem:perturbation} on $\bar{\Sigma}_C$ with the following constants for $c_{4}, c_{5}$ large enough
\begin{equation}\begin{split}
    {}& a = c_{4}\bar{\psi}\sqrt{\psi_\eta n\log\frac{1}{\delta}}, \quad b = c_{5}\psi_\eta\sqrt{n T \log\frac{1}{\delta}},\\
    {}& \delta_M = \frac{\phi_u\phi_C^2\phi_R^2}{32} T , \quad \psi_M = 2c_{3}\frac{\psi_C^2\psi_w\psi_A^2}{1-\rho_A^2}r T \log\frac{1}{\delta}, \quad \phi_M =  T , \quad \sigma_1(\Lambda_2) =  T .
\end{split}\end{equation}
Again, $\bar{\psi} = \sqrt{\log\b{\frac{\psi_C^2\psi_w\psi_A^4}{1-\rho_A^2} r\log\frac{r}{\delta}}}$.
From Equation \ref{eq:thm_col_1}, it is clear that $\sqrt{\phi_M} \gtrsim \alpha$ and $\delta_M \gtrsim \alpha\sqrt{\psi_M} + \beta$. Therefore,
\begin{equation}\begin{split}
    \norm{\wh{\Phi}_C^\perp\t\Phi_C} = \norm{\wh{\Phi}_C^\perp\t\tilde{\Phi}_C} \lesssim {}& \frac{\psi_\eta}{\phi_u\phi_C^2\phi_R^2}\sqrt{\frac{n}{ T }\log\frac{1}{\delta}} + \bar{\psi}\frac{\sqrt{\psi_\eta}}{\phi_u\phi_C^2\phi_R^2}\sqrt{\frac{n}{ T }\log\frac{1}{\delta}} + \bar{\psi}^2\frac{\psi_\eta\sqrt{\psi_C^2\psi_w\psi_A^2}}{\phi_u\phi_C^2\phi_R^2\sqrt{1-\rho_A^2}}\frac{n\sqrt{r}}{ T }\log^2\frac{1}{\delta}\\
    \lesssim {}& \frac{\psi_\eta}{\phi_u\phi_C^2\phi_R^2}\sqrt{\frac{n}{ T }\log\frac{1}{\delta}} + \bar{\psi}^2\frac{\psi_\eta\sqrt{\psi_C^2\psi_w\psi_A^2}}{\phi_u\phi_C^2\phi_R^2\sqrt{1-\rho_A^2}}\frac{n\sqrt{r}}{ T }\log^2\frac{1}{\delta}\\
    \lesssim {}& \frac{\psi_\eta}{\phi_u\phi_C^2\phi_R^2}\sqrt{\frac{n}{ T }\log\frac{1}{\delta}}.
\end{split}\end{equation}
\end{proof}

\subsubsection{Supporting Details}
\begin{lemma}\label{lem:inv}
    Consider the same setting as \Cref{lem:1col_full}. 
    Consider full row rank matrices $\{\alpha \in\bbR^{a\times r}\}_{k\in[K]}$ for any positive integer $a\leq r$.
    Let
    \begin{equation}\begin{split}
        \Sigma_\alpha = \sum_{t\in[T]} \alpha x_tx_t\t \alpha \t, \quad \phi_\alpha = \min\{\sigma_{\min}\b{\alpha }, 1\}, \quad \psi_\alpha = \max\{\sigma_1\b{\alpha },1\}.
    \end{split}\end{equation}
    If $T  \gtrsim \frac{\psi_w^2\psi_\alpha^4\psi_A^4}{\phi_u^2\phi_{\alpha}^4\phi_R^4}r^3\log\frac{r}{\delta}\cdot\log\b{\frac{\psi_\alpha^2\psi_w\psi_A^4}{1-\rho_A^2} r\log\frac{r}{\delta}}$, then the following holds with probability at least $1-\delta$ 
    \begin{equation}\begin{split}
        \Sigma_\alpha \succeq \frac{\phi_u\phi_\alpha^2\phi_R^2}{8} T I.    
    \end{split}\end{equation}
\end{lemma}

\begin{proof}
    For further analysis, we let $\tilw_{t} \coloneqq Bu_{t} + w_{t} \sim \calN\b{0, \Sigma_{\tilw}}$ with $\Sigma_{\tilw} \coloneqq B\Sigma_{u}B\t + \Sigma_{w}$ and we can rewrite the system dynamics as follows
\begin{equation}\begin{split}\label{eq:system_meta}
    x_{t+1} = {}& Ax_t + Bu_{t} + w_{t} = Ax_t + \tilw_{t},\\
    y_t = {}& Cx_t + \eta_t.
\end{split}\end{equation}

For simplicity, we define $\Sigma_{\alpha,\tau} \coloneqq \alpha A^{\tau}\b{\sum_{t=1}^{T-\tau}x_\tau x_\tau\t} (\alpha A^{\tau})\t$. It is then clear that
    \begin{equation}\begin{split}
        \Sigma_{\alpha} = {}& \Sigma_{\alpha,0} = \sum_{t=1}^{T} \alpha x_tx_t\t \alpha \t\\
        ={}&   \sum_{t=1}^{T} \alpha \b{Ax_{t-1}x_{t-1}\t A\t  + \tilw_{t-1}\tilw_{t-1}\t  + Ax_{t-1}\tilw_{t-1}\t  + \tilw_{t-1}x_{t-1}\t A\t} \alpha \t \\
        ={}& \Sigma_{\alpha,1} +  \sum_{t=0}^{T-1} \alpha \tilw_{t}\tilw_{t}\t \alpha \t  +  \sum_{t=0}^{T-1} \b{\alpha Ax_t\b{\alpha \tilw_{t}}\t  + \b{\alpha \tilw_{t}}x_t\t A\t \alpha \t}.
    \end{split}\label{eq:covA_1}\end{equation}
    By Lemma \ref{lem:inv_crossConc_stable} and \ref{lem:inv_cov_noise}, the following events hold with probability at least $1-\delta/r$, 
    \begin{equation}\begin{split}
        \norm{\sum_{t=0}^{T-1} \alpha \tilw_{t}\tilw_{t}\t \alpha \t - T\cdot \alpha \Sigma_{\tilw}\alpha\t} \lesssim {}& \psi_w\psi_\alpha^2\sqrt{rT\log\frac{2r}{\delta}}, \\
        \norm{\b{\Sigma_{\alpha,1} + TI}^{-\frac{1}{2}} \sum_{t=0}^{T-1} \alpha Ax_t\b{\alpha \tilw_{t}}\t} \lesssim {}& \bar{\psi}\sqrt{r\psi_w\psi_\alpha^2\log\frac{2r}{\delta}},
    \end{split}\end{equation}
    with $\bar{\psi} = \sqrt{\log\b{\frac{\psi_\alpha^2\psi_w\psi_A^4}{1-\rho_A^2} r\log\frac{2r}{\delta}}}$.
    From the first inequality, it is clear that the following holds for some large enough positive constant $c_1$
    \begin{equation}\begin{split}
        \sum_{t=0}^{T-1} \alpha \tilw_{t}\tilw_{t}\t \alpha \t  \succeq T\cdot \alpha \Sigma_{\tilw}\alpha\t -c_1 \psi_w\psi_{\alpha}^2\sqrt{rT\log\frac{2r}{\delta}}I.
    \end{split}\end{equation}
    From the second inequality, we have the following
    \begin{equation}\begin{split}
        {}& \norm{\b{\Sigma_{\alpha,1}+ T  I}^{-\frac{1}{2}} \sum_{t=0}^{T-1} \alpha Ax_t\b{\alpha \tilw_{t}}\t \b{\Sigma_{\alpha,1}+ T  I}^{-\frac{1}{2}}}\\
        \lesssim {}& \bar{\psi}\sqrt{r\psi_w\psi_\alpha^2\log\frac{2r}{\delta}} \norm{\b{\Sigma_{\alpha,1} +  TI}^{-\frac{1}{2}}}\\
        \lesssim {}& \bar{\psi}\sqrt{\frac{r\psi_w\psi_\alpha^2\log\frac{2r}{\delta}}{ T }}.
    \end{split}\end{equation}
    This implies that the following holds for large enough positive constant $c_2$
    \begin{equation}\begin{split}
        {}&  \sum_{t=0}^{T-1} \b{\alpha Ax_t\b{\alpha \tilw_{t}}\t  + \b{\alpha \tilw_{t}}x_t\t A\t \alpha \t} 
        \succeq -c_2\bar{\psi}\sqrt{\frac{r\psi_w\psi_\alpha^2\log\frac{2r}{\delta}}{ T }} \b{\Sigma_{\alpha,1} + TI}.
    \end{split}\end{equation}
    Plugging back into Equation (\ref{eq:covA_1}) gives the following for some positive constant $c_3$
    \begin{equation}\begin{split}\label{eq:covA_2}
        \Sigma_{\alpha,0} \succeq {}& \Sigma_{\alpha,1} + T \alpha \Sigma_{\tilw}\alpha\t - c_1 \psi_w\psi_{\alpha}^2\sqrt{rT\log\frac{2r}{\delta}}I - c_2\bar{\psi}\sqrt{\frac{r\psi_w\psi_\alpha^2\log\frac{2r}{\delta}}{ T }}\b{\Sigma_{\alpha,1} +  T  I}\\
        \succeq {}& \b{1 - c_2\bar{\psi}\sqrt{\frac{r\psi_w\psi_\alpha^2\log\frac{2r}{\delta}}{ T }}}\Sigma_{\alpha,1} + T\alpha\Sigma_{\tilw}\alpha\t\\
        {}& \hspace{2em} - c_1 \psi_w\psi_{\alpha}^2\sqrt{rT\log\frac{2r}{\delta}}I - c_2\bar{\psi}\sqrt{r\psi_w\psi_\alpha^2\log\frac{2r}{\delta}T}I\\
        \succeq {}& \b{1 - c_2\bar{\psi}\sqrt{\frac{r\psi_w\psi_\alpha^2\log\frac{2r}{\delta}}{ T }}}\Sigma_{\alpha,1} + T\alpha\Sigma_{\tilw}\alpha\t - c_3\psi_w\psi_\alpha^2\bar{\psi}\sqrt{rT\log\frac{2r}{\delta}}I\\
    \end{split}\end{equation}
    Similarly, we expand $\Sigma_{\alpha,1}, \Sigma_{\alpha,2}, \cdots, \Sigma_{\alpha,r-1}$ and have the following with probability at least $1-\delta$
    {\small\begin{equation}\begin{split}
        \Sigma_{\alpha,0} \succeq {}& \b{T\alpha\Sigma_{\tilw}\alpha\t - c_3\psi_w\psi_\alpha^2\bar{\psi}\sqrt{rT\log\frac{2r}{\delta}}I} + \b{1 - c_2\bar{\psi}\sqrt{\frac{r\psi_w\psi_\alpha^2\log\frac{2r}{\delta}}{ T }}}\Sigma_{\alpha,1}\\
        \succeq {}& \b{T\alpha\Sigma_{\tilw}\alpha\t - c_3\psi_w\psi_\alpha^2\bar{\psi}\sqrt{rT\log\frac{2r}{\delta}}I}\\
        {}& \hspace{2em} + \b{(T-1)(\alpha A)\Sigma_{\tilw}(\alpha A)\t - c_3\psi_w\psi_\alpha^2\psi_A^2\bar{\psi}\sqrt{r(T-1)\log\frac{2r}{\delta}}I}\cdot \b{1 - c_2\bar{\psi}\sqrt{\frac{r\psi_w\psi_\alpha^2\log\frac{2r}{\delta}}{ T }}}\\
        {}& \hspace{2em} + \b{1 - c_2\bar{\psi}\sqrt{\frac{r\psi_w\psi_\alpha^2\log\frac{2r}{\delta}}{ T }}}\b{1 - c_2\bar{\psi}\sqrt{\frac{r\psi_w\psi_\alpha^2\psi_A^2\log\frac{2r}{\delta}}{ T-1 }}}\Sigma_{\alpha,2}\\
        \succeq {}& \cdots \succeq \b{(T-r+1)\sum_{i=0}^{r-1}\b{\alpha A^i}\Sigma_{\tilw}\b{\alpha A^i}\t - c_3\psi_w\psi_\alpha^2\psi_A^2\bar{\psi} \cdot r\sqrt{rT\log\frac{2r}{\delta}}I} \b{1 - c_2\bar{\psi}\sqrt{\frac{r\psi_w\psi_\alpha^2\psi_A^2\log\frac{2r}{\delta}}{ T-r+1 }}}^{r-1}\\
        {}& \hspace{4em}+\b{1 - c_2\bar{\psi}\sqrt{\frac{r\psi_w\psi_\alpha^2\psi_A^2\log\frac{2r}{\delta}}{ T-r+1 }}}^r\Sigma_{\alpha,r}\\
    \end{split}\end{equation}}
    The above inequality is further simplified as follows
    {\small\begin{equation}\begin{split}
        \Sigma_{\alpha,0} \overset{(i)}{\succeq} {}& \b{\phi_u(T-r+1)\sum_{i=0}^{r-1}\b{\alpha A^i}BB\t\b{\alpha A^i}\t - c_3\psi_w\psi_\alpha^2\psi_A^2\bar{\psi} \cdot r\sqrt{rT\log\frac{2r}{\delta}}I} \b{1 - c_2\bar{\psi}\sqrt{\frac{r\psi_w\psi_\alpha^2\psi_A^2\log\frac{2r}{\delta}}{ T-r+1 }}}^{r}\\
        \overset{(ii)}{\succeq} {}& \b{\phi_u(T-r+1)\alpha \b{\sum_{i=0}^{r-1}A^iB(A^iB)\t}\alpha\t - c_3\psi_w\psi_\alpha^2\psi_A^2\bar{\psi} \cdot r\sqrt{rT\log\frac{2r}{\delta}}I} \b{1 - \frac{1}{2r}}^{r}\\
        \overset{(iii)}{\succeq} {}& \frac{1}{2}\b{\phi_u(T-r+1)\alpha \b{\sum_{i=0}^{r-1}A^iB(A^iB)\t}\alpha\t - c_3\psi_w\psi_\alpha^2\psi_A^2\bar{\psi} \cdot r\sqrt{rT\log\frac{2r}{\delta}}I}\\
        \succeq {}& \frac{1}{2}\b{\phi_u\phi_\alpha^2\phi_R^2(T-r+1) - c_3\psi_w\psi_\alpha^2\psi_A^2\bar{\psi} \cdot r\sqrt{rT\log\frac{2r}{\delta}}}I\\
        \overset{(iv)}{\succeq} {}& \frac{\phi_u\phi_\alpha^2\phi_R^2}{8} T I\\
    \end{split}\end{equation}}
    Here $(i)$ is because $\Sigma_{\alpha,r} \succeq 0$ and $\Sigma_{\tilw} \succeq B\Sigma_u B\t \succeq \phi_u BB\t$, $(ii)$ is because $T  \gtrsim \psi_w\psi_\alpha^2\psi_A^2r^3\log\frac{r}{\delta}\cdot\log\b{\frac{\psi_\alpha^2\psi_w\psi_A^4}{1-\rho_A^2} r\log\frac{r}{\delta}}$, $(iii)$ is because $(1-\frac{1}{2r})^r \geq \frac{1}{2}$ for all positive integers and $(iv)$ is because $T  \gtrsim \frac{\psi_w^2\psi_\alpha^4\psi_A^4}{\phi_u^2\phi_{\alpha}^4\phi_R^4}r^3\log\frac{r}{\delta}\cdot\log\b{\frac{\psi_\alpha^2\psi_w\psi_A^4}{1-\rho_A^2} r\log\frac{r}{\delta}}$.
\end{proof}

\begin{lemma}\label{lem:inv_crossConc_stable}
    Consider the same setting as \Cref{lem:1col_full}. 
    For any positive integer $a$, let $\{\zeta_{t}\in\bbR^a\}_{t=0}^{T}$ be a sequence of i.i.d Gaussian vectors from $\calN(0,\Sigma_{\zeta})$ such that $\zeta_{t}$ is independent of $x_t$. Define $\psi_\zeta = \sigma_1(\Sigma_{\zeta})$.
    We make the following defition for any matrix $P\in\bbR^{b\times r}$ for positive integer $b$
    \begin{equation}\begin{split}
        \bar{\Sigma}_P = \sum_{t\in[T]} Px_tx_t\t P\t +  T  I, \quad \psi_P = \sigma_1(P).
    \end{split}\end{equation}
    Then for any $\delta$, the following holds with probability at least $1-\delta$,
    \begin{equation}\begin{split}
        \norm{\b{\bar{\Sigma}_P}^{-\frac{1}{2}} \sum_{t\in[T]} Px_t\zeta_{t}\t} \lesssim \bar{\psi}\sqrt{\max\{r,a\}\psi_\zeta\log\frac{1}{\delta}}.
    \end{split}\end{equation}
    Here $\bar{\psi} = \sqrt{\log\b{\frac{\psi_P^2\psi_w\psi_A^2}{1-\rho_A^2} r\log\frac{1}{\delta}}}$. 
\end{lemma}

\begin{proof}
    From Lemma \ref{lem:gen_norm2Net}, we know that the following holds for some vector $v\in\mathbb{S}^{a-1}$
    \begin{equation}\label{eq:thm1_1}\begin{split}
        \bbP\b{\norm{\b{\bar{\Sigma}_P}^{-\frac{1}{2}} \sum_{t\in[T]} Px_t\zeta_{t}\t}>z} \leq 5^a \bbP\b{\norm{\b{\bar{\Sigma}_P}^{-\frac{1}{2}} \sum_{t\in[T]} Px_t\zeta_{t}\t v}>\frac{z}{2}}.
    \end{split}\end{equation}
    Notice that $\zeta_{t}\t v$ are independent Gaussian variables from distribution $\calN(0, v\t \Sigma_{\zeta}v)$ for all $t\in[T]$, which is $c_{1}\sqrt{\psi_\zeta}$-subGaussian for some positive constant $c_{1}$. Then applying Theorem 1 in \cite{abbasi-yadkori_improved_2011} on sequence $\{Px_t\}_{t\in[T]}$ and sequence $\{\zeta_{t}\t v\}_{t\in[T]}$, gives the following inequality
    \begin{equation}\begin{split}
        \bbP\b{\norm{\b{\bar{\Sigma}_P}^{-\frac{1}{2}} \sum_{t\in[T]} Px_t\zeta_{t}\t v} >  \sqrt{2c_{1}^2\psi_\zeta  \log\b{\frac{\det(\bar{\Sigma}_P)^{\frac{1}{2}}\det( T  I)^{-\frac{1}{2}}}{\delta}}}} \leq \delta.
    \end{split}\end{equation}
    Substituting the above result back gives the following inequality
    \begin{equation}\begin{split}
        \bbP\b{\norm{\b{\bar{\Sigma}_P}^{-\frac{1}{2}} \sum_{t\in[T]} Px_t\zeta_{t}} >  \sqrt{8c_{1}^2\psi_\zeta  \log\b{\frac{\det(\bar{\Sigma}_P)^{\frac{1}{2}}\det( T  I)^{-\frac{1}{2}}}{\delta}}}} \leq 5^a\delta,
    \end{split}\end{equation}
    which implies the following inequality holds with probability at least $1-\delta/2$
    \begin{equation}\label{eq:cross_1}\begin{split}
        \norm{\b{\bar{\Sigma}_P}^{-\frac{1}{2}} \sum_{t\in[T]} Px_t\zeta_{t}} \leq \sqrt{8c_{1}^2\psi_\zeta\log\b{\frac{2\det(\bar{\Sigma}_P)^{\frac{1}{2}}\det( T  I)^{-\frac{1}{2}}}{\delta}} + 8c_{1}^2\psi_\zeta a\log5}.
    \end{split}\end{equation}

    Now consider $\bar{\Sigma}_P = \sum_{t\in[T]} Px_tx_t\t P\t  +  T  I$. Then
    \begin{equation}\begin{split}
        \det( T  I) ={}&  T^b, \quad \det(\bar{\Sigma}_P) = T^{b-r}\prod_{i=1}^{r} \b{ T  +\lambda_i\b{\sum_{t\in[T]} Px_tx_t\t P\t }}.
    \end{split}\end{equation}
    From Lemma \ref{lem:gram_upper}, we have the following with probability at least $1-\delta/2$
    \begin{equation}\begin{split}
        \norm{\sum_{t\in[T]} x_tx_t\t} \lesssim \frac{\psi_w\psi_A^2r}{1-\rho_A^2}T\log\frac{2}{\delta}\lesssim \frac{\psi_w\psi_A^2r}{1-\rho_A^2}T\log\frac{1}{\delta}.
    \end{split}\end{equation}
    Therefore,
    \begin{equation}\begin{split}
        \lambda_i\b{\sum_{t\in[T]} Px_tx_t\t P\t }
        \leq {}& \norm{\sum_{t\in[T]} x_tx_t\t}\norm{ PP\t} \\
        \lesssim {}& \frac{\psi_P^2\psi_w\psi_A^2}{1-\rho_A^2}r T \log\frac{1}{\delta}\\
    \end{split}\end{equation}
    Substituting back gives the following for some positive constant $c_{2}$
    \begin{equation}\begin{split}\label{eq:cross_2}
        \det(\bar{\Sigma}_P) ={}&   T^{b-r}\prod_{i=1}^{r} \b{ T  +\lambda_i\b{\sum_{t\in[T]} Px_tx_t\t P\t }}\\
        \leq {}&  T^b \b{1+c_{2}\frac{\psi_P^2\psi_w\psi_A^2}{1-\rho_A^2} r\log\frac{1}{\delta}}^r,
    \end{split}\end{equation}
    which gives
    \begin{equation}\begin{split}
        \det(\bar{\Sigma}_P)^{\frac{1}{2}}\det( T  I)^{-\frac{1}{2}} \leq \b{1+c_{2}\frac{\psi_P^2\psi_w\psi_A^2}{1-\rho_A^2} r\log\frac{1}{\delta}}^{\frac{r}{2}}.
    \end{split}\end{equation}

    Finally, with a union bound over above events, we get the following with probability at least $1-\delta$ from Equation \ref{eq:cross_1} and \ref{eq:cross_2}
    \begin{equation}\begin{split}
        {}& \norm{\b{\bar{\Sigma}_P}^{-\frac{1}{2}} \sum_{t\in[T]} Px_t\zeta_{t}}\\
        \leq {}& \sqrt{8c_{1}^2\psi_\zeta\log\b{\frac{2\det(\bar{\Sigma}_P)^{\frac{1}{2}}\det( T  I)^{-\frac{1}{2}}}{\delta}} + 8c_{1}^2\psi_\zeta a\log5}\\
        \leq {}& \sqrt{4c_{1}^2r\psi_\zeta\log\b{2+2c_{2}\frac{\psi_P^2\psi_w\psi_A^2}{1-\rho_A^2} r\log\frac{1}{\delta}} + 8c_{1}^2\psi_\zeta\log\frac{1}{\delta} + 8c_{1}^2\psi_\zeta a\log5}\\
        \lesssim {}& \sqrt{\max\{r,a\}\psi_\zeta\log\frac{1}{\delta}}\cdot\sqrt{\log\b{\frac{\psi_P^2\psi_w\psi_A^2}{1-\rho_A^2} r\log\frac{1}{\delta}}}.
    \end{split}\end{equation}
    This completes the proof.
\end{proof} 

\begin{lemma}\label{lem:inv_cov_noise} 
    For any positive integer $a$, let $\{w_{t}\in\bbR^a\}_{t\in[T]}$ be independent Gaussian vectors from $w_t\sim\calN(0, \Sigma_{w,t})$. Let 
    \begin{equation}\begin{split}
        \psi_w = \max_{t\in[T]} \sigma_1(\Sigma_{w,t}), \quad \phi_w = \min_{t\in[T]} \sigma_{\min}(\Sigma_{w,t}).
    \end{split}\end{equation}
    Then with probability at least $1-\delta$, 
    \begin{equation}\begin{split}
        \norm{\sum_{t=1}^T w_{t}w_{t}\t - \sum_{t=1}^T \Sigma_{w,t}} \lesssim \psi_w \sqrt{aT\log\frac{1}{\delta}}.
    \end{split}\end{equation}
\end{lemma}

\begin{proof}
    From Lemma \ref{lem:gen_norm2Net}, we know that there exist $x, y\in\mathbb{S}^{a-1}$ s.t.
    \begin{equation}\begin{split}\label{eq:cov_dependent_1}
        \bbP\b{\norm{\sum_{t=1}^T w_tw_t\t - \sum_{t=1}^T \Sigma_{w,t}} > z} \leq 5^{2a}\bbP\b{\left|\sum_{t=1}^T \Big[(y\t w_{t}) (x\t w_{t}) - y\t \Sigma_{w,t}x\Big]\right| > \frac{1}{4}z}
    \end{split}\end{equation}
    For any $x, y$, following the definition of $w_{t}$, we have 
    \begin{equation}\begin{split}
        \|(y\t w_{t}) (x\t w_{t}) - y\t \Sigma_{w,t}x\|_{\psi_1} \lesssim \norm{(y\t w_{t}) (x\t w_{t})}_{\psi_1} \leq \norm{y\t w_t}_{\psi_2}\norm{x\t w_t}_{\psi_2} \lesssim \psi_w.
    \end{split}\end{equation}
    Here the first and second inequalities are from Lemma 2.6.8 and Lemma 2.7.7 in \cite{vershynin_high-dimensional_2018}, where $\|\cdot\|_{\psi_1}$ and $\|\cdot\|_{\psi_2}$ denote the sub-exponential and sub-Gaussian norms, respectively.
    Directly applying Bernstein's inequality on $\{(y\t w_{t}) (x\t w_{t}) - y\t \Sigma_{w,t}x\}_{t=1}^T$ (as sub-exponential random variables) gives the following for some positive constant $c_{1}$ with probability at least $1-\delta'$
    \begin{equation}\begin{split}
        \left|\sum_{t=1}^T \Big[(y\t w_{t}) (x\t w_{t}) - y\t \Sigma_{w,t}x\Big] \right| \leq c_{1} \psi_w \sqrt{T \log(\frac{2}{\delta'})},
    \end{split}\end{equation}
    which is equivalent to
    \begin{equation}\begin{split}
        5^{2a}\bbP\b{\left|\sum_{t=1}^T \Big[(y\t w_{t}) (x\t w_{t}) - y\t \Sigma_{w,t}x\Big]\right| > c_{1}\psi_w \sqrt{T \log(\frac{2}{\delta'})}} \leq 5^{2a}\delta'.
    \end{split}\end{equation}
    Letting $\delta = 5^{2a}\delta'$ gives
    \begin{equation}\begin{split}
        5^{2a}\bbP\b{\left|\sum_{t=1}^T \Big[(y\t w_{t}) (x\t w_{t}) - y\t \Sigma_{w,t}x\Big]\right| > c_{1}\psi_w \sqrt{2aT \log(\frac{10}{\delta})}} \leq \delta.
    \end{split}\end{equation}
    Substituting back into Equation \ref{eq:cov_dependent_1} gives
    \begin{equation}\begin{split}
        \bbP\b{\norm{\sum_{t=1}^T w_tw_t\t - \sum_{t=1}^T \Sigma_{w,t}} > 4c_{1}\psi_w \sqrt{2aT \log(\frac{10}{\delta})}} \leq \delta.
    \end{split}\end{equation}
    This completes the proof.

\end{proof}

\begin{lemma}\label{lem:gram_upper}
    Consider the same setting as \Cref{lem:1col_full}. For any $\delta$, with probability at least $1-\delta$, 
    \begin{equation}\begin{split}
        \norm{ \sum_{t=0}^T x_tx_t\t} \lesssim \frac{\psi_w\psi_A^2r}{1-\rho_A^2} T \log\frac{1}{\delta}.
    \end{split}\end{equation}
\end{lemma}

\begin{proof}
    Directly applying Proposition 8.4 in \cite{sarkar_near_2019} gives the following with probability at least $1-\delta$
    \begin{equation}\begin{split}
        \norm{\sum_{t=0}^T x_tx_t\t} \lesssim {}& \psi_w\tr\b{\sum_{t=0}^{T-1} \Gamma_t(A)}\log\frac{1}{\delta} \leq \psi_w r \norm{\sum_{t=0}^{T-1} \Gamma_t(A)} \log\frac{1}{\delta}.
    \end{split}\end{equation}
    From Assumption \ref{assmp:sys1}, we know that 
    \begin{equation}\begin{split}
        \norm{\Gamma_t(A)} = \norm{\sum_{\tau=0}^t A^{\tau}A^{\tau}\t} \preceq \sum_{\tau=0}^t \norm{A^\tau}^2 \leq 1 + \frac{\psi_A^2}{1-\rho_A^2} \lesssim \frac{\psi_A^2}{1-\rho_A^2}.
    \end{split}\end{equation}
    Substituting back gives
    \begin{equation}\begin{split}
        \norm{\sum_{t=0}^T x_tx_t\t} \lesssim \frac{\psi_w\psi_A^2r}{1-\rho_A^2}T\log\frac{1}{\delta}.
    \end{split}\end{equation}
\end{proof}

\subsubsection{Other Lemmas}
\begin{lemma}\label{lem:perturbation}
    Given perturbation matrices $\Delta_1, \Delta_2$ and p.d square matrix $M$ of the same dimension. Suppose they satisfy the following for some positive constant $a\in[0,1)$ and $b \geq 0$
    \begin{equation}\begin{split}
        \norm{M^{-\frac{1}{2}}\Delta_1} \leq a < \sqrt{\phi_M},\quad \norm{\Delta_2} \leq b.
    \end{split}\end{equation}
    Denote the SVD of the matrices as follows
    \begin{equation}\begin{split}
        M = \begin{bmatrix}
            U_1 & U_2
        \end{bmatrix} \begin{bmatrix}
            \Lambda_1 &\\
            & \Lambda_2
        \end{bmatrix}\begin{bmatrix}
            U_1\t\\
            U_2\t
        \end{bmatrix}, \quad \tilM \coloneqq M + \b{\Delta_1 + \Delta_1\t}+\Delta_2 = \begin{bmatrix}
            \tilU_1 & \tilU_2
        \end{bmatrix} \begin{bmatrix}
            \tilde{\Lambda}_1 &\\
            & \tilde{\Lambda}_2
        \end{bmatrix}\begin{bmatrix}
            \tilU_1\t\\
            \tilU_2\t
        \end{bmatrix},
    \end{split}\end{equation}
    and define the following constants
    \begin{equation}\begin{split}
        \psi_M = \sigma_1(M), \quad \phi_M = \sigma_{\min}(M), \quad 0< \delta_M \leq \sigma_{\min}(\Lambda_1) - \sigma_1(\Lambda_2).
    \end{split}\end{equation}
    Then we have
    \begin{equation}\begin{split}
        \norm{\tilU_2\t U_1} \leq \frac{a^2 + b}{\delta_M - 4a\sqrt{\psi_M} - 3a^2 - b} + \frac{a}{\sqrt{\phi_M}} \b{1 + \frac{2-a/\sqrt{\phi_M}}{1-a/\sqrt{\phi_M}}\frac{\sigma_1\b{\Lambda_2} + 2a\sqrt{\psi_M} + a^2}{\delta_M - 2a\sqrt{\psi_M} - a^2}}
    \end{split}\end{equation}
    In the case where $\sqrt{\phi_M} \geq Ca$, $\delta_M \geq C\b{a\sqrt{\psi_M} + b}$ for large enough $C$, the above result is further simplified to
    \begin{equation}\begin{split}
        \norm{\tilU_2\t U_1} \lesssim \frac{b}{\delta_M} + \frac{a\sigma_1(\Lambda_2)/\sqrt{\phi_M}}{\delta_M}+\frac{a^2\sqrt{\psi_M/\phi_M}}{\delta_M}.
    \end{split}\end{equation}
\end{lemma}

\begin{remark}
    This result is adapted from previous subspace relative perturbation results (Theorem 3.2 in \cite{li_relative_1998}). It provides tighter bound for matrices with smaller condition number, i.e. when $\phi_M$ is close to $\psi_M$, as compared to the standard Davis-Kahan $\sin\Theta$ theorem. To be more specific, Davis-Kahan gives the bound of order $\frac{a\sqrt{\psi_M}}{\delta_M}$. Here our result features the bound $\frac{a\sigma_1(\Lambda_2)/\sqrt{\phi_M}}{\delta_M}$. 
\end{remark}

\begin{proof}
    Define
    \begin{equation}\begin{split}
        \htM \coloneqq M + \b{\Delta_1+\Delta_1\t} + \Delta_1\t M^{-1}\Delta_1 = \begin{bmatrix}
            \htU_1  & \htU_2    
        \end{bmatrix} \begin{bmatrix}
            \wh{\Lambda}_1 &\\
            & \wh{\Lambda}_2
        \end{bmatrix}\begin{bmatrix}
            \htU_1\t\\
            \htU_2\t
        \end{bmatrix}.
    \end{split}\end{equation}
    From its definition, we immediately have
    \begin{equation}\begin{split}\label{eq:pert_1}
        \norm{\htM - M} \leq {}& 2\norm{\Delta_1} + \norm{\Delta_1\t M^{-1}\Delta_1}\\
        \overset{(i)}{\leq} {}& 2a\sqrt{\psi_M} + \norm{M^{-\frac{1}{2}}\Delta_1}^2\\
        \leq {}& 2a\sqrt{\psi_M} + a^2.
    \end{split}\end{equation}
    Here the first term in $(i)$ is because 
    \begin{equation}\begin{split}
        a \geq \norm{M^{-\frac{1}{2}}\Delta_1} \geq \sigma_{\min}\b{M^{-\frac{1}{2}}}\norm{\Delta_1} = \frac{1}{\sqrt{\psi_M}}\norm{\Delta_1}.
    \end{split}\end{equation}
    And the second term is because $\norm{A\t A} = \norm{A}^2$. Then from Theorem 1 in \cite{stewart_matrix_1990}, we know that
    \begin{equation}\begin{split}
        \sigma_1\b{\wh{\Lambda}_2} \leq \sigma_1\b{\Lambda_2} + \norm{\htM - M} \leq \sigma_1\b{\Lambda_2} + 2a\sqrt{\psi_M} + a^2.
    \end{split}\end{equation}
    which in turn gives 
    \begin{equation}\begin{split}\label{eq:pert_2}
        \wh{\delta} \coloneqq \sigma_{\min}(\Lambda_1) - \sigma_1(\wh{\Lambda}_2) \geq \delta_M - 2a\sqrt{\psi_M} - a^2.
    \end{split}\end{equation}
    For the rest of the proof, we upper bound $\norm{\htU_2\t U_1}$, $\norm{\tilU_2\t \htU_1}$ and then derives the desired result.
    To upper bound $\norm{\htU_2\t U_1}$, we let $D = M^{-1}\Delta_1$ with

    \begin{equation}\begin{split}
        {}& \norm{D} \leq \norm{M^{-\frac{1}{2}}}\norm{M^{-\frac{1}{2}}\Delta_1} \leq \frac{a}{\sqrt{\phi_M}} < 1,\\
        {}& \norm{\b{I+D}^{-1}} = \frac{1}{\sigma_{\min}\b{I+D}} \overset{(i)}{\leq} \frac{1}{1 - \sigma_1(D)} \leq \frac{1}{1- a/\sqrt{\phi_M}}.
    \end{split}\end{equation}
    Here $(i)$ is because $|\sigma_{\min}(I+D)-\sigma_{\min}(I)| \leq  \sigma_1(D)$ from Thereom 1 in \cite{stewart_matrix_1990}. And the definition of $D$ gives
    \begin{equation}\begin{split}
        \htM = \b{I + \Delta_1\t M^{-1}}M\b{I + M^{-1}\Delta_1} = \b{I+D\t}M\b{I+D}.
    \end{split}\end{equation}
    Now from Equation \ref{eq:pert_2} we apply Theorem 3.2 in \cite{li_relative_1998} which gives
    \begin{equation}\begin{split}
        \norm{\htU_2\t U_1} \leq {}& \norm{I-(I+D\t)} + \b{\sigma_1\b{\Lambda_2} + 2a\sqrt{\psi_M} + a^2}\frac{\norm{(I+D)\t - (I+D)^{-1}}}{\wh{\delta}}\\
        \leq {}& \frac{a}{\sqrt{\phi_M}} + \frac{\sigma_1\b{\Lambda_2} + 2a\sqrt{\psi_M} + a^2}{\wh{\delta}}\b{\norm{(I+D)\t-I}+\norm{I-\b{I+D}^{-1}}}\\
        \leq {}& \frac{a}{\sqrt{\phi_M}} + \frac{\sigma_1\b{\Lambda_2} + 2a\sqrt{\psi_M} + a^2}{\wh{\delta}}\b{\frac{a}{\sqrt{\phi_M}}+\norm{\b{I+D}^{-1}}\norm{I+D-I}}\\
        \leq {}& \frac{a}{\sqrt{\phi_M}} + \frac{\sigma_1\b{\Lambda_2} + 2a\sqrt{\psi_M} + a^2}{\wh{\delta}}\b{\frac{a}{\sqrt{\phi_M}}+\frac{1}{1-a/\sqrt{\phi_M}}\frac{a}{\sqrt{\phi_M}}}\\
        = {}& \frac{a}{\sqrt{\phi_M}} \b{1 + \frac{2-a/\sqrt{\phi_M}}{1-a/\sqrt{\phi_M}}\frac{\sigma_1\b{\Lambda_2} + 2a\sqrt{\psi_M} + a^2}{\wh{\delta}}}.
    \end{split}\end{equation}

    We now upper bound $\norm{\tilU_2\t \htU_1}$. We start from upper bounding $\norm{\htM - \tilM}$
    \begin{equation}\begin{split}
        \norm{\htM - \tilM} = \norm{\Delta_1\t M^{-1}\Delta_1 - \Delta_2} \leq \norm{\Delta_1\t M^{-1}\Delta_1} + \norm{\Delta_2} \leq a^2 + b.
    \end{split}\end{equation}
    Moreover, we know that 
    \begin{equation}\begin{split}
        {}& \sigma_{\min}(\wh{\Lambda}_1) - \sigma_1(\tilde{\Lambda}_2)\\
        ={}&  \sigma_{\min}(\wh{\Lambda}_1) - \sigma_1(\wh{\Lambda}_2) + \sigma_1(\wh{\Lambda}_2) - \sigma_1(\tilde{\Lambda}_2)\\
        \geq {}& \sigma_{\min}(\wh{\Lambda}_1) - \sigma_1(\wh{\Lambda}_2) - a^2 - b\\
        = {}& \sigma_{\min}(\wh{\Lambda}_1) - \sigma_{\min}(\Lambda_1) + \sigma_{\min}(\Lambda_1) - \sigma_1(\Lambda_2) + \sigma_1(\Lambda_2) - \sigma_1(\wh{\Lambda}_2) - a^2 - b\\
        \geq {}& -2\norm{\htM-M} - a^2 - b + \delta_M\\
        \geq {}& - 4a\sqrt{\psi_M} - 2a^2 - a^2 - b + \delta_M.
    \end{split}\end{equation}
    Here the first and second inequalities are due to  \citet[Theorem 1]{stewart_matrix_1990} . Directly applying the Davis-Kahan $\sin\theta$ Theorem in \cite{davis_rotation_1970} gives
    \begin{equation}\begin{split}
        \norm{\tilU_2\t \htU_1} \leq \frac{\norm{\htM-\tilM}}{\sigma_{\min}(\wh{\Lambda}_1) - \sigma_1(\tilde{\Lambda}_2)} \leq \frac{a^2 + b}{- 4a\sqrt{\psi_M} - 3a^2 - b + \delta_M}.
    \end{split}\end{equation}

    Finally, we conclude that
    \begin{equation}\begin{split}
        \norm{\tilU_2\t U_1} = {}& \norm{\tilU_2\t \b{\htU_1\htU_1\t+\htU_2\htU_2\t}U_1} \leq \norm{\tilU_2\t\htU_1} + \norm{\htU_2\t U_1}\\
        \leq {}& \frac{a^2 + b}{\delta_M - 4a\sqrt{\psi_M} - 3a^2 - b} + \frac{a}{\sqrt{\phi_M}} \b{1 + \frac{2-a/\sqrt{\phi_M}}{1-a/\sqrt{\phi_M}}\frac{\sigma_1\b{\Lambda_2} + 2a\sqrt{\psi_M} + a^2}{\wh{\delta}}}\\
        \leq {}& \frac{a^2 + b}{\delta_M - 4a\sqrt{\psi_M} - 3a^2 - b} + \frac{a}{\sqrt{\phi_M}} \b{1 + \frac{2-a/\sqrt{\phi_M}}{1-a/\sqrt{\phi_M}}\frac{\sigma_1\b{\Lambda_2} + 2a\sqrt{\psi_M} + a^2}{\delta_M - 2a\sqrt{\psi_M} - a^2}}.
    \end{split}\end{equation}
\end{proof}

\begin{lemma}\label{lem:gen_norm2Net}
    For any positive integer $a, b$, let $M\in\bbR^{a\times b}$ be any random matrix. There exist  $x\in\mathbb{S}^{b-1}$ and $y\in\mathbb{S}^{a-1}$ such that for all $\varepsilon < 1$
    \begin{equation}\begin{split}
        \bbP(\norm{M} > z) \leq \b{1+\frac{2}{\varepsilon}}^{a+b} \bbP\b{\left|{y\t Mx}\right| > (1-\varepsilon)^2z}.
    \end{split}\end{equation}
    Here, $\mathbb{S}^{b-1}$ denotes the unit $(b-1)$-sphere embedded in $\mathbb{R}^b$.
\end{lemma}

\begin{proof}
    From Proposition 8.1 of \cite{sarkar_near_2019}, we know that there exist an $x\in\mathbb{S}^{b-1}$ s.t. the following holds for any $\varepsilon < 1$
    \begin{equation}\begin{split}
        \bbP(\norm{M} > z) \leq \b{1+\frac{2}{\varepsilon}}^b \bbP\b{\norm{M x} > (1-\varepsilon)z} = \b{1+\frac{2}{\varepsilon}}^b \bbP\b{\norm{x\t M\t} > (1-\varepsilon)z}.
    \end{split}\end{equation}
    Applying the proposition again on $x\t M\t$ gives the following for some vector $y\in\mathbb{S}^{a-1}$
    \begin{equation}\begin{split}
        \bbP\b{\norm{x\t M\t} > (1-\varepsilon)z} \leq \b{1+\frac{2}{\varepsilon}}^{a} \bbP\b{\left|{x\t M y}\right| > (1-\varepsilon)^2z}.
    \end{split}\end{equation}
    Therefore, we have 
    \begin{equation}\begin{split}
        \bbP(\norm{M} > z) \leq \b{1+\frac{2}{\varepsilon}}^{a+b} \bbP\b{\left|{x\t M\t y}\right| > (1-\varepsilon)^2z}
    \end{split}\end{equation}
\end{proof}

\section[]{Lower Bounds for \ac{hdsysid} --- Proof of \Cref{thm:1lowerbound_full}}\label{sec:lowerbound}

An essential part for the proof is the Birge's Inequality stated below.
\begin{theorem}\label{thm:birge}
    Let $\{\bbP_i\}_{i\in[N]}$ be probability distributions on $(\Omega, \calF)$. Let $\{F_i\}_{i\in[N]}\in\calF$ be \textbf{pairwise disjoint events}. If $\delta = \min_{i\in[N]}\bbP_i(F_i) \geq \frac{1}{2}$ (\textit{the success rate}), then
    \begin{equation}\begin{split}
        (2\delta-1)\log\b{\frac{\delta}{1-\delta}(N-1)} \leq \frac{1}{N-1}\sum_{i\in[N]} \kl(\bbP_i||\bbP_1).
    \end{split}\end{equation}
\end{theorem}


\begin{proof}
    From the Birge's inequality \cite{boucheron_concentration_2013}, we directly have
    \begin{equation}\begin{split}
        \delta\log\b{\delta/\frac{1-\delta}{N-1}} + (1-\delta)\log\b{(1-\delta)/(1-\frac{1-\delta}{N-1})} \leq \frac{1}{N-1}\sum_{i\in[N]} \kl(\bbP_i||\bbP_1),
    \end{split}\end{equation}
    On the other hand, the following holds for any $N \geq 3$
    \begin{equation}\begin{split}
        {}& \delta\log\b{\delta/\frac{1-\delta}{N-1}} + (1-\delta)\log\b{(1-\delta)/(1-\frac{1-\delta}{N-1})}\\
        \overset{(i)}{\geq} {}& \delta\log\b{\delta/\frac{1-\delta}{N-1}} + (1-\delta)\log\b{(1-\delta)/((N-1)\delta)}\\
        = {}& (2\delta-1)\log\b{\frac{\delta}{1-\delta}(N-1)}
    \end{split}\end{equation}
    Here $(i)$ holds because $1 - \frac{1-\delta}{N-1} \leq 1 \leq 2\delta \leq (N-1)\delta$. 
    For $N=2$, the above inequality naturally holds.
    Combining the two cases finishes the proof.
\end{proof}

Now we prove \Cref{thm:1lowerbound} by proving the following more general version that holds for multiple trajectories.

\begin{theorem}\label{thm:1lowerbound_full}
    Suppose $n \geq m \geq r$ with $n\geq 2$, $K\geq 1$, and choose positive scalars $\delta \leq \frac{1}{2}$. 
    Consider the class of minimal systems $\calM=(r,n,m,A,B,C,\Sigma_w,\Sigma_\eta)$ with different $A$,$B$,$C$ matrices. All other parameters are fixed and known. Moreover, $r<n$ and $\Sigma_\eta$ is positive definite. For every $k\in[K]$, let $ \calD_k = \{y_{k,t}\}_{t=0}^{T_k}\cup\{u_{k,t}\}_{t=0}^{T_k-1}\cup\text{\{all known parameters\}}$ denote the associated single trajectory dataset. 
    Here the inputs $u_{k,t}$ satisfy: 1). $u_{k,t}$ is sampled independently;2). $\bbE(u_{k,t}) = 0$. Consider any estimator $\htf$ mapping the datasets $\calD\coloneqq\cup_{k\in[K]}\calD_k$ to $(\htA(\calD), \htB(\calD),\htC(\calD))\in\bbR^{r\times r}\times \bbR^{r\times m}\times \bbR^{n\times r}$. If 
    \begin{equation}
        \calT \coloneqq \sum_{k\in[K]}T_k < \frac{\phi_\eta(1-2\delta)\log1.4}{50\b{\psi_w + \psi_u}}\cdot\frac{n+\log\frac{1}{2\delta}}{\epsilon^2},    
    \end{equation}
    there exists a system $\calM_0=(r,n,m,A_0,B_0,C_0,\Sigma_w,\Sigma_\eta)$ with dataset $\calD$ such that
    \begin{equation}\begin{split}
        \bbP {}& 
        \left\{\norm{\htC\b{\calD}\htB\b{\calD} - C_0B_0} \geq \epsilon\right\} \geq \delta.
    \end{split}\end{equation}
    Here $\bbP$ denotes the distribution of $\calD$ generated by system $\calM_0$. Related constants are defined as follows
    \begin{equation*}\begin{split}
        \phi_\eta &\coloneqq \sigma_{\min}(\Sigma_\eta), \quad \psi_w \coloneqq \sigma_{1}(\Sigma_w), \\\quad \psi_u &\coloneqq \max_{k,t} \sigma_{\max}\b{\bbE(u_{k,t}u_{k,t}\t)}.    \end{split}\end{equation*}
\end{theorem}

\begin{proof}
    \textbf{Step 1: We first construct a set of candidate minimal systems with slightly different observer column spaces.}
    Fix a constant $\epsilon < 0.1$ and a $\frac{1}{2}$-packing of the intersection of the unit ball in $\bbR^n$, denoted by $\tilde{\calP}$. Now let $\calP = \tilde{\calP}\backslash \{(-1, 0,\dots, 0)\t\}$. We know $|\calP| \geq 2^n-1$ and we let $\calP = \{p_i\}_{i=1}^{|\calP|}$.
    Now consider the following set of matrices $\{A_i, B_i, C_i\}_{i=1}^{|\calP|}$ where
    \begin{equation}\begin{split}
        {}& A_i = \frac{1}{2} I \in\bbR^{r\times r}, \quad B_i = \sum_{l\in[r]}E_{l,l} \in\bbR^{r\times m},\quad C_i = 4\epsilon p_ie_1\t + \sum_{l\in[r]}E_{l,l} \in \bbR^{n\times r},
    \end{split}\end{equation}
    where $E_{l,l}$ is a matrix (of appropriate dimensions) with $1$ in the $(l,l)$ entry and $0$ elsewhere, $e_l$ is a vector (of appropriate dimensions) with $1$ in the $l$ entry and $0$ elsewhere.
    It is clear that all systems are minimal, since $B_i$ are full row rank and $C_i$ are full column rank. Moreover, in all systems, only matrices $C_i$ are different for each $i$, and $A_i$ and $B_i$ remain the same. We will see that this ``high-dimensional perturbation'' only in each observer matrix $C_i$ is enough for creating a worst case scenario.

    Now consider any estimator $\htf$ mapping the dataset $\calD$ to $\left(\htA(\calD), \htB(\calD),\htC(\calD)\right)\in\bbR^{r\times r}\times \bbR^{r\times m}\times \bbR^{n\times r}$.
    We define events $\{F_i\}_{i\in[|\calP|]}$ as follows:
    \begin{equation}\begin{split}
        F_i = {}& \left\{\calD : \norm{\htC\b{\calD}\htB\b{\calD} - C_iB_i} < \epsilon \right\}
    \end{split}\end{equation}
    Under our construction, for any $i, j\in[|\calP|]$, we have
    \begin{equation}\begin{split}
        \norm{C_iB_i - C_jB_j} = \norm{(C_i - C_j)B_j} = 4\epsilon\norm{(p_i-p_j)e_1\t} = 4\epsilon\norm{p_i-p_j}_2 \geq 2\epsilon.
    \end{split}\end{equation}
    Therefore, the following holds and all events $\{F_i\}_{i\in[|\calP|]}$ are pairwise disjoint.
    \begin{equation}\begin{split}
        {}& \norm{\htC\b{\calD}\htB\b{\calD} - C_iB_i} + \norm{\htC\b{\calD}\htB\b{\calD} - C_jB_j}
        \geq \norm{C_iB_i - C_jB_j}
        \geq 2\epsilon.
    \end{split}\end{equation}

    \textbf{Step 2: We now bound the KL-divergence of dataset distribution coming from any pair of systems we constructed.}
    For any $k\in[K]$, let $u^k$, $x^k$, $y^k$ denote the corresponding trajectories of the inputs, latent variables and observations. Moreover, let $u = \cup_{k\in[K]}u^k$, $y = \cup_{k\in[K]}y^k$ and $x = \cup_{k\in[K]}x^k$ denote the collection of all trajectories. For any $i\in[|\calP|]$, let $\bbP_{i}^{UXY}$ denote the joint distribution of $\{u,x,y\}$ generated by system $\calM_i = (r,n,m,A_i,B_i,C_i,\Sigma_w,\Sigma_\eta)$. Let $\bbP_{i}^{UY}$, $\bbP_{i}^{UX}$ be the marginal distribution of $\{y,u\}$ and $\{x,u\}$. Similarly, define $\bbP_{i}^{X|U}$, $\bbP_{i}^{X|UY}$ and $\bbP_{i}^{Y|UX}$ to be the conditional distributions. 
    For the rest of the proof, we apply Theorem \ref{thm:birge} on the events $\{F_i\}_{i\in[|\calP|]}$ and distributions $\{\bbP_i^{UY}\}_{i\in[|\calP|]}$.

    We aim to upper bound the KL-divergence between distributions $\{\bbP_i^{UY}\}_{i=1}^{|\calP|}$. To do this, we start from showing that $\kl(\bbP_i^{UY}||\bbP_j^{UY}) \leq \kl(\bbP_i^{UXY}||\bbP_j^{UXY})$. 
    \begin{equation}\begin{split}
        \kl(\bbP_i^{UXY}||\bbP_j^{UXY}) ={}&  \int \bbP_i^{UXY}(u,x, y)\log\frac{\bbP_i^{UXY}(u,x,y)}{\bbP_j^{UXY}(u,x,y)}\d u\d x\d y\\
        = {}& \int \bbP_i^{UY}(u,y)\bbP_i^{X|U=u,Y=y}(x)\log\frac{\bbP_i^{UY}(u,y)\bbP_i^{X|U=u,Y=y}(x)}{\bbP_j^{UY}(u,y)\bbP_j^{X|U=u,Y=y}(x)}\d u\d x\d y\\
        = {}& \int \bbP_i^{UY}(u,y)\bbP_i^{X|U=u,Y=y}(x)\log\frac{\bbP_i^{UY}(u,y)}{\bbP_j^{UY}(u,y)}\d u\d x\d y\\
        {}& \hspace{2em} + \int \bbP_i^{UY}(u,y)\bbP_i^{X|U=u,Y=y}(x)\log\frac{\bbP_i^{X|U=u,Y=y}(x)}{\bbP_j^{X|U=u,Y=y}(x)} \d u\d x\d y\\
        = {}& \int \bbP_i^{UY}(u,y)\b{\int \bbP_i^{X|U=u,Y=y}(x)\d x}\log\frac{\bbP_i^{UY}(u,y)}{\bbP_j^{UY}(u,y)}\d u\d y\\
        {}& \hspace{2em} + \int \bbP_i^{UY}(u,y)\kl\b{\bbP_i^{X|U=u,Y=y}||\bbP_j^{X|U=u,Y=y}}\d u\d y\\
        \overset{(i)}{\geq} {}& \int \bbP_i^{UY}(u,y)\log\frac{\bbP_i^{UY}(u,y)}{\bbP_j^{UY}(u,y)}\d u\d y\\
        = {}& \kl(\bbP_i^{UY}||\bbP_j^{UY}).
    \end{split}\end{equation}
    Here $(i)$ holds because KL-divergence (the second term) is non-negative. Therefore, we only need to upper bound $\kl(\bbP_i^{UXY}||\bbP_j^{UXY})$. Notice that 

    \begin{equation}\begin{split}
        {}& \log\frac{\bbP_i^{UXY}(u,x,y)}{\bbP_j^{UXY}(u,x,y)} = \log\frac{\bbP_i^U(u)\bbP_i^{X|U=u}(x)\bbP_i^{Y|U=u,X=x}(y)}{\bbP_j^U(u)\bbP_j^{X|U=u}(x)\bbP_j^{Y|U=u,X=x}(y)}\\
        = {}& \log\frac{\prod_{k=1,t=0}^{K,T_k-1}\bbP_i(u_{k,t})}{\prod_{k=1,t=0}^{K,T_k-1}\bbP_j(u_{k,t})} + \log\frac{\prod_{k=1,t=1}^{K,T_k}\bbP_i(x_{k,t}|x_{k,0:t-1},u_{k,0:t-1})}{\prod_{k=1,t=1}^{K,T_k}\bbP_j(x_{k,t}|x_{k,0:t-1},u_{k,0:t-1})}\\
        {}& \hspace{2em} + \log\frac{\prod_{k=1,t=0}^{K,T_k}\bbP_i(y_{k,t}|x_{k,0:t},u_{k,0:t})}{\prod_{k=1,t=0}^{K,T_k}\bbP_j(y_{k,t}|x_{k,0:t},u_{k,0:t})}\\
        = {}& \sum_{k=1,t=0}^{K,T_k-1}\log\frac{\bbP_i(u_{k,t})}{\bbP_j(u_{k,t})} + \sum_{k=1,t=1}^{K,T_k}\log\frac{\bbP_i(x_{k,t}|x_{k,t-1},u_{k,t-1})}{\bbP_j(x_{k,t}|x_{k,t-1},u_{k,t-1})}\\
        {}& \hspace{2em} + \sum_{k=1,t=0}^{K,T_k}\log\frac{\bbP_i(y_{k,t}|x_{k,t},u_{k,t})}{\bbP_j(y_{k,t}|x_{k,t},u_{k,t})}\\
        = {}& \sum_{k=1,t=0}^{K,T_k}\log\frac{\bbP_i(y_{k,t}|x_{k,t},u_{k,t})}{\bbP_j(y_{k,t}|x_{k,t},u_{k,t})}.\\
    \end{split}\end{equation}
    Here the last line is because the distribution of $u$ and $x$ doesn't depend on the $i$ or $j$. This is holds for $x$ because for any $i$, $A_i$ and $B_i$ take the same value.
    Therefore, 
    \begin{equation}\begin{split}\label{eq:kl_obs}
        {}& \kl(\bbP_i^{UXY}||\bbP_j^{UXY})\\
        = {}& \int \bbP_i^{UXY}(u,x, y)\log\frac{\bbP_i^{UXY}(u,x,y)}{\bbP_j^{UXY}(u,x,y)}\d x\d y\\
        = {}& \int \bbP_i^{UXY}(u,x, y)\b{\sum_{k=1,t=0}^{K,T_k}\log\frac{\bbP_i(y_{k,t}|x_{k,t},u_{k,t})}{\bbP_j(y_{k,t}|x_{k,t},u_{k,t})}} \d u\d x\d y\\
        = {}& \sum_{k=1,t=0}^{K,T_k}\bbE_i\b{\log\frac{\bbP_i(y_{k,t}|x_{k,t},u_{k,t})}{\bbP_j(y_{k,t}|x_{k,t},u_{k,t})}}\\
        \overset{(i)}{=} {}& \sum_{k=1,t=0}^{K,T_k} \bbE_i \b{\kl\b{\bbP_i^Y(\cdot|x_{k,t},u_{k,t})||\bbP_j^Y(\cdot|x_{k,t},u_{k,t})}}\\
        \overset{(ii)}{=} {}& \frac{1}{2}\sum_{k=1,t=0}^{K,T_k} \bbE_i \b{x_{k,t}\t(C_i-C_j)\t\Sigma_\eta^{-1}(C_i-C_j)x_{k,t}}
    \end{split}\end{equation}
    Here in $(i)$ follows from the definition of KL-divergence and by taking the expectation of $x_t$ (or $y_t$ in the second term). And $(ii)$ follows from the fact that $P_i^X = P_j^X$ (this is because $A_i=A_j$ and $B_i=B_j$) and from the KL-divergence between two gaussians. We further simplify the expression as follows
    \begin{equation}\begin{split}
        \kl(\bbP_i^{UXY}||\bbP_j^{UXY}) = {}& \frac{1}{2}\sum_{k=1,t=0}^{K,T_k} \bbE_i \b{x_{k,t}\t(C_i-C_j)\t\Sigma_\eta^{-1}(C_i-C_j)x_{k,t}}\\
        \leq {}& \frac{1}{2\phi_\eta}\sum_{k=1,t=0}^{K,T_k} \bbE_i \b{x_{k,t}\t(C_i-C_j)\t(C_i-C_j)x_{k,t}}\\
        = {}& \frac{1}{2\phi_\eta} \sum_{k=1,t=0}^{K,T_k} \tr\b{(C_i-C_j)\bbE_i\b{x_{k,t}x_{k,t}\t}(C_i-C_j)\t}\\
    \end{split}\end{equation}
    Now we consider term $\bbE_i\b{x_{k,t}x_{k,t}\t}$. Notice that $x_{k,t} = \sum_{\tau=0}^{t-1}A_i^\tau \b{B_iu_{k,t-1-\tau}+w_{k,t-1-\tau}}$, we have
    \begin{equation}\begin{split}\label{eq:gram_obs}
        \bbE_i(x_{k,t}x_{k,t}\t) ={}&  \bbE_i\bs{\sum_{\tau=0}^{t-1}A_i^{\tau}\b{B_iu_{k,t-1-\tau}+w_{k,t-1-\tau}}\sum_{\tau=0}^{t-1}\b{B_iu_{k,t-1-\tau}+w_{k,t-1-\tau}}\t (A_i^{\tau})\t}\\
        = {}& \sum_{\tau_1,\tau_2=0}^{t-1}A_i^{\tau_1}\bbE_i\bs{\b{B_iu_{k,t-1-\tau_1}+w_{k,t-1-\tau_1}}\b{B_iu_{k,t-1-\tau_2}+w_{k,t-1-\tau_2}}\t} (A_i^{\tau_2})\t\\
        \overset{(i)}{=} {}& \sum_{\tau=0}^{t-1}A_i^{\tau}\bbE_i\bs{\b{B_iu_{k,t-1-\tau}+w_{k,t-1-\tau}}\b{B_iu_{k,t-1-\tau}+w_{k,t-1-\tau}}\t} (A_i^{\tau})\t\\
        = {}& \sum_{\tau=0}^{t-1}A_i^\tau \b{B_i\bbE(u_{k,t}u_{k,t}\t)B_i\t+\Sigma_w} (A_i^\tau)\t\\
        \preceq {}& \b{\psi_w + \norm{B}^2\psi_u}\Gamma_{t-1}(A_i)\\
        = {}& \b{\psi_w + \psi_u}\Gamma_{t-1}(A_i).
    \end{split}\end{equation}
    Here $(i)$ is because $u_{\tau_1}, w_{\tau_1}$ are zero-mean and independent of $u_{\tau_2},w_{\tau_2}$ for $\forall \tau_1\neq \tau_2$. 
    Substituting back gives
    \begin{equation}\begin{split}
        {}& \kl(\bbP_i^{UXY}||\bbP_j^{UXY})\\
        \leq {}& \frac{1}{2\phi_\eta} \sum_{k=1,t=0}^{K,T_k} \tr\b{(C_i-C_j)\bbE_i\b{x_{k,t}x_{k,t}\t}(C_i-C_j)\t}\\
        \leq {}& \frac{\psi_w+\psi_u}{2\phi_\eta} \sum_{k=1,t=1}^{K,T_k}  \tr\b{(C_i-C_j)\Gamma_{t-1}(A_i)(C_i-C_j)\t}.
    \end{split}\end{equation}
    Here the last line is because $x_0=0$. Since $A = \frac{1}{2}I$, we conclude that 
    \begin{equation}\begin{split}
        \Gamma_{t}(A_i) = \sum_{\tau=0}^{t}A_i^t (A_i^t)\t \preceq \sum_{\tau=0}^t \b{\frac{1}{2}}^{2t}I \preceq \frac{4}{3}I.
    \end{split}\end{equation}
    Substituting it back gives
    \begin{equation}\begin{split}
        \kl(\bbP_i^{UXY}||\bbP_j^{UXY}) \leq {}& \frac{2(\psi_w + \psi_u)}{3\phi_\eta} \sum_{k=1,t=0}^{K,T_k-1} \tr\b{(C_i-C_j)(C_i-C_j)\t}\\
    \end{split}\end{equation}
    From the definitions of $C_i$ and $C_j$, we notice that 
    \begin{equation}\begin{split}
        C_i - C_j ={}&  4\epsilon (p_i-p_j)e_1\t.
    \end{split}\end{equation}
    Then 
    \begin{equation}\begin{split}
        \kl(\bbP_i^{UXY}||\bbP_j^{UXY}) \leq {}& \frac{32\b{\psi_w + \psi_u}\epsilon^2}{3\phi_\eta} \sum_{k=1,t=0}^{K,T_k-1}  \tr\b{\b{p_ie_1\t-p_je_1\t}\b{p_ie_1\t-p_je_1\t}\t}\\
        = {}& \frac{32\b{\psi_w + \psi_u}\epsilon^2}{3\phi_\eta} \sum_{k=1,t=0}^{K,T_k-1} \norm{p_ie_1\t-p_je_1\t}_{\f}^2\\
        = {}& \frac{32\b{\psi_w + \psi_u}\epsilon^2}{3\phi_\eta} \sum_{k=1,t=0}^{K,T_k-1} \norm{p_i-p_j}_2^2\\
    \end{split}\end{equation}
    Since $p_i$ and $p_j$ lie in the unit ball, we know that $\norm{p_i-p_j} \leq 2$. Therefore,
    \begin{equation}\begin{split}\label{eq:lowerbound_2_obs}
        \kl(\bbP_i^{UXY}||\bbP_j^{UXY}) \leq {}& \frac{128\b{\psi_w + \psi_u}\epsilon^2}{3\phi_\eta} \sum_{k=1,t=0}^{K,T_k-1}1 = \frac{128\b{\psi_w + \psi_u}\epsilon^2}{3\phi_\eta} \sum_{k=1}^{K} T_k\\
        \leq {}& \frac{50\b{\psi_w + \psi_u}\epsilon^2}{\phi_\eta} \calT.
    \end{split}\end{equation}
    Here the last line is due to the definition of $\calT$.

    \textbf{Step 3: Now we invoke \Cref{thm:birge} and show that learning at least one of those minimal realizations is hard.}
    Now applying Theorem \ref{thm:birge} on distributions $\{\bbP_i^{UY}\}_{i\in[|\calP|]}$ and events $\{F_i\}_{i\in[|\calP|]}$ gives the following result. If $\min_{i\in[|\calP|]}\bbP_i^{UY}(F_i) \geq 1-\delta$, then the following holds
    \begin{equation}\begin{split}
        (1-2\delta)\log\b{\frac{1-\delta}{\delta}(|\calP|-1)} \leq \frac{1}{N-1}\sum_{i\in[|\calP|]} \kl(\bbP_i^{UY}||\bbP_1^{UY}) \leq \sup_{i\in[\calP]}\kl(\bbP_i^{UY}||\bbP_1^{UY}).
    \end{split}\end{equation}
    For the LHS, since $|\calP| \geq 2^n-1$, we have $|\calP|-2 \geq 1.4^n$. Therefore,
    \begin{equation}\begin{split}
        \text{LHS} = {}& (1-2\delta)\log\b{\frac{1-\delta}{\delta}(|\calP|-1)}\\
        \geq {}& (1-2\delta)\log\b{\frac{1-\delta}{\delta}1.4^n}\\
        \geq {}& (1-2\delta)n\log1.4 + (1-2\delta)\log\b{\frac{1-\delta}{\delta}}\\
        \geq {}& (1-2\delta)n\log1.4 + (1-2\delta)\log\frac{1}{2\delta}\\
    \end{split}\end{equation}
    For the RHS, we get the following from Equation \ref{eq:lowerbound_2_obs}
    \begin{equation}\begin{split}
        \text{RHS} \leq \sup_{i\in[\calP]}\kl(\bbP_i^{UXY}||\bbP_1^{UXY}) \leq \frac{50\b{\psi_w + \psi_u}\epsilon^2}{\phi_\eta}\calT.
    \end{split}\end{equation}
    Therefore, we have the following for any $\epsilon$
    \begin{equation}\begin{split}
        \calT \geq \frac{(1-2\delta)\phi_\eta\log1.4}{50\b{\psi_w + \psi_u}}\cdot\frac{n+\log\frac{1}{2\delta}}{\epsilon^2} \geq \frac{\phi_\eta(1-2\delta)\log1.4}{50\b{\psi_w + \psi_u}}\cdot\frac{n+\log\frac{1}{2\delta}}{\epsilon^2}.
    \end{split}\end{equation}
    Namely, if $\min_{i\in[|\calP|]}\bbP_i^{UY}(F_i) \geq 1-\delta$, then $\calT \geq \frac{\phi_\eta(1-2\delta)\log1.4}{50\b{\psi_w + \psi_u}}\cdot\frac{n+\log\frac{1}{2\delta}}{\epsilon^2}$. In other words, if $\calT < \frac{\phi_\eta(1-2\delta)\log1.4}{50\b{\psi_w + \psi_u}}\cdot\frac{n+\log\frac{1}{2\delta}}{\epsilon^2}$, there exists at least one $i\in[|\calP|]$ such that $\bbP_i^{UY}(F_i) < 1-\delta$. For this $i$, we then know that
    \begin{equation}
        \bbP_i^{UY}(F_i^c) \geq \delta.
    \end{equation}
    where $F_i^c$ is the complement event of $F_i$. This finishes the proof.

\end{proof}

\begin{corollary}\label{cor:thm1_1}
    Consider the same setting as \Cref{thm:1lowerbound}. Moreover, assume $\norm{\htB(\calD)}, \norm{\htC(\calD)} \leq \bar{\psi}$, i.e., the estimator outputs bounded approximations. Let $\epsilon$ be any scalar such that $\epsilon\leq \min\{0.2/\bar{\psi}, \norm{\htB(\calD)}\cdot\norm{\htC(\calD)}/\bar{\psi}\}$. Then if 
    \begin{equation}\begin{split}
        \calT < \frac{\phi_\eta(1-2\delta)\log1.4}{450\bar{\psi}^2\b{\psi_w + \psi_u}}\cdot\frac{n+\log\frac{1}{2\delta}}{\epsilon^2}
    \end{split}\end{equation}
    there exists a minimal system $\calM_0=(r,n,m,A_0,B_0,C_0,\Sigma_w,\Sigma_\eta)$ with dataset $\calD$ such that
    \begin{equation}\begin{split}
        \bbP {}& 
        \left\{\max\left\{\norm{\htC\b{\calD} - C_0S}, \norm{\htB\b{\calD} - S^{-1}B_0}\right\} \geq \epsilon \right\} \geq \delta.
    \end{split}\end{equation}
\end{corollary}

\begin{proof}
    Let $\epsilon' = 3\bar{\psi}\epsilon$. For all $\epsilon \leq 0.2/\bar{\psi}$, we know that $\epsilon'\leq 0.6$. Therefore, applying Theorem \ref{thm:1lowerbound} with $\epsilon'$ gives the following result. If $T < \frac{\phi_\eta(1-2\delta)\log1.4}{450\b{\psi_w + \psi_u}\bar{\psi}^2}\cdot\frac{n+\log\frac{1}{2\delta}}{\epsilon^2}$, there exists $\calM_0=(r,n,m,A_0,B_0,C_0,\Sigma_w,\Sigma_\eta)$ such that 
    \begin{equation}\begin{split}
        \bbP {}& 
        \left\{\norm{\htC\b{\calD}\htB\b{\calD} - C_0B_0} \geq 3\bar{\psi}\epsilon\right\} \geq \delta.
    \end{split}\end{equation}
    Since $\epsilon \leq \frac{\norm{\htB(\calD)}\norm{\htC(\calD)}}{\bar{\psi}}$, we know that $3\bar{\psi}\epsilon \leq 3\norm{\htB(\calD)}\norm{\htC(\calD)}$. Therefore, applying Lemma \ref{lem:thm1_1} gives
    \begin{equation}\begin{split}
        {}& \norm{\htC\b{\calD}\htB\b{\calD} - C_0B_0} \geq 3\bar{\psi}\epsilon\\
        \Rightarrow {}& \norm{\htC(\calD) - C_0S} \geq \frac{3\bar{\psi}\epsilon}{3\norm{\htB}} \geq \epsilon \quad \text{or} \quad \norm{\htB(\calD) - S^{-1}B_0} \geq \frac{\epsilon}{3\norm{\htC}} \geq \epsilon.
    \end{split}\end{equation}
    Therefore, 
    \begin{equation}\begin{split}
        {}& \bbP
        \left\{\max\left\{\norm{\htC\b{\calD} - C_0S}, \norm{\htB\b{\calD} - S^{-1}B_0}\right\} \geq \epsilon \right\}\\
        \geq {}& \bbP
        \left\{\norm{\htC\b{\calD}\htB\b{\calD} - C_0B_0} \geq 3\bar{\psi}\epsilon\right\} \geq \delta.\\
        \geq {}& \delta.
    \end{split}\end{equation}
\end{proof}

\begin{lemma}\label{lem:thm1_1}
    Fix any invertible matrix $S$. Suppose $\norm{\htC\htB - CB} \geq \epsilon$ for $\epsilon \leq 3\norm{\htB}\norm{\htC}$. Then either
    \begin{equation}\begin{split}
        \norm{\htC - CS} \geq \frac{\epsilon}{3\norm{\htB}} \quad \text{ or } \quad \norm{\htB - S^{-1}B} \geq \frac{\epsilon}{3\norm{\htC}}.
    \end{split}\end{equation}
\end{lemma}

\begin{remark}
    In other words, for any similarity transformation, we can not learn either $\tilde{B}$ or $\tilde{C}$ well.     
\end{remark}

\begin{proof}
        
    We show this claim by contradiction. For any invertible matrix $S$, suppose
    \begin{equation}\begin{split}
        \norm{\htC - CS} < \frac{\epsilon}{3\norm{\htB}} \quad \text{ and } \quad \norm{\htB - S^{-1}B} < \frac{\epsilon}{3\norm{\htC}}.
    \end{split}\end{equation}
    Then for any $\epsilon \leq \norm{B}\norm{C}$, we know that
    \begin{equation}\begin{split}
        {}& \norm{\htC\htB - CB}\\
        \leq {}& \norm{\htC\htB - CS\htB} + \norm{CS\htB - CSS^{-1}B}\\
        \leq {}& \norm{\htC - CS}\norm{\htB} + \norm{CS}\norm{\htB - S^{-1}B}\\
        \leq {}& \frac{\epsilon}{3} + \b{\norm{CS-\htC} + \norm{\htC}}\frac{\epsilon}{3\norm{\htC}} \\
        \leq {}& \frac{\epsilon}{3} + \frac{\epsilon}{3} + \frac{\epsilon}{3}\frac{\epsilon}{3\norm{\htB}\norm{\htC}} \\
        \leq {}& \epsilon,
    \end{split}\end{equation}
    which causes a contradiction. 
\end{proof}


\section{Upper Bounds for \ac{meta-algo} --- Proof of \Cref{thm:2meta}}\label{sec:setting_ext}
For this section, we consider a slightly more general setting than \ac{metasysid} and prove the theorem in this setting. Consider $K$ minimal systems $\calM_k = \b{r_k, n, m_k, A_k, B_k, C_k, \Sigma_{w,k}, \sigma_{\eta,k}^2I}$ with $r_k,m_k \ll n (\forall k\in[K])$ and the same observer column space. Namely,
\begin{equation}\begin{split}
    \col(C_1) = \col(C_2) = \cdots = \col(C_K)
\end{split}\end{equation}
For every system $\calM_k$, we choose an independent input sequence $\calU_k = \{u_{k,t}\}_{t=0}^{T_k-1}$ with $u_{k,t}\overset{i.i.d.}{\sim}\calN(0, \Sigma_{u,k})$, and we observe $\calY_k = \{y_{k,t}\}_{t=0}^{T_k}$. 
This single trajectory dataset by $\calD_{k} = \calY_k\cup\calU_k$. 

For the above systems, we define the following notations
\begin{equation}\begin{split}
    {}& r = \max_{k\in[K]} r_k, \quad m = \max_{k\in[K]} m_k,\quad \phi_u = \min_{k\in[K]} \sigma_{\min}(\Sigma_u)\\
    {}& \psi_C = \max_{k\in[K]}\sigma_1\b{C_k},\quad \psi_\eta = \max_{k\in[K]}\sigma_{\eta,k}^2, \quad \psi_w = \max_{k\in[K]}\sigma_1(\Sigma_{w,k} + B_k\Sigma_{u,k}B_k\t),\\
    {}& \phi_C = \min\kk\sigma_{\min}(C_k),\quad \phi_O = \min_{k\in[K]}\sigma_{\min}\b{\begin{bmatrix}
        C_{k}\\
        C_{k}A_{k}\\
        \vdots\\
        C_{k}A_{k}^{r-1}
    \end{bmatrix}}, \quad \phi_R = \min_{k\in[K]}\sigma_{\min}\b{\begin{bmatrix}
        B_k & A_kB_k & \dots & A_k^{r-1}B_k
    \end{bmatrix}}\\
\end{split}\end{equation}

For $\wh{\Phi}_{C,k}$ from \Cref{alg:2meta}, we define the following auxiliary system $\wh{\calM}_{k}$ (the projected version of system $\calM_k$) which will be useful for further analysis. 
\begin{equation}\begin{split}\label{eq:3col_2_multi}
    x_{k,t+1} = {}& A_{k}x_{k,t} + B_{k}u_{k,t} + w_{k,t},\\
    \tily_{k,t} = {}& \wh{\Phi}_{C,k}\t C_{k}x_{k,t} + \wh{\Phi}_{C,k}\t \eta_{k,t}.
\end{split}\end{equation}
 
Now we are ready to restate \Cref{thm:2meta} in full details.
\begin{theorem}[\Cref{thm:2meta} Restated for the More General Setting]\label{thm:2meta_full}
    Consider the systems $\calM_{[K]}$, datasets $\calD_{[K]} = \calY_{[K]}\cup\calU_{[K]}$ and constants defined above. Let $\calM_{[K]}$ satisfy Assumption \ref{assmp:sys1} with constants $\psi_A$ and $\rho_A$. Fix any system $k_0\in[K]$. 
    If $T_{k_0}$  and $T_{-k_0} \coloneqq \sum_{k\neq k_0} T_k$ satisfy
    \begin{equation}\begin{split}\label{eq:3col_1_multi}
        &T_{-k_0} \gtrsim \kappa_3\cdot n^2r^3, \quad T_{k_0} \geq \kappa_1\cdot \poly\b{r,m},
    \end{split}\end{equation}
    then $\b{\htA_{k_0}, \htB_{k_0}, \htC_{k_0}}$ from Algorithm \ref{alg:2meta} satisfy the following for some invertible matrix $S$ with probability at least $1-\delta$
    \begin{equation}\begin{split}
        {}& \max\left\{\norm{S^{-1}A_{k_0}S - \htA_{k_0}}, \norm{S^{-1}B_{k_0} - \htB_{k_0}}, \norm{C_{k_0}S - \htC_{k_0}}\right\}\\
        {}& \lesssim \kappa_4 \cdot \sqrt{\frac{n}{T_{-k_0}}}\norm{\htC_{k_0}} + \kappa_2\cdot\sqrt{\frac{\poly\b{r,m}}{T_{k_0}}}.
    \end{split}\end{equation}
    Here $\kappa_1=\kappa_1(\wh{\calM}_{k_0}, \calU_{k_0},\delta)$ and $\kappa_2=\kappa_2(\wh{\calM}_{k_0}, \calU_{k_0},\delta)$ are defined in \Cref{def:idoracle}.
    $\kappa_3 = \kappa_3\b{\calM_{[K]},\calU_{[K]},\delta}, \quad \kappa_4 = \kappa_4\b{\calM_{-k_0},\calU_{-k_0},\delta}$ are detailed below. All of them are problem-related constants independent of system dimensions  modulo logarithmic factors.
    \begin{equation}\begin{split}
        \kappa_3\b{\calM_{[K]}, \calU_{[K]},\delta} = {}& \max\left\{\kappa_4^2\frac{\psi_A^2\psi_C^2}{(1-\rho_A^2)\phi_O^2}, ~\b{\frac{\psi_\eta^2\psi_C^2\psi_w\psi_A^2}{(1-\rho_A^2)\phi_u^2\phi_C^4\phi_R^4}}^2\log^2\b{\frac{\psi_C^2\psi_w\psi_A^4}{1-\rho_A^2} r\log\frac{Kr}{\delta}}\log^4(\frac{Kr}{\delta})\right\},\\
        \kappa_4\b{\calM_{-k_0}, \calU_{-k_0},\delta} = {}& \frac{\psi_\eta}{\phi_u\phi_C^2\phi_R^2}\sqrt{\log\frac{1}{\delta}}.
    \end{split}\end{equation}
\end{theorem}

\begin{proof}[Proof of \Cref{thm:2meta_full}]
For clarity of the proof, we abbreviate all subscripts $k_0$. 
Based on \Cref{eq:3col_1_multi}, $T_{-k_0}$ satisfies the condition of \Cref{lem:2col_full}. We apply \Cref{lem:2col_full} on $(\calD_{-k_0}, \Sigma_{\eta,-k_0})$ and get the following with probability at least $1-\frac{\delta}{2}$
\begin{equation}
    \norm{\wh{\Phi}_C^\perp\t \Phi_C} \lesssim \kappa_4 \sqrt{\frac{n}{T_{-k_0}}} \coloneqq \Delta_{\Phi}.
\end{equation}

The system generating dataset $\tilde{\calD}$, denoted by $\wh{\calM}$, is rewritten in \Cref{eq:3col_2_multi}.
Following exactly the same derivation as in Step 1 of proof of \Cref{thm:1single_full}, we know that $\wh{\calM}$ is minimal. Since $\wh{\Phi}_C$ is independent of the trajectory from $\calM$, the noises $\{\wh{\Phi}_C\t \eta_t\}_{t=0}^T$ are sampled independently of other variables of this trajectory. Moreover, $\wh{\calM}$ satisfy Assumption \ref{assmp:sys1}. Therefore, we can apply \ac{oracle}.
    Following exactly the same derivation as in Step 2 of proofs of \Cref{thm:1single_full}, we get the following for some invertible matrix $S$ with probability at least $1-\delta$
    \begin{equation}
        \max\left\{\norm{S^{-1}AS - \htA}, \norm{S^{-1}B - \htB}, \norm{CS - \htC}\right\} \lesssim \kappa_2\cdot\sqrt{\frac{\poly\b{r,m}}{T}} + \kappa_4 \cdot \sqrt{\frac{n}{T_{-k_0}}}\norm{\htC}.
    \end{equation}
\end{proof}

\subsection[]{Upper Bounds for \text{col-approx}}



The theoretical guarantee for \text{col-approx} with multiple trajectories in \Cref{alg:2meta} is presented in the following lemma.
\begin{lemma}\label{lem:2col_full}
    Fix any positive integer $K$. Consider $K$ systems $\calM_{[K]}$ and their datasets $\calD_{[K]} = \calY_{[K]}\cup\calU_{[K]}$ satisfying the description in \Cref{sec:setting_ext}. Suppose the systems satisfy Assumption \ref{assmp:sys1} with constants $\psi_A$ and $\rho_A$.
    With the notations defined in \Cref{sec:setting_ext}, 
    if $\calT \coloneqq \sum\kk T_k$ satisfies the following inequality
    \begin{equation}\begin{split}\label{eq:thm_col_1_multi}
        \calT \gtrsim \underbrace{\b{\frac{\psi_\eta^2\psi_C^2\psi_w\psi_A^2}{(1-\rho_A^2)\phi_u^2\phi_C^4\phi_R^4}}^2\log^2\b{\frac{\psi_C^2\psi_w\psi_A^4}{1-\rho_A^2} r\log\frac{Kr}{\delta}}\log^4(\frac{Kr}{\delta})}_{\kappa_5\b{\calM_{[K]},\calU_{[K]},\delta}}\cdot n^2r^3,
    \end{split}\end{equation}
    then $\wh{\Phi}_C = \text{col-approx}\b{\calY_{[K]}}$ (\Cref{alg:column_space}) satisfies the following with probability at least $1-\delta$
    \begin{equation}\begin{split}
        \htr_c = \rank(C_1)=\cdots=\rank(C_K), \quad \norm{\wh{\Phi}_C^\perp\t\Phi_C} \lesssim \underbrace{\frac{\psi_\eta}{\phi_u\phi_C^2\phi_R^2}\sqrt{\log\frac{1}{\delta}}}_{\kappa_4\b{\calM_{[K]},\calU_{[K]},\delta}}\cdot\sqrt{\frac{n}{\calT}}.
    \end{split}\end{equation}
\end{lemma}

\begin{proof}
From the system dynamics, we know that
\begin{equation}\begin{split}
    \Sigma_y = {}& \sum\kk\sum_{t=0}^{T_k} y_{k,t}y_{k,t}\t =  \sum\kk\sum\tk C_k x_{k,t}x_{k,t}\t C_k\t + \sum\kk\sum\tk \eta_{k,t}\eta_{k,t}\t + \sum\kk\sum\tk \b{C_kx_{k,t}\eta_{k,t}\t + \eta_{k,t} C_k\t x_{k,t}\t}\\
    = {}& \sum\kk\sum\tk C_k x_{k,t}x_{k,t}\t C_k\t + \sum\kk\sum\tk (\eta_{k,t}\eta_{k,t}\t-\sigma_{\eta,k}^2I) + \sum\kk\sum\tk \b{C_kx_{k,t}\eta_{k,t}\t + \eta_{k,t} C_k\t x_{k,t}\t}\\
    {}& \hspace{2em} + \sum_{k\in[K]} (T_k+1)\sigma_{\eta,k}^2I\\
\end{split}\end{equation}
Here the first term is the information on the shared column space of $\{C_k\}\kk$, while the second and third terms are only noises. 
For the rest of the proof, we first upper bound the norms of the noises (\textit{step 1}). With this, we show that $\htr_c = \rank(C_1)$ with high probability (\textit{step 2}). Then we apply our subspace perturbation result to upper bound the influence of the noises on the eigenspace of the first term (\textit{step 3}).

\textbf{Step 1: Noise Norm Upper Bounds.} With $\rank(C_1) = r_c$ and $\col(C_1)=\cdots=\col(C_K)$, we can rewrite $C_k = \Phi_C\alpha_k$ where an orthonormal basis of $\col(C_1)$ forms the columns of $\Phi_C\in\bbR^{n\times r_c}$ and $\alpha_k\in\bbR^{r_c\times r}$ is a full row rank matrix. 
It is then clear that
\begin{equation}\begin{split}
    \sigma_{\min}(\alpha_k) = \sigma_{\min}(C_k) \geq \phi_C, \quad \sigma_{1}(\alpha_k) = \sigma_{1}(C_k) \leq \psi_C, \quad \forall k\in[K].
\end{split}\end{equation}
Let $\Sigma_C = \sum\kk\sum\tk C_k x_{k,t}x_{k,t}\t C_k\t$, $\bar{\Sigma}_C = \Sigma_C + \calT I$ and $\Sigma_\alpha = \sum\kk\sum\tk \alpha_k x_{k,t}x_{k,t}\t \alpha_k\t$. We then have $\Sigma_C = \Phi_C\Sigma_\alpha \Phi_C\t$ and $\bar{\Sigma}_C = \Phi_C\Sigma_\alpha \Phi_C\t + \calT I$.
Then from Lemma \ref{lem:inv_multi}, Lemma \ref{lem:inv_crossConc_stable_multi}, Lemma \ref{lem:inv_cov_noise}, and Lemma \ref{lem:gram_upper}, we have the following with probability at least $1-\delta$
\begin{equation}\begin{split}
    {}& \frac{\phi_u\phi_C^2\phi_R^2}{8}\calT I \preceq \Sigma_\alpha \precsim \frac{\psi_C^2\psi_w\psi_A^2}{1-\rho_A^2}r\calT\log\frac{K}{\delta}I,\\
    {}& \norm{\b{\bar{\Sigma}_C}^{-\frac{1}{2}}\sum\kk\sum\tk C_kx_{k,t}\eta_{k,t}\t} \lesssim \bar{\psi}\sqrt{\psi_\eta n\log\frac{1}{\delta}},\\
    {}& \norm{\sum\kk\sum\tk \b{\eta_{k,t}\eta_{k,t}\t - \sigma_{\eta,k}^2I}} \lesssim \psi_\eta\sqrt{n\calT\log\frac{1}{\delta}},
\end{split}\end{equation}
with $\bar{\psi} = \sqrt{\log\b{\frac{\psi_\alpha^2\psi_w\psi_A^4}{1-\rho_A^2} r\log\frac{2Kr}{\delta}}}$.

\textbf{Step 2: Order Estimation Guarantee.}
\begin{equation}\begin{split}
    \Sigma_y = {}& \Sigma_C + \underbrace{\sum\kk\sum\tk \b{\eta_{k,t}\eta_{k,t}\t - \sigma_{\eta,k}^2I} + \sum\kk\sum\tk \b{C_kx_{k,t}\eta_{k,t}\t + \eta_{k,t} C_k\t x_{k,t}\t}}_{\Delta} + \sum_{k\in[K]} (T_k+1)\sigma_{\eta,k}^2I.
\end{split}\end{equation}
The inequalities of Step 1 imply
\begin{equation}\begin{split}
    \norm{\Delta} \leq {}& \norm{\sum\kk\sum\tk \b{\eta_{k,t}\eta_{k,t}\t - \sigma_{\eta,k}^2I}} + 2\norm{\sum\kk\sum\tk C_kx_{k,t}\eta_{k,t}\t}\\
    \lesssim {}& \psi_\eta\sqrt{n\calT\log\frac{1}{\delta}} + \bar{\psi}\sqrt{\psi_\eta n\log\frac{1}{\delta}}\norm{(\bar{\Sigma}_C)^{\frac{1}{2}}}\\
    = {}& \psi_\eta\sqrt{n\calT\log\frac{1}{\delta}} + \bar{\psi}\sqrt{\psi_\eta n\log\frac{1}{\delta}}\b{\norm{(\calT I+\Sigma_\alpha)^{\frac{1}{2}}}}\\
    \lesssim {}& \psi_\eta\sqrt{n\calT\log\frac{1}{\delta}} + \bar{\psi}\sqrt{\psi_\eta n\log\frac{1}{\delta}}\sqrt{\frac{\psi_C^2\psi_w\psi_A^2}{1-\rho_A^2}r\calT\log\frac{K}{\delta}+\calT}\\
    \lesssim {}& \bar{\psi}\psi_\eta \sqrt{n\log\frac{1}{\delta}}\sqrt{\frac{\psi_C^2\psi_w\psi_A^2}{1-\rho_A^2}r\calT\log\frac{K}{\delta}}\\
    = {}& \sqrt{\frac{\psi_\eta^2\psi_C^2\psi_w\psi_A^2}{1-\rho_A^2}}\bar{\psi}\log\frac{K}{\delta}\sqrt{nr\calT}\\
\end{split}\end{equation}
Therefore, for each $i\in[r_c]$ we have the following for some positive constant $c_1$
\begin{equation}\begin{split}
    \sigma_i\b{\Sigma_y} - \sum_{k\in[K]} (T_k+1)\sigma_{\eta,k}^2 \geq {}& \sigma_i(\Sigma_C) - \norm{\Delta} \geq \sigma_{\min}(\Sigma_\alpha) - \norm{\Delta}\\
    \geq {}& \frac{\phi_u\phi_C^2\phi_R^2}{8}\calT - c_1\sqrt{\frac{\psi_\eta^2\psi_C^2\psi_w\psi_A^2}{1-\rho_A^2}}\bar{\psi}\log\frac{K}{\delta}\sqrt{nr\calT}\\
    \geq {}& \frac{\phi_u\phi_C^2\phi_R^2}{16}\calT.
\end{split}\end{equation}
Here the first inequality holds according to Theorem 1 in \cite{stewart_matrix_1990} and the last line is because of Equation \ref{eq:thm_col_1_multi}. 
For $i\in[r_c+1, n]$, 
\begin{equation}\begin{split}
    \sigma_i(\Sigma_y) - \sum_{k\in[K]} (T_k+1)\sigma_{\eta,k}^2 \leq \norm{\Delta} \leq c_1\sqrt{\frac{\psi_\eta^2\psi_C^2\psi_w\psi_A^2}{1-\rho_A^2}}\bar{\psi}\log\frac{K}{\delta}\sqrt{nr\calT}.
\end{split}\end{equation}
Based on the above three inequalities and Equation \ref{eq:thm_col_1_multi}, we conclude that with probability at least $1-\delta$, the following hold
\begin{equation}\begin{split}
    \sigma_i(\Sigma_y)-\sigma_j(\Sigma_y) \leq {}& 2c_1\sqrt{\frac{\psi_\eta^2\psi_C^2\psi_w\psi_A^2}{1-\rho_A^2}}\bar{\psi}\log\frac{K}{\delta}\sqrt{nr\calT}\\
    < {}& \calT^{3/4}, \quad \forall i<j\in[r_c],\\
    \sigma_{r_c}(\Sigma_y)-\sigma_{r_c+1}(\Sigma_y) \geq {}& \frac{\phi_u\phi_C^2\phi_R^2}{16}\calT - c_1\sqrt{\frac{\psi_\eta^2\psi_C^2\psi_w\psi_A^2}{1-\rho_A^2}}\bar{\psi}\log\frac{K}{\delta}\sqrt{nr\calT}\\
    > {}& \calT^{3/4}.
\end{split}\end{equation}
Therefore, from the definition of $\htr_c$, we know that $\htr_c = r_c$ with probability at least $1-\delta$.

\textbf{Step 3: Column Space Estimation Guarantee.} With $\htr_c = r_c$, now we try to apply our subspace perturbation result, i.e. Lemma \ref{lem:perturbation}, on matrix $\Sigma_y-\sum_{k\in[K]} (T_k+1)\sigma_{\eta,k}^2I+\calT I$. Notice that this matrix has exactly the same eigenspace as $\Sigma_y$ and therefore the eigenspace of its first $r_c$ eigenvectors is $\wh{\Phi}_C$ (line 9 in Algorithm \ref{alg:column_space}). This matrix can be decomposed as
\begin{equation}\begin{split}
    {}& \Sigma_y-\sum_{k\in[K]} (T_k+1)\sigma_{\eta,k}^2I+\calT I\\
    = {}& \underbrace{\bar{\Sigma}_C}_{M \text{ in Lemma \ref{lem:perturbation}}} + \underbrace{\sum\kk\sum\tk \b{\eta_{k,t}\eta_{k,t}\t - \sigma_{\eta,k}^2I}}_{\Delta_2 \text{ in Lemma \ref{lem:perturbation}}} + \underbrace{\sum\kk\sum\tk \b{C_kx_{k,t}\eta_{k,t}\t + \eta_{k,t} C_k\t x_{k,t}\t}}_{\Delta_1+\Delta_1\t \text{ in Lemma \ref{lem:perturbation}}}.
\end{split}\end{equation}
For matrix $\bar{\Sigma}_C = \Phi_C \Sigma_\alpha \Phi_C\t + \calT I = \Phi_C\b{\Sigma_\alpha+\calT I}\Phi_C\t + \calT\Phi_C^\perp\Phi_C^\perp\t$, it is clear that its SVD can be written as
\begin{equation}\begin{split}
    \bar{\Sigma}_C = \begin{bmatrix}
        \tilde{\Phi}_C & \Phi_C^\perp
    \end{bmatrix} \begin{bmatrix}
        \Lambda_1 & \\
        & \calT I
    \end{bmatrix} \begin{bmatrix}
        \tilde{\Phi}_C\t\\
        \Phi_C^\perp\t
    \end{bmatrix}, \quad \Lambda_1 = \diag\b{\sigma_1(\Sigma_\alpha)+\calT, \dots, \sigma_{\min}(\Sigma_\alpha)+\calT}. 
\end{split}\end{equation}
where $\tilde{\Phi}_C$ is an orthonormal basis of $\col(C_1)$. 
Then we conclude that the following hold for some large enough positive constant $c_3$
\begin{equation}\begin{split}
    {}& \sigma_{1}\b{\bar{\Sigma}_C} \leq \calT + \sigma_{1}(\Sigma_\alpha) \leq c_3 \b{1+\frac{\psi_C^2\psi_w\psi_A^2}{1-\rho_A^2}r\log\frac{K}{\delta}}\calT,\\
    {}& \sigma_{r_c}\b{\bar{\Sigma}_C} - \sigma_{r_c+1}(\bar{\Sigma}_C) = \sigma_{\min}(\Sigma_\alpha) \geq \frac{\phi_u\phi_C^2\phi_R^2}{32}\calT, \quad \sigma_{\min}(\bar{\Sigma}_C) = \calT.
\end{split}\end{equation}
Now we are ready to apply Lemma \ref{lem:perturbation} on $\bar{\Sigma}_C$ with the following for positive constants $c_4, c_5$ large enough
\begin{equation}\begin{split}
    {}& \alpha = c_4\bar{\psi}\sqrt{\psi_\eta n\log\frac{1}{\delta}}, \quad \beta = c_5\psi_\eta\sqrt{n\calT\log\frac{1}{\delta}},\\
    {}& \delta_M = \frac{\phi_u\phi_C^2\phi_R^2}{32}\calT, \quad \psi_M = 2c_3\frac{\psi_C^2\psi_w\psi_A^2}{1-\rho_A^2}r\calT\log\frac{K}{\delta}, \quad \phi_M = \calT, \quad \sigma_{1}(\Lambda_2) = \calT.
\end{split}\end{equation}
Again, $\bar{\psi} = \sqrt{\log\b{\frac{\psi_\alpha^2\psi_w\psi_A^4}{1-\rho_A^2} r\log\frac{2Kr}{\delta}}}$.
From Equation \ref{eq:thm_col_1_multi}, it is clear that $\sqrt{\phi_M} \gtrsim \alpha$ and $\delta_M \gtrsim \alpha\sqrt{\psi_M} + \beta$. Therefore,
\begin{equation}\begin{split}
    \norm{\wh{\Phi}_C^\perp\t\Phi_C} = \norm{\wh{\Phi}_C^\perp\t\tilde{\Phi}_C} \lesssim {}& \frac{\psi_\eta}{\phi_u\phi_C^2\phi_R^2}\sqrt{\frac{n}{\calT}\log\frac{1}{\delta}} + \bar{\psi}\frac{\sqrt{\psi_\eta}}{\phi_u\phi_C^2\phi_R^2}\sqrt{\frac{n}{\calT}\log\frac{1}{\delta}} + \bar{\psi}^2\frac{\psi_\eta\sqrt{\psi_C^2\psi_w\psi_A^2}}{\phi_u\phi_C^2\phi_R^2\sqrt{1-\rho_A^2}}\frac{n\sqrt{r}}{\calT}\log^2\frac{K}{\delta}\\
    \lesssim {}& \frac{\psi_\eta}{\phi_u\phi_C^2\phi_R^2}\sqrt{\frac{n}{\calT}\log\frac{1}{\delta}} + \bar{\psi}^2\frac{\psi_\eta\sqrt{\psi_C^2\psi_w\psi_A^2}}{\phi_u\phi_C^2\phi_R^2\sqrt{1-\rho_A^2}}\frac{n\sqrt{r}}{\calT}\log^2\frac{K}{\delta}\\
    \lesssim {}& \frac{\psi_\eta}{\phi_u\phi_C^2\phi_R^2}\sqrt{\frac{n}{\calT}\log\frac{1}{\delta}}.
\end{split}\end{equation}
\end{proof}

\subsubsection{Supporting Details}
\begin{lemma}\label{lem:inv_multi}
    Consider the same setting as Theorem \ref{lem:2col_full}. 
    Consider full row rank matrices $\{\alpha_k\in\bbR^{a\times r_k}\}_{k\in[K]}$ for any positive integer $a\leq \min_kr_k$.
    Let
    \begin{equation}\begin{split}
        \Sigma_\alpha = \sum_{k\in[K],t\in[T_k]} \alpha_kx_{k,t}x_{k,t}\t \alpha_k\t, \quad \phi_\alpha = \min\left\{\min_{k\in[K]} \{\sigma_{a}\b{\alpha_k}\},1\right\}, \quad \psi_\alpha = \max\left\{\max_{k\in[K]} \{\sigma_{1}\b{\alpha_k}\},1\right\}.
    \end{split}\end{equation}
    Then the following holds with probability at least $1-\delta$ if $\calT \gtrsim \frac{\psi_w^2\psi_\alpha^4\psi_A^2}{\phi_w^2\phi_u^2\phi_\alpha^4\phi_R^4}Kr^3\log\frac{r}{\delta}\cdot\log\b{\frac{\psi_\alpha^2\psi_w\psi_A^4}{1-\rho_A^2} r\log\frac{2Kr}{\delta}}$,
    \begin{equation}\begin{split}
        \Sigma_\alpha \succeq \frac{\phi_u\phi_\alpha^2\phi_R^2}{8}\calT I.
    \end{split}\end{equation}
\end{lemma}

\begin{proof}
    
We first define set $\calK = \{k: T_k \geq r\}$. For further analysis, we let $\tilw_{k,t} \coloneqq B_ku_{k,t} + w_{k,t} \sim \calN\b{0, \Sigma_{\tilw,k}}$ with $\Sigma_{\tilw,k} \coloneqq B_k\Sigma_{u,k}B_k\t + \Sigma_{w,k}$ and we can rewrite the system dynamics as follows
\begin{equation}\begin{split}\label{eq:system_meta_multi}
    x_{k,t+1} = {}& A_kx_{k,t} + \tilw_{k,t} = A_kx_{k,t} + B_ku_{k,t} + w_{k,t},\\
    y_{k,t} = {}& C_kx_{k,t} + \eta_{k,t}.
\end{split}\end{equation}
Then we know that $\sigma_{1}(\Sigma_{\tilw,k}) \leq \psi_w$.
For simplicity, we define $\Sigma_{\alpha,\tau} \coloneqq \sum_{k\in\calK}\sum_{t=1}^{T_k-\tau}\b{\alpha_k A_k^{\tau}}\b{x_{k,\tau} x_{k,\tau}\t} (\alpha_k A_k^{\tau})\t$. It is then clear that
    \begin{equation}\begin{split}\label{eq:covA_1_multi}
        \Sigma_{\alpha} \succeq {}& \Sigma_{\alpha,0} = \sum_{k\in\calK} \sum_{t=1}^{T_k} \alpha_kx_{k,t}x_{k,t}\t \alpha_k\t \\
        ={}& \sum_{k\in\calK} \sum_{t=1}^{T_k} \alpha_k\b{A_kx_{k,t-1}x_{k,t-1}\t A_k\t  + \tilw_{k,t-1}\tilw_{k,t-1}\t  + A_kx_{k,t-1}\tilw_{k,t-1}\t  + \tilw_{k,t-1}x_{k,t-1}\t A_k\t} \alpha_k\t \\
        = {}& \Sigma_{\alpha,1} + \sum_{k\in\calK} \sum_{t=0}^{T_k-1} \alpha_k\tilw_{k,t}\tilw_{k,t}\t \alpha_k\t  + \sum_{k\in\calK} \sum_{t=0}^{T_k-1} \b{\alpha_kA_kx_{k,t}\b{\alpha_k\tilw_{k,t}}\t  + \b{\alpha_k\tilw_{k,t}}x_{k,t}\t A_k\t \alpha_k\t}.
    \end{split}\end{equation}
    By Lemma \ref{lem:inv_crossConc_stable_multi} and \ref{lem:inv_cov_noise}, the following events hold with probability at least $1-\delta/r$, 
    \begin{equation}\begin{split}
        \norm{\sum_{k\in\calK} \sum_{t=0}^{T_k-1} \alpha_k\tilw_{k,t}\tilw_{k,t}\t \alpha_k\t - \sum_{k\in\calK} T_k\alpha_k\Sigma_{\tilw,k}\alpha_k\t} \lesssim {}& \psi_w\psi_\alpha^2\sqrt{r\calT\log\frac{2r}{\delta}},\\
        \norm{\b{\Sigma_{\alpha,1}+\calT I}^{-\frac{1}{2}}\sum_{k\in\calK} \sum_{t=0}^{T_k-1} \alpha_kA_kx_{k,t}\b{\alpha_k\tilw_{k,t}}\t} \lesssim {}& \bar{\psi}\sqrt{r\psi_w\psi_\alpha^2\log\frac{2r}{\delta}},
    \end{split}\end{equation}
    with $\bar{\psi} = \sqrt{\log\b{\frac{\psi_\alpha^2\psi_w\psi_A^4}{1-\rho_A^2} r\log\frac{2Kr}{\delta}}}$.
    From the first inequality, it is clear that the following holds for some large enough positive constant $c_1$
    \begin{equation}\begin{split}
        \sum_{k\in\calK} \sum_{t=0}^{T_k-1} \alpha_k\tilw_{k,t}\tilw_{k,t}\t \alpha_k\t \succeq \sum_{k\in\calK} T_k\alpha_k\Sigma_{\tilw,k}\alpha_k\t - c_1\psi_w\psi_\alpha^2\sqrt{r\calT\log\frac{2r}{\delta}}I
    \end{split}\end{equation}
    
    We have the following from the second inequality
    \begin{equation}\begin{split}
        {}& \norm{\b{\Sigma_{\alpha,1}+\calT I}^{-\frac{1}{2}}\sum_{k\in\calK} \sum_{t=0}^{T_k-1} \alpha_kA_kx_{k,t}\b{\alpha_k\tilw_{k,t}}\t \b{\Sigma_{\alpha,1}+\calT I}^{-\frac{1}{2}}}\\
        \lesssim {}& \bar{\psi}\sqrt{r\psi_w\psi_\alpha^2\log\frac{2r}{\delta}} \norm{\b{\Sigma_{\alpha,1} + \calT I}^{-\frac{1}{2}}}\\
        \lesssim {}& \bar{\psi}\sqrt{\frac{r\psi_w\psi_\alpha^2\log\frac{2r}{\delta}}{\calT}}.
    \end{split}\end{equation}
    This implies the following for some positive constant $c_2$ 
    \begin{equation}\begin{split}
        {}& \sum_{k\in\calK} \sum_{t=0}^{T_k-1} \b{\alpha_kA_kx_{k,t}\b{\alpha_k\tilw_{k,t}}\t  + \b{\alpha_k\tilw_{k,t}}x_{k,t}\t A_k\t \alpha_k\t} 
        \succeq -c_2\bar{\psi}\sqrt{\frac{r\psi_w\psi_\alpha^2\log\frac{2r}{\delta}}{\calT}} \b{\Sigma_{\alpha,1} + \calT I}.
    \end{split}\end{equation}
    Plugging back into Equation \ref{eq:covA_1_multi} gives the following for some positive constant $c_3$
    \begin{equation}\begin{split}\label{eq:covA_2_multi}
        \Sigma_{\alpha,0} \succeq {}& \Sigma_{\alpha,1} + \sum_{k\in\calK} T_k \alpha_k\Sigma_{\tilw,k}\alpha_k\t - c_1\psi_w\psi_\alpha^2\sqrt{r\calT\log\frac{2r}{\delta}}I - c_2\bar{\psi}\sqrt{\frac{r\psi_w\psi_\alpha^2\log\frac{2r}{\delta}}{\calT}}\b{\Sigma_{\alpha,1} + \calT I}\\
        \succeq {}& \b{1 - c_2\bar{\psi}\sqrt{\frac{r\psi_w\psi_\alpha^2\log\frac{2r}{\delta}}{\calT}}}\Sigma_{\alpha,1} + \b{\sum_{k\in\calK} T_k \alpha_k\Sigma_{\tilw,k}\alpha_k\t - c_3\psi_w\psi_\alpha^2\bar{\psi}\sqrt{r\calT\log\frac{2r}{\delta}}I}\\
    \end{split}\end{equation}
    Similarly, we expand $\Sigma_{\alpha,1}, \Sigma_{\alpha,2}, \cdots, \Sigma_{\alpha,r-1}$ and have the following with probability at least $1-\delta$
    \begin{equation*}\begin{split}
        \Sigma_{\alpha,0} \succeq {}& \b{\sum_{k\in\calK}(T_k-r+1)\sum_{i=0}^{r-1}\b{\alpha_k A_k^i}\Sigma_{\tilw,k}\b{\alpha_k A_k^i}\t - c_3\psi_w\psi_\alpha^2\psi_A^2\bar{\psi} \cdot r\sqrt{r(\calT-Kr+K)\log\frac{2r}{\delta}}I} \\
        {}& \hspace{2em}\cdot\b{1 - c_{2}\bar{\psi}\sqrt{\frac{r\psi_w\psi_\alpha^2\psi_A^2\log\frac{2r}{\delta}}{ \calT-Kr+K }}}^{r-1} + \b{1 - c_{2}\bar{\psi}\sqrt{\frac{r\psi_w\psi_\alpha^2\psi_A^2\log\frac{2r}{\delta}}{ \calT-Kr+K }}}^r\Sigma_{\alpha,r}\\
    \end{split}\end{equation*}
    The above inequality is further simplified as follows
    \begin{equation*}\begin{split}
        \Sigma_{\alpha,0} \overset{(i)}{\succeq} {}& \b{\phi_u\sum_{k\in\calK}(T_k-r+1)\sum_{i=0}^{r-1}\b{\alpha_k A_k^i}B_kB_k\t\b{\alpha_k A_k^i}\t - c_3\psi_w\psi_\alpha^2\psi_A^2\bar{\psi} \cdot r\sqrt{r\calT\log\frac{2r}{\delta}}I}\\
        {}& \hspace{2em} \cdot \b{1 - c_2\bar{\psi}\sqrt{\frac{r\psi_w\psi_\alpha^2\psi_A^2\log\frac{2r}{\delta}}{ \calT-Kr+K }}}^{r}\\
        \overset{(ii)}{\succeq} {}& \b{\phi_u\sum_{k\in\calK}(T_k-r)\alpha_k \b{\sum_{i=0}^{r-1}A_k^iB_k(A_k^iB_k)\t}\alpha_k\t - c_3\psi_w\psi_\alpha^2\psi_A^2\bar{\psi} \cdot r\sqrt{r\calT\log\frac{2r}{\delta}}I} \b{1 - \frac{1}{2r}}^{r}\\
        \overset{(iii)}{\succeq} {}& \frac{1}{2}\b{\phi_u\sum_{k\in\calK}(T_k-r)\alpha_k \b{\sum_{i=0}^{r-1}A_k^iB_k(A_k^iB_k)\t}\alpha_k\t - c_3\psi_w\psi_\alpha^2\psi_A^2\bar{\psi} \cdot r\sqrt{r\calT\log\frac{2r}{\delta}}I}\\
        \overset{(iv)}{\succeq} {}& \frac{1}{2}\b{\phi_u\phi_\alpha^2\phi_R^2(\calT-2Kr) - c_3\psi_w\psi_\alpha^2\psi_A^2\bar{\psi} \cdot r\sqrt{r\calT\log\frac{2r}{\delta}}}I \\
        \overset{(v)}{\succeq} {}& \frac{\phi_u\phi_\alpha^2\phi_R^2}{8} \calT I\\
    \end{split}\end{equation*}
    Here $(i)$ is because $\Sigma_{\alpha,r} \succeq 0$ and $\Sigma_{\tilw,k} \succeq B_k\Sigma_{u,k} B_k\t \succeq \phi_u B_kB_k\t$, $(ii)$ is because $\calT \gtrsim \psi_w\psi_\alpha^2\psi_A^2Kr^3\log\frac{r}{\delta}\cdot\log\b{\frac{\psi_\alpha^2\psi_w\psi_A^4}{1-\rho_A^2} r\log\frac{Kr}{\delta}}$, $(iii)$ is because $(1-\frac{1}{2r})^r \geq \frac{1}{2}$ for all positive integers, $(iv)$ is because $\sum_{k\in\calK}(T_k-r) \geq -Kr + \sum_{k\in\calK} T_k = -Kr + \calT - \sum_{k\notin \calK} T_k \geq -Kr + \calT - Kr$, and $(v)$ is because $\calT \gtrsim \frac{\psi_w^2\psi_\alpha^4\psi_A^4}{\phi_u^2\phi_{\alpha}^4\phi_R^4}Kr^3\log\frac{r}{\delta}\cdot\log\b{\frac{\psi_\alpha^2\psi_w\psi_A^4}{1-\rho_A^2} r\log\frac{Kr}{\delta}}$. 
\end{proof}

\begin{lemma}\label{lem:inv_crossConc_stable_multi}
    Consider the same setting as Theorem \ref{lem:2col_full}. 
    For any positive integer $a$, let $\{\zeta_{k,t}\in\bbR^a\}_{t=0}^{T_k}$ be a sequence of i.i.d Gaussian vectors from $\calN(0,\Sigma_{\zeta,k})$ such that $\zeta_{k,t}$ is independent of $x_{k,t}$. Define $\psi_\zeta = \max_k\sigma_{1}(\Sigma_{\zeta,k})$.
    We make the following definition for any $b\times r$ matrix $P_{[K]}$ for any positive integer $b$
    \begin{equation}\begin{split}
        \bar{\Sigma}_P = \sum_{k\in[K],t\in[T_k]} P_kx_{k,t}x_{k,t}\t P_k\t + \calT I, \quad \psi_P = \max_{k\in[K]} \sigma_1(P_k). 
    \end{split}\end{equation}
    Then the following holds with probability at least $1-\delta$,
    \begin{equation}\begin{split}
        \norm{\b{\bar{\Sigma}_P}^{-\frac{1}{2}} \sum_{k\in[K],t\in[T_k]} P_kx_{k,t}\zeta_{k,t}\t} \lesssim \bar{\psi}\sqrt{\max\{r,a\}\psi_\zeta\log\frac{1}{\delta}}.
    \end{split}\end{equation}
    Here $\bar{\psi} = \sqrt{\log\b{\frac{\psi_P^2\psi_w\psi_A^2}{1-\rho_A^2} r\log\frac{K}{\delta}}}$. 
\end{lemma}

\begin{proof}
    From Lemma \ref{lem:gen_norm2Net}, we know that the following holds for some vector $v\in\mathbb{S}^{a-1}$
    \begin{equation}\label{eq:thm1_1_full}\begin{split}
        \bbP\b{\norm{\b{\bar{\Sigma}_P}^{-\frac{1}{2}} \sum_{k\in[K],t\in[T_k]} P_kx_{k,t}\zeta_{k,t}\t}>z} \leq 5^a \bbP\b{\norm{\b{\bar{\Sigma}_P}^{-\frac{1}{2}} \sum_{k\in[K],t\in[T_k]} P_kx_{k,t}\zeta_{k,t}\t v}>\frac{z}{2}}.
    \end{split}\end{equation}
    Notice that $\zeta_{k,t}\t v$ are independent Gaussian varaibles from distribution $\calN(0, v\t \Sigma_{\zeta,k}v)$ for all $k\in[K],t\in[T_k]$, which is $c_1\sqrt{\psi_\zeta}$-subGaussian for some positive constant $c_1$. Then applying Theorem 1 in \cite{abbasi-yadkori_improved_2011} on sequence $\{P_kx_{k,t}\}_{t,k}$ and sequence $\{\zeta_{k,t}\t v\}_{t,k}$, gives the following inequality
    \begin{equation}\begin{split}
        \bbP\b{\norm{\b{\bar{\Sigma}_P}^{-\frac{1}{2}} \sum_{k\in[K],t\in[T_k]} P_kx_{k,t}\zeta_{k,t}\t v} >  \sqrt{2c_1^2\psi_\zeta  \log\b{\frac{\det(\bar{\Sigma}_P)^{\frac{1}{2}}\det(\calT I)^{-\frac{1}{2}}}{\delta}}}} \leq \delta.
    \end{split}\end{equation}
    Substituting the above result back gives the following inequality
    \begin{equation}\begin{split}
        \bbP\b{\norm{\b{\bar{\Sigma}_P}^{-\frac{1}{2}} \sum_{k\in[K],t\in[T_k]} P_kx_{k,t}\zeta_{k,t}} >  \sqrt{8c_1^2\psi_\zeta  \log\b{\frac{\det(\bar{\Sigma}_P)^{\frac{1}{2}}\det(\calT I)^{-\frac{1}{2}}}{\delta}}}} \leq 5^a\delta,
    \end{split}\end{equation}
    which implies the following inequality holds with probability at least $1-\frac{\delta}{2}$
    \begin{equation}\label{eq:cross_1_full}\begin{split}
        \norm{\b{\bar{\Sigma}_P}^{-\frac{1}{2}} \sum_{k\in[K],t\in[T_k]} P_kx_{k,t}\zeta_{k,t}} \leq \sqrt{8c_1^2\psi_\zeta\log\b{\frac{2\det(\bar{\Sigma}_P)^{\frac{1}{2}}\det(\calT I)^{-\frac{1}{2}}}{\delta}} + 8c_1^2\psi_\zeta a\log5}.
    \end{split}\end{equation}

    Now consider $\bar{\Sigma}_P = \sum_{k\in[K],t\in[T_k]} P_kx_{k,t}x_{k,t}\t P_k\t  + \calT I$. Then
    \begin{equation}\begin{split}
        \det(\calT I) ={}&  (\calT )^b, \quad \det(\bar{\Sigma}_P) = (\calT )^{b-r}\prod_{i=1}^{r} \b{\calT +\lambda_i\b{\sum_{k\in[K],t\in[T_k]} P_kx_{k,t}x_{k,t}\t P_k\t }}.
    \end{split}\end{equation}
    From Lemma \ref{lem:gram_upper}, we have the following with probability at least $1-\delta/(2K)$ for every $k\in[K]$
    \begin{equation}\begin{split}
        \norm{\sum_{t\in[T_k]} x_{k,t}x_{k,t}\t} \lesssim \frac{\psi_w\psi_A^2r}{1-\rho_A^2}T_k\log\frac{2K}{\delta} \lesssim \frac{\psi_w\psi_A^2r}{1-\rho_A^2}T_k\log\frac{K}{\delta}.
    \end{split}\end{equation}
    Therefore,
    \begin{equation}\begin{split}
        {}& \lambda_i\b{\sum_{k\in[K],t\in[T_k]} P_kx_{k,t}x_{k,t}\t P_k\t } \\
        \leq {}& \sum_{k\in[K]}\norm{\sum_{t\in[T_k]} P_kx_{k,t}x_{k,t}\t P_k\t }\\
        \leq {}& \sum\kk \norm{\sum_{t\in[T_k]} x_{k,t}x_{k,t}\t}\norm{ P_kP_k\t} \\
        \lesssim {}& \frac{\psi_P^2\psi_w\psi_A^2}{1-\rho_A^2}r\calT\log\frac{K}{\delta}\\
    \end{split}\end{equation}
    Substituting back gives the following for some positive constant $c_2$
    \begin{equation}\begin{split}\label{eq:cross_2_full}
        \det(\bar{\Sigma}_P) ={}&  \calT^{b-r}\prod_{i=1}^{r} \b{\calT +\lambda_i\b{\sum_{k\in[K],t\in[T_k]} P_kx_{k,t}x_{k,t}\t P_k\t }}\\
        \leq {}& \calT^b \b{1+c_2\frac{\psi_P^2\psi_w\psi_A^2}{1-\rho_A^2} r\log\frac{K}{\delta}}^r,
    \end{split}\end{equation}
    which gives
    \begin{equation}\begin{split}
        \det(\bar{\Sigma}_P)^{\frac{1}{2}}\det(\calT I)^{-\frac{1}{2}} \leq \b{1+c_2\frac{\psi_P^2\psi_w\psi_A^2}{1-\rho_A^2} r\log\frac{K}{\delta}}^{\frac{r}{2}}.
    \end{split}\end{equation}

    Finally, with a union bound over all above events, we get the following with probability at least $1-\delta$ from Equation \ref{eq:cross_1_full} and \ref{eq:cross_2_full}
    \begin{equation}\begin{split}
        {}& \norm{\b{\bar{\Sigma}_P}^{-\frac{1}{2}} \sum_{k\in[K],t\in[T_k]} P_kx_{k,t}\zeta_{k,t}}\\
        \leq {}& \sqrt{8c_1^2\psi_\zeta\log\b{\frac{2\det(\bar{\Sigma}_P)^{\frac{1}{2}}\det(\calT I)^{-\frac{1}{2}}}{\delta}} + 8c_1^2\psi_\zeta a\log5}\\
        \leq {}& \sqrt{4c_1^2r\psi_\zeta\log\b{2+2c_2\frac{\psi_P^2\psi_w\psi_A^2}{1-\rho_A^2} r\log\frac{K}{\delta}} + 8c_1^2\psi_\zeta\log\frac{1}{\delta} + 8c_1^2\psi_\zeta a\log5}\\
        \lesssim {}& \sqrt{\max\{r,a\}\psi_\zeta\log\frac{1}{\delta}}\cdot\sqrt{\log\b{\frac{\psi_P^2\psi_w\psi_A^2}{1-\rho_A^2} r\log\frac{K}{\delta}}}.
    \end{split}\end{equation}
    This completes the proof.
\end{proof}

\section[]{Example \ac{oracle} --- The Ho-Kalman Algorithm \cite{oymak_non-asymptotic_2019}}\label{sec:idoracle}

We now introduce an example oracle for systems with \textit{isotropic noise covariances}, $\calM = (r, n, m, A, B, C, \sigma_w^2I, \sigma_\eta^2I)$. Before going into details, we define necessary notations. Let $\calG_d(\calM)$ and $\calH_d(\calM)$ as follows for any positive integer $d$
\begin{equation}\begin{split}\label{eq:oracle_1}
    \calG_d(\calM) = {}& \begin{bmatrix}
        CB &  CAB & \cdots & CA^{2d-1}B
    \end{bmatrix}\in\bbR^{n\times 2dm},\\
    \calH_d(\calM) ={}&  \begin{bmatrix}
        CB & CAB & \cdots & CA^{d-1}B & CA^dB\\
        CAB & CA^2B & \cdots & CA^dB & CA^{d+1}B\\
        \vdots & \vdots & \ddots & \vdots & \vdots\\
        CA^{d-1}B & CA^dB & \cdots & CA^{2d-2}B & CA^{2d-1}B
    \end{bmatrix} \in\bbR^{dn\times(d+1)m}.
\end{split}\end{equation}
Moreover, let $\calH_d^-(\calM)\in\bbR^{dn\times dm}$ and $\calH_d^+(\calM)\in\bbR^{dn\times dm}$ be the first and last $dm$ columns of $\calH_d(\calM)$, respectively. 
Now we define all necessary constants for the above $\calM$
\begin{equation}\begin{split}\label{eq:oracle_2}
    {}& \phi_{\calH}(\delta) = \sigma_{\min}\b{\calH_d^-(\calM)}, \quad \psi_{\calH}(\delta) = \sigma_{1}\b{\calH_d(\calM)}, \quad d = \max\left\{r, \lceil\log\frac{1}{\delta}\rceil \right\},\\
    {}& \psi_B = \sigma_{1}(B), \quad \phi_B = \sigma_{r}(B), \quad \psi_C = \sigma_{1}(C),\quad \phi_C = \sigma_{\min}(C)
\end{split}\end{equation}


Here we assume all $\psi$'s satisfy $\psi\geq 1$, otherwise we define $\psi$ to be $\max\{1, \sigma_1(\cdot)\}$. Similarly, we assume all $\phi$'s satisfy $\phi\leq 1$, otherwise we define $\phi$ to be $\min\{1, \sigma_{\min}(\cdot)\}$.
The algorithm is summarized as follows. We note that this algorithm is not as general as we defined in \ac{oracle}. To be more specific, it only works with identity noise covariances and requires to know the dimension of the latent variables. 

\begin{algorithm}[H]
\begin{algorithmic}[1]
    \caption{Ho-Kalman}
    \label{alg:hokalman}
    \STATE {\bfseries Input:} Latent Variable Dimension $r$. Dataset $\{y_t\}_{t=0}^T, \{u_t\}_{t=0}^{T-1}$. Failure Probability $\delta$\\

    \STATE Estimating markov parameters $\wh{\calG} = \begin{bmatrix}
        \wh{CB} & \wh{CAB} & \cdots & \wh{CA^{2d-1}B}
    \end{bmatrix}$
    \begin{equation*}\begin{split}
        d \gets \max\left\{r, \lceil\log\frac{1}{\delta}\rceil \right\}, \quad \wh{\calG} \gets \mathop{\arg\min}_{\calG\in\bbR^{n\times 2dm}} \sum_{t=2d}^{T} \norm{y_t - \calG \begin{bmatrix}
            u_{t-1}\\
            u_{t-2}\\
            \vdots\\
            u_{t-2d}
        \end{bmatrix}}^2
    \end{split}\end{equation*}\\
    \STATE Constructing Hankel matrices $\wh{\calH}^-$, $\wh{\calH}^+$
    \begin{equation*}\begin{split}
        \wh{\calH}^- = \begin{bmatrix}
            \wh{CB} & \wh{CAB} & \cdots & \wh{CA^{d-1}B} \\
            \wh{CAB} & \wh{CA^2B} & \cdots & \wh{CA^dB} \\
            \vdots & \vdots & \ddots & \vdots \\
            \wh{CA^{d-1}B} & \wh{CA^dB} & \cdots & \wh{CA^{2d-2}B} 
        \end{bmatrix}, \quad \wh{\calH}^+ = \begin{bmatrix}
            \wh{CAB} & \wh{CA^2B} & \cdots & \wh{CA^dB} \\
            \wh{CA^2B} & \wh{CA^3B} & \cdots & \wh{CA^{d+1}B} \\
            \vdots & \vdots & \ddots & \vdots \\
            \wh{CA^dB} & \wh{CA^{d+1}B} & \cdots & \wh{CA^{2d-1}B} 
        \end{bmatrix}
    \end{split}\end{equation*}\\
    \STATE Truncated (rank-$r$) SVD on $\wh{\calH}^-$
    \begin{equation*}\begin{split}
        U_r, \Sigma_r, V_r \gets \text{Top }r \text{ singular vectors \& values of }\wh{\calH}^-
    \end{split}\end{equation*}
    \STATE Estimating system parameters
    \begin{equation*}\begin{split}
        \htA \gets {}& \b{U_r\Sigma_r^{1/2}}^\dagger \wh{\calH}^+ \b{\Sigma_r^{1/2}V_r\t}^\dagger,\\
        \htB \gets {}&  \text{ First $m$ columns of } \Sigma_r^{1/2}V_r\t,\\
        \htC \gets {}& \text{ First $n$ rows of $U_r\Sigma_r^{1/2}$ }
    \end{split}\end{equation*}
\end{algorithmic}
\end{algorithm}

For the paper to be self-contained, a theoretical guarantee of the above algorithm is provided. 
\begin{corollary}\label{thm:9hokalman}
    Consider $\calM = (r, n, m, A, B, C, \sigma_w^2I, \sigma_\eta^2I)$ with any $r,n,m$, independent inputs $u_t\overset{\text{i.i.d.}}{\sim} \calN(0,\sigma_u^2I)$ and notations in \Cref{eq:oracle_1,eq:oracle_2}. Suppose $\calM$ satisfy Assumption \ref{assmp:sys1} with constants $\psi_A$ and $\rho_A$. Consider single trajectory dataset $\calY=\{y_t\}_{t=0}^T, \calU=\{u_t\}_{t=0}^{T-1}$ from the system. If
    \begin{equation}\begin{split}\label{eq:hokalman_1}
        {}& T \gtrsim \underbrace{\frac{\sigma_\eta^2+(\sigma_w^2+\sigma_u^2\psi_B^2)\psi_C^2\psi_A^6/(1-\rho_A^2)^2}{\phi_\calH^2\sigma_u^2}\log^5(r(r+n+m))\log^2T\log^9\frac{1}{\delta}}_{\kappa_1(\calM,\calU,\delta) \text{ in \Cref{def:idoracle}}}\cdot r^3(r+n+m).\\
    \end{split}\end{equation}
    then with probability at least $1-\delta$, Algorithm \ref{alg:hokalman} outputs $\b{\htA, \htB, \htC}$ s.t. there exists an invertible matrix $S$
    \begin{equation}\begin{split}\label{eq:hokalman_2}
        {}& \max\left\{\norm{S^{-1}AS - \htA}, \norm{S^{-1}B - \htB}, \norm{CS - \htC}\right\}\\
            \lesssim {}& \underbrace{\frac{\psi_\calH}{\phi_\calH^2}\frac{\sigma_\eta + \psi_A^3\psi_C\sqrt{\b{\sigma_w^2 + \sigma_u^2\psi_B^2}/(1-\rho_A^2)}}{\sigma_u}\sqrt{\log^5(r(r+n+m))\log^2T\log^9\frac{1}{\delta}}}_{\kappa_2(\calM,\calU,\delta) \text{ in \Cref{def:idoracle}}}\cdot\sqrt{\frac{r^5(r+n+m)}{T}}.
    \end{split}\end{equation}
    Here $\phi_\calH=\phi_{\calH}(\delta)$, $\psi_\calH = \psi_{\calH}(\delta)$.
\end{corollary}

\begin{proof}
    Proof of this Theorem is mainly the applications of Theorems in \cite{oymak_non-asymptotic_2019}. Let $q\coloneqq r+n+m$ and recall $d = \max\left\{r, \lceil\log\frac{1}{\delta}\rceil \right\}$ in \Cref{alg:hokalman}. For simplicity, let $\calG = \calG_d(\calM)$.

    \textbf{Step 1: $\bm{\wh{\calG}}$ Estimation Guarantee. } From \Cref{eq:hokalman_1}, we know that $T \gtrsim dn/(1-\rho_A^{2d+1}) \geq n(2d+1)/(1-\rho_A^{2d+1})$ and $T \gtrsim dn\log^4(dn)\log^2(T) \gtrsim (2d+1)q\log^2((4d+2)q)\log^2(2Tq)$. We then apply Theorem 3.1 of \cite{oymak_non-asymptotic_2019} and get
    \begin{equation}\begin{split}
        \norm{\wh{\calG} - \calG} \lesssim {}& \frac{1}{\sigma_u}\b{\sigma_\eta + \sigma_w\sqrt{\norm{C\b{\sum_{i=0}^{2d-1}A^iA^i\t}C\t}} + \psi_A\norm{C}\norm{A^{2d}}\sqrt{\frac{(2d+1)\norm{\Gamma_{\infty}}}{1-\rho_A^{4d+2}}}}\\
        {}& \cdot\sqrt{\frac{(2d+1)q\log^2((4d+2)q)\log^2(2Tq)}{T}},
    \end{split}\end{equation}
    Here $\Gamma_\infty = \sigma_w^2\sum_{i=0}^{\infty}A^iA^i\t + \sigma_u^2\sum_{i=0}^{\infty}A^iBB\t A^i\t$.
    From Assumption \ref{assmp:sys1}, it is clear that 
    \begin{equation}\begin{split}
        {}& \norm{C\b{\sum_{i=0}^{2d-1}A^iA^i\t}C\t} \leq \psi_C^2\norm{\sum_{i=0}^{2d-1}A^iA^i\t} \leq \psi_C^2\sum_{i=0}^{2d-1} \norm{A^i}^2 \lesssim \frac{\psi_C^2 \psi_A^2}{1-\rho_A^2},\\
        {}& \norm{\Gamma_\infty} \leq \b{\sigma_w^2 + \sigma_u^2\psi_B^2} \sum_{i=0}^{\infty} \norm{A^i}^2 \lesssim \frac{\b{\sigma_w^2 + \sigma_u^2\psi_B^2}\psi_A^2}{1-\rho_A^2}.
    \end{split}\end{equation}
    Substituting back gives
    \begin{equation}\begin{split}
        \norm{\wh{\calG} - \calG} \lesssim {}& \frac{1}{\sigma_u}\b{\sigma_\eta + \sigma_w\sqrt{\frac{\psi_C^2 \psi_A^2}{1-\rho_A^2}} + \psi_A^2\psi_C\rho_A^{2d-1}\sqrt{\frac{d}{1-\rho_A^{4d+2}}\frac{\b{\sigma_w^2 + \sigma_u^2\psi_B^2}\psi_A^2}{1-\rho_A^2}\log\frac{24d}{\delta}}}\sqrt{\frac{dq\log^4(dq)\log^2T}{T}}\\
        \lesssim {}& \frac{1}{\sigma_u}\b{\sigma_\eta + \sigma_w\sqrt{\frac{\psi_C^2 \psi_A^2}{1-\rho_A^2}} + \psi_A^2\psi_C\sqrt{\frac{\b{\sigma_w^2 + \sigma_u^2\psi_B^2}\psi_A^2}{(1-\rho_A^2)^2}d\log\frac{d}{\delta}}}\sqrt{\frac{dq\log^4(dq)\log^2T\log\frac{1}{\delta}}{T}}\\
        \overset{(i)}{\lesssim} {}& \frac{1}{\sigma_u}\b{\sigma_\eta + \psi_A^3\psi_C\sqrt{\frac{\b{\sigma_w^2 + \sigma_u^2\psi_B^2}}{(1-\rho_A^2)^2}}}\sqrt{\frac{d^2q\log^5(dq)\log^2T\log\frac{1}{\delta}}{T}}.\\
        \lesssim {}& \frac{1}{\sigma_u} \b{\sigma_\eta + \psi_A^3\psi_C\sqrt{\frac{\b{\sigma_w^2 + \sigma_u^2\psi_B^2}}{(1-\rho_A^2)^2}}}\sqrt{\frac{r^2q\log^5(rq)\log^2T\log^8\frac{1}{\delta}}{T}}.\\
    \end{split}\end{equation}
    Here $(i)$ is because $d \lesssim r\log(\frac{1}{\delta})$ and $\log(dq) \lesssim \log(rq)\log(\frac{1}{\delta})$.
    The above inequalities hold with probability at least 
    {\small\begin{equation*}
        1 - \b{2\exp(-(2d+1)q) + 3(2Tm)^{-\log(2Tm)\log^2((4d+2)m)} + (2d+1)\b{\exp(-100(2d+1)q) + 2\exp(-100 n \log\frac{24d}{\delta})}}\footnote{The last term is because we use $\tau'=\tau\log(\frac{24d}{\delta})$ in Corollary D.3 of \cite{oymak_non-asymptotic_2019}.}
    \end{equation*}}
    We further simplify the probability as follows
    {\small\begin{equation}\begin{split}
        {}& 1 - \b{2\exp(-(2d+1)q) + 3(2Tm)^{-\log(2Tm)\log^2((4d+2)m)} + (2d+1)\b{\exp(-100(2d+1)q) + 2\exp(-100 n\log\frac{24d}{\delta})}}\\
        \overset{(i)}{\geq} {}& 1 - \frac{\delta}{4} - \b{2\exp(-2dq) + (2d+1)\b{\exp(-100dq) + 2\exp(-100 n\log\frac{24d}{\delta})}}\\
        \overset{(ii)}{\geq} {}& 1 - \frac{\delta}{4} - \frac{\delta}{4} - \b{(2d+1)\b{\exp(-100dq) + 2\exp(-100 n\log\frac{24d}{\delta})}}\\
        \overset{(iii)}{\geq} {}& 1 - \frac{\delta}{4} - \frac{\delta}{4} - \frac{\delta}{4} - (4d+2)\exp\b{-100 n\log\frac{24d}{\delta}}\\
        \geq {}& 1 - \frac{\delta}{4} - \frac{\delta}{4} - \frac{\delta}{4} - \frac{4d+2}{24d}\delta\\
        \geq {}& 1 - \delta. 
    \end{split}\end{equation}}
    Here $(i)$ is because $3(2Tm)^{-\log(2Tm)\log^2(4dm)} \leq 3(2\cdot \frac{1}{\delta})^{-7} \leq \frac{\delta}{4}$, where the first inequality is because $\log^2(4dm)\geq\log^2(16)$. $(ii)$ is because $d \geq \log\frac{1}{\delta}$ and $2\exp(-2dq) \leq 2\exp(-6\log(1/\delta)) \leq 2\delta^6 \leq \delta/4$ for $\delta \in(0,e^{-1})$ and $(iii)$ is because $(2d+1)\exp(-100dq) \leq (2d+1)\exp(-2d-1)\exp(-90dq) \leq \exp(-270d) \leq \delta/4$ for $\delta\in(0,e^{-1})$. 

    \textbf{Step 2: Parameter Recovery Guarantee. } 
    From the above result, we get the following for some constant $c_1$
    \begin{equation}\begin{split}
        \sqrt{d}\norm{\wh{\calG} - \calG} \leq \frac{c_1}{\sigma_u}\b{\sigma_\eta + \psi_A^3\psi_C\sqrt{\frac{\b{\sigma_w^2 + \sigma_u^2\psi_B^2}}{(1-\rho_A^2)^2}}}\sqrt{\frac{r^2q\log^5(rq)\log^2T\log^8\frac{1}{\delta}}{T}}\cdot\sqrt{r\log\frac{1}{\delta}} \leq \frac{\phi_\calH}{4}.
    \end{split}\end{equation}
    Here the last inequality is because of Equation \ref{eq:hokalman_1}. Then we apply Corollary 5.4\footnote{The following inequality can be obtained by substituting all $r\norm{L-\wh{L}}$ by $\frac{r^2\norm{L-\wh{L}}^2}{\sigma_{\min}(L)}$ in the original result. This is valid because the original proof used $\frac{2}{\sqrt{2}-1}\frac{\norm{L-\wh{L}}_F^2}{\sigma_{\min}(L)} \leq 5r\norm{L-\wh{L}}$ in Lemma B.1 of \cite{oymak_non-asymptotic_2019}} from \cite{oymak_non-asymptotic_2019} and get the following for some balanced realization $\bar{A}, \bar{B}, \bar{C}$ and some unitary matrix $U$
    \begin{equation}\begin{split}
        {}& \max\left\{\norm{\bar{A} - U\t \htA U}, \norm{\bar{B} - U\t \htB}, \norm{\bar{C} - \htC U}\right\}\\
        \lesssim {}& \frac{\psi_\calH}{\phi_\calH^2}\sqrt{r^2d}\norm{\wh{\calG}-\calG}\\
        \lesssim {}& \frac{\psi_\calH}{\phi_\calH^2}\frac{\sigma_\eta + \psi_A^3\psi_C\sqrt{\b{\sigma_w^2 + \sigma_u^2\psi_B^2}/(1-\rho_A^2)^2}}{\sigma_u}\sqrt{\frac{r^5q\log^5(rq)\log^2T\log^9\frac{1}{\delta}}{T}}.
    \end{split}\end{equation}
    Since $\barA, \barB, \barC$ is a balanced realization, there exist some invertible matrix $S$ s.t.
    \begin{equation}\begin{split}
        {}& \max\left\{\norm{S^{-1}AS - \htA}, \norm{S^{-1}B - \htB}, \norm{CS - \htC}\right\}\\
        \lesssim {}& \frac{\psi_\calH}{\phi_\calH^2}\frac{\sigma_\eta + \psi_A^3\psi_C\sqrt{\b{\sigma_w^2 + \sigma_u^2\psi_B^2}/(1-\rho_A^2)^2}}{\sigma_u}\sqrt{\frac{r^5q\log^5(rq)\log^2T\log^9\frac{1}{\delta}}{T}}.
    \end{split}\end{equation}
\end{proof}

\subsection{Upper Bounds for \ac{col-algo} (\Cref{alg:1single}) with Oracle \Cref{alg:hokalman}}
\begin{corollary}[\Cref{cor:e2e} Restated]\label{cor:e2e_full}
    Consider $\calM = (r, n, m, A, B, C, \sigma_w^2I, \sigma_\eta^2I)$ and datasets $\calD_1=\calU_1\cup\calY_1,\calD_2=\calU_2\cup\calY_2$ (with length $T_1, T_2$ respectively) where the inputs are sampled independently from $\calN(0,\sigma_u^2I)$. Consider constants defined for $\calM$ in \Cref{sec:hdsysid} and in \Cref{eq:oracle_2}. Suppose $\calM$ satisfies Assumption \ref{assmp:sys1} with constants $\psi_A$ and $\rho_A$.
    If
    \begin{equation}\begin{split}\label{eq:e2e_1}
        &T_1 \gtrsim \tilde{\kappa}_3\cdot n^2r^3, \quad T_2 \gtrsim \tilde{\kappa}_1\cdot r^3(r+m),
    \end{split}\end{equation}
    then $(\htA,\htB,\htC)$ from Algorithm \ref{alg:1single} satisfy the following for some invertible matrix $S$ with probability at least $1-\delta$
    \begin{equation}\begin{split}
        {}& \max\left\{\norm{S^{-1}AS - \htA}, \norm{S^{-1}B - \htB}, \norm{CS - \htC}\right\}\lesssim \tilde{\kappa}_4 \cdot \sqrt{\frac{n}{T_1}}\norm{\htC} + \tilde{\kappa}_2\cdot\sqrt{\frac{r^5(r+m)}{T_2}}.
    \end{split}\end{equation}
    $\kappa_{[4]}$ are detailed below. All of them are problem-related constants independent of system dimensions modulo log factors.
    \begin{equation}\begin{split}
        \tilde{\kappa}_1 = {}& \frac{\psi_\eta+\psi_w\psi_C^2\psi_A^6/(1-\rho_A^2)^2}{\phi_\calH^2\sigma_u^2}\log^5(r(r+m))\log^2T\log^9\frac{1}{\delta},\\
        \tilde{\kappa}_2 ={}&  \frac{\psi_\calH}{\phi_\calH^2}\frac{\sqrt{\psi_\eta} + \psi_A^3\psi_C\sqrt{\psi_w^2/(1-\rho_A^2)}}{\sigma_u}\sqrt{\log^5(r(r+m))\log^2T\log^9\frac{1}{\delta}},\\
        \tilde{\kappa}_3 = {}& \frac{1}{\phi_\calH^2}\max\left\{\tilde{\kappa}_4^2\frac{\psi_A^4\psi_B^2\psi_C^2}{(1-\rho_A^2)\phi_O^2}, ~\b{\frac{\psi_w\psi_\eta\psi_C^2\psi_A^2}{\phi_u\phi_C^2\phi_R^2(1-\rho_A^2)}}^4\log^4(r)\log^8\b{\frac{1}{\delta}}\right\},\\
        \tilde{\kappa}_4 = {}& \frac{\psi_\eta}{\phi_u\phi_C^2\phi_R^2}\sqrt{\log\frac{1}{\delta}}.
    \end{split}\end{equation}
    Here $\phi_\calH=\phi_\calH(\delta)$ and $\psi_\calH=\psi_\calH(\delta)$ are defined in \Cref{eq:oracle_2}.
\end{corollary}

\begin{proof}
    From \Cref{eq:e2e_1}, $T_1 \gtrsim \tilde{\kappa}_3 n^2r^3$, which satisfies the condition of \Cref{lem:1col_full}. We apply \Cref{lem:1col_full} on $(\calD_1, \Sigma_\eta)$ and get the following with probability at least $1-\frac{\delta}{2}$ for some constant $c_1$
    \begin{equation}
        \norm{\wh{\Phi}_C^\perp\t \Phi_C} \leq c_1 \tilde{\kappa}_4 \sqrt{\frac{n}{T_1}} \leq \frac{\phi_\calH(1-\rho_A^2)}{8\psi_A^2\psi_B\psi_C}.
    \end{equation}
    Then from \Cref{lem:9hokalman_1} we know $\kappa_1\b{\wh{\calM}, \calU_2,\delta}
    \leq 4\kappa_1\b{\calM, \calU_2,\delta}, \quad \kappa_2\b{\wh{\calM}, \calU_2,\delta} \leq 
    4\kappa_2\b{\calM, \calU_2,\delta}$.
    Here $\wh{\calM}=(r,\rank(\wh{\Phi}_C\t C),m,A,B,\wh{\Phi}_C\t C,\sigma_w^2I, \sigma_\eta^2I)$.
    Moreover, from their definitions, we know that $\kappa_1(\calM,\calU_2,\delta) \leq \tilde{\kappa}_1$ and $\kappa_2(\calM,\calU_2,\delta) \leq \tilde{\kappa}_2$.
    Combining with \Cref{eq:e2e_1} gives
    \begin{equation}\begin{split}
        T_2 \gtrsim {}& \tilde{\kappa}_1r^3(r+m) \geq \kappa_1\b{\calM, \calU_2,\delta} r^3(r+m)\\
        \gtrsim {}& \kappa_1\b{\wh{\calM}, \calU_2,\delta} r^3(r+\rank(\wh{\Phi}_C\t C)+m).
    \end{split}\end{equation} 
    From the proof of \Cref{thm:1single} (step 1), we know that we can apply \Cref{alg:hokalman} on $\wh{\calM}$. Applying \Cref{alg:hokalman} gives the following with probability at least $1-\frac{\delta}{2}$ from \Cref{thm:9hokalman}
    \begin{equation}\begin{split}
        {}& \max\left\{\norm{S^{-1}AS - \htA}, \norm{S^{-1}B - \htB}, \norm{\wh{\Phi}_C\t CS - \tilC}\right\}\\
        \lesssim {}& \kappa_2(\wh{\calM},\calU_2,\delta)\cdot\sqrt{\frac{r^5(r+\rank(\wh{\Phi}_C\t C)+m)}{T_2}}\\
        \lesssim {}& \kappa_2(\calM,\calU_2,\delta)\cdot\sqrt{\frac{r^5(r+m)}{T_2}}\\
        \lesssim {}& \tilde{\kappa}_2\cdot \sqrt{\frac{r^5(r+m)}{T_2}}\\
    \end{split}\end{equation}
    Then from the proof of \Cref{thm:1single} (step 2), we know that 
    \begin{equation}\begin{split}
        {}& \max\left\{\norm{S^{-1}AS - \htA}, \norm{S^{-1}B - \htB}, \norm{CS - \htC}\right\}\\
        \lesssim {}& \tilde{\kappa}_2\cdot\sqrt{\frac{r^5(r+m)}{T_2}} + \tilde{\kappa}_4 \cdot \sqrt{\frac{n}{T_1}}\norm{\htC}.
    \end{split}\end{equation}
\end{proof}

\subsection{Other Lemmas for $\kappa_1$ and $\kappa_2$}

\begin{lemma}\label{lem:9hokalman_1}
    Fix any $\delta$. Consider system $\calM=(r,n,m,A,B,C,\sigma_w^2I,\sigma_\eta^2)$ and independent inputs $\calU=\{u_t\}_{t=0}^{T-1}$ with $u_t\sim\calN(0,\sigma_u^2I)$. Let $\wh{\calM}=(r,\rank(\wh{\Phi}_C\t C),m,A,B,\wh{\Phi}_C\t C,\sigma_w^2I, \sigma_\eta^2I)$, where $\wh{\Phi}_C$ is a orthonormal matrix satisfying
    \begin{equation}
        \norm{\b{\wh{\Phi}_C^\perp}\t \Phi_C} \leq \Delta_{\Phi} \leq  \frac{\phi_{\calH}(1-\rho_A^2)}{8\psi_A^2\psi_B\psi_C}.
    \end{equation}
    Let $\kappa_1(\cdot, \calU,\delta), \kappa_2(\cdot, \calU,\delta)$ be defined as in \Cref{eq:hokalman_1,eq:hokalman_2}. Then we have
    \begin{equation}\begin{split}
        \kappa_1\b{\wh{\calM}, \calU,\delta}
        \leq 4\kappa_1\b{\calM, \calU,\delta}, \quad \kappa_2\b{\wh{\calM}, \calU,\delta} \leq 
        4\kappa_2\b{\calM, \calU,\delta}
    \end{split}\end{equation}
\end{lemma}

\begin{proof}
    For system $\calM$, we list all related parameters as follows\footnote{Here we omit $\sigma_u$ since it doesn't change between the two systems.}.
    \begin{equation}
        \sigma_w^2, \sigma_\eta^2, \rho_A, \psi_A, \psi_B, \psi_C, \phi_{\calH}, \psi_{\calH}. 
    \end{equation}
    For the projected system $\wh{\calM}$, we denote corresponding parameters as     \begin{equation}
        \wh{\sigma}_w^2, \wh{\sigma}_\eta^2, \wh{\rho}_A, \wh{\psi}_A, \wh{\psi}_B, \wh{\psi}_C, \wh{\phi}_{\calH}, \wh{\psi}_{\calH}. 
    \end{equation}
    It is clear that $\wh{\sigma}_w^2, \wh{\sigma}_\eta^2, \wh{\rho}_A, \wh{\psi}_A, \wh{\psi}_B$ remain unchanged. We only need to consider parameters $\wh{\psi}_C, \wh{\phi}_{\calH}, \wh{\psi}_{\calH}$. 
    We know that
    \begin{equation}\begin{split}
        \wh{\psi}_C = {}& \norm{\wh{\Phi}_C\t C} \leq \norm{\wh{\Phi}_C} \norm{C} = \psi_C.\\
    \end{split}\end{equation}
    Moreover, for $d=\max\{r,\lceil \frac{1}{\delta}\rceil\}$, we have
    \begin{equation}\begin{split}
        \norm{\calH_d^+(\wh{\calM})} = \norm{\diag\b{{\wh{\Phi}_C\t, \dots, \wh{\Phi}_C\t}}\calH_d^+(\calM)} \leq \norm{\wh{\Phi}_C}\norm{\calH_d^+(\calM)} = \psi_{\calH}.
    \end{split}\end{equation}
    Therefore, $\wh{\psi}_{\calH} \leq \psi_{\calH}$. 
    For $\wh{\phi}_{\calH}$, we know from \Cref{lem:hokalman_1} that $\wh{\phi}_{\calH} \geq \phi_{\calH}/2$ for $\Delta_{\Phi} \leq  \frac{\phi_{\calH}(1-\rho_A^2)}{8\psi_A^2\psi_B\psi_C}$. 
    
    These bounds imply that
    \begin{equation}\begin{split}
    \kappa_2\b{\calM, \calU,\delta} 
    \geq {}& \frac{1}{4}\kappa_2\b{\wh{\calM}, \calU,\delta}, \quad 
    \kappa_1\b{\calM, \calU,\delta} 
    \geq \frac{1}{4}\kappa_1\b{\wh{\calM}, \calU,\delta}.
    \end{split}\end{equation}
\end{proof}

\begin{lemma}\label{lem:hokalman_1}
    Consider the same setting as \Cref{lem:9hokalman_1}. Consider constants $\phi_{\calH}$ and $\wh{\phi}_{\calH}$ of $\calM$ and $\wh{\calM}$. 
    Then if $\Delta_\Phi \leq  \frac{\phi_{\calH}(1-\rho_A^2)}{8\psi_A^2\psi_B\psi_C}$, then $\wh{\phi}_{\calH} \geq \frac{\phi_{\calH}}{2}$.
\end{lemma}

\begin{proof}
    Recall that $d = $ From the definition of $\wh{\phi}_\calH$, we know that 
    \begin{equation}\begin{split}
        {}& \sigma_{\min}\b{\diag\b{{\wh{\Phi}_C\t, \dots \wh{\Phi}_C\t}}\calH_d^-(\calM)} = \sigma_{\min}\b{\diag\b{\wh{\Phi}_C\wh{\Phi}_C\t, \dots, \wh{\Phi}_C\wh{\Phi}_C\t}\calH_d^-(\calM)}\\
        = {}& \sigma_{\min}\b{\calH_d^-(\calM) - \diag\b{\wh{\Phi}_C^\perp\wh{\Phi}_C^\perp\t, \dots, \wh{\Phi}_C^\perp\wh{\Phi}_C^\perp\t} \calH_d^-(\calM)}\\
        \geq {}& \sigma_{\min}\b{\calH_d^-(\calM)} - \sigma_{1}\b{\diag\b{\wh{\Phi}_C^\perp\wh{\Phi}_C^\perp\t, \dots, \wh{\Phi}_C^\perp\wh{\Phi}_C^\perp\t} \calH_d^-(\calM)}\\
    \end{split}\end{equation}
    Notice that 
    \begin{equation}\begin{split}
        \diag\b{\wh{\Phi}_C^\perp\wh{\Phi}_C^\perp\t, \dots, \wh{\Phi}_C^\perp\wh{\Phi}_C^\perp\t} \calH_d^-(\calM) = \begin{bmatrix}
            \wh{\Phi}_C^\perp\wh{\Phi}_C^\perp\t C\\
            \wh{\Phi}_C^\perp\wh{\Phi}_C^\perp\t CA\\
            \vdots\\
            \wh{\Phi}_C^\perp\wh{\Phi}_C^\perp\t CA^{d-1}
        \end{bmatrix} \begin{bmatrix}
            B & AB & \cdots & A^{d-1}B
        \end{bmatrix}
    \end{split}\end{equation}
    From Equation \ref{eq:1_col_full_2}, we directly know the following 
    \begin{equation}\begin{split}
        \norm{\begin{bmatrix}
            \wh{\Phi}_C^\perp\wh{\Phi}_C^\perp\t C\\
            \wh{\Phi}_C^\perp\wh{\Phi}_C^\perp\t CA\\
            \vdots\\
            \wh{\Phi}_C^\perp\wh{\Phi}_C^\perp\t CA^{d-1}
        \end{bmatrix}}
        \leq {}& \frac{2\psi_A\psi_C}{\sqrt{1-\rho_A^2}}\Delta_\Phi.
    \end{split}\end{equation}
    Moreover, 
    \begin{equation}
        \norm{\begin{bmatrix}
            B & AB & \cdots A^{d-1}B
        \end{bmatrix}} \leq \sqrt{\sum_{i=0}^{d-1}\norm{A^iBB\t A^i\t}} \leq \frac{2\psi_A\psi_B }{\sqrt{1-\rho_A^2}}.
    \end{equation}
    Substituting back gives
    \begin{equation}
        \wh{\phi}_{\calH} \geq \phi_{\calH} - 4\frac{\psi_A^2\psi_B\psi_C}{(1-\rho_A^2)}\Delta_{\Phi}. 
    \end{equation}
    Since $\Delta_\Phi \leq \frac{\phi_{\calH}(1-\rho_A^2)}{8\psi_A^2\psi_B\psi_C}$, we know $\wh{\phi}_{\calH} \geq \phi_{\calH}/2$.
\end{proof}

\end{document}